\documentclass[11pt]{article}
\usepackage[T1]{fontenc}
\usepackage{array,longtable}
\usepackage{xcolor}
\usepackage{fancyhdr}
\usepackage{graphicx}
\usepackage{float}
\usepackage{algorithmic}
\usepackage[american]{babel}
\usepackage{csquotes}
\usepackage{enumitem} 
\usepackage{subcaption} 
\usepackage{multirow, booktabs}
\usepackage{setspace} 
\usepackage{algorithm}
\usepackage{algorithmic}

\usepackage[backend=biber,citestyle=authoryear,sortcites=true,natbib,sorting=nyt,style=apa]{biblatex}
\DeclareLanguageMapping{american}{american-apa}
\bibliography{bibtex-MonoSplines.bib}
\usepackage[
  top=2.5cm, bottom=2.5cm,
  left=2.6cm, right=2.6cm,
  headsep=3mm,
  headheight=15pt]{geometry}
\usepackage{newpxtext}
\usepackage{tikz}
\usetikzlibrary{shapes.geometric}
\usetikzlibrary{through} 
\usepackage{hyperref}
\hypersetup{
    colorlinks=true,
    linkcolor=blue,
    filecolor=magenta,      
    urlcolor=cyan,
    }
\usepackage{amsmath, amssymb}
\usepackage{ifthen}

\DeclareMathOperator*{\argmin}{arg\,min}
\DeclareMathOperator*{\tr}{tr}

\DeclareMathOperator*{\Var}{Var}
\DeclareMathOperator*{\MSE}{MSE}
\DeclareMathOperator*{\rank}{rank}

\newcommand\IR{\mathrm{I\!R}}
\newcommand\subto{\mathrm{s.t.}}

\newcommand\cE{{\mathcal E}}

\newcommand\cG{{\mathcal G}}

\newcommand\cJ{{\mathcal J}}

\newcommand\cL{{\mathcal L}}

\newcommand\cQ{{\mathcal Q}}

\newcommand\cS{{\mathcal S}}

\newcommand\cX{{\mathcal X}}

\newcommand\bA{{\mathbf A}}
\newcommand\bB{{\mathbf B}}

\newcommand\bD{{\mathbf D}}

\newcommand\bG{{\mathbf G}}
\newcommand\bH{{\mathbf H}}
\newcommand\bI{{\mathbf I}}

\newcommand\bS{{\mathbf S}}

\newcommand\bZ{{\mathbf Z}}

\newcommand\bfA{{\mathbf A}}
\newcommand\bfB{{\mathbf B}}
\newcommand\bfC{{\mathbf C}}
\newcommand\bfD{{\mathbf D}}

\newcommand\bfG{{\mathbf G}}
\newcommand\bfH{{\mathbf H}}
\newcommand\bfI{{\mathbf I}}

\newcommand\bfL{{\mathbf L}}

\newcommand\bfP{{\mathbf P}}

\newcommand\bfb{{\mathbf b}}

\newcommand\bfe{{\mathbf e}}
\newcommand\bff{{\mathbf f}}
\newcommand\bfg{{\mathbf g}}

\newcommand\bfw{{\mathbf w}}
\newcommand\bfx{{\mathbf x}}
\newcommand\bfy{{\mathbf y}}

\newcommand\bbE{{\mathbb E}}

\newcommand{\bOmega}{\boldsymbol{\Omega}}
\newcommand{\bSigma}{\boldsymbol{\Sigma}}
\newcommand{\bfOmega}{\boldsymbol{\Omega}}
\newcommand{\one}{\boldsymbol{1}}
\newcommand\zero{\boldsymbol{0}}

\newcommand\Bias{\mathrm{Bias}}

\newcommand\ls{\mathrm{ls}}

\newcommand\df{\mathrm{df}}

\newcommand\CI{\mathrm{CI}}
\newcommand\CB{\mathrm{CB}}
\newcommand\iso{\mathrm{iso}}

\newcommand{\supp}{\href{bka-supp.pdf}{Supplementary Material}}

\usepackage{amsthm}
\ifthenelse{\isundefined{\definition}}{\newtheorem{definition}{Definition}}{}
\ifthenelse{\isundefined{\fact}}{\newtheorem{fact}{Fact}}{}
\ifthenelse{\isundefined{\theorem}}{\newtheorem{theorem}{Theorem}}{}
\ifthenelse{\isundefined{\proposition}}{\newtheorem{proposition}{Proposition}}{}
\ifthenelse{\isundefined{\corollary}}{\newtheorem{corollary}{Corollary}}{}
\ifthenelse{\isundefined{\lemma}}{\newtheorem{lemma}{Lemma}}{}
\ifthenelse{\isundefined{\remark}}{\newtheorem{remark}{Remark}}{}
\ifthenelse{\isundefined{\assumption}}{\newtheorem{assumption}{Assumption}}{}

\usepackage[affil-it]{authblk}
\title{Monotone Cubic B-Splines with a Neural-Network Generator}
\author[1,2]{Lijun Wang%
  \thanks{Lijun Wang was a doctoral student at The Chinese University of Hong Kong and is now a postdoctoral associate at Yale University. Email: \texttt{ljwang@link.cuhk.edu.hk}, \texttt{lijun.wang@yale.edu}}
}
\affil[1]{Department of Statistics, The Chinese University of Hong Kong, Hong Kong SAR, China}
\affil[2]{Department of Biostatistics, Yale University, New Haven, Connecticut, USA}
\author[1]{Xiaodan Fan%
  \thanks{Email: \texttt{xfan@cuhk.edu.hk}}
}
\author[3]{Huabai Li%
    \thanks{Email: \texttt{hbli@cuhk.edu.hk}}
}
\affil[3]{Department of Physics, The Chinese University of Hong Kong, Hong Kong SAR, China}
\author[4]{Jun S. Liu%
\thanks{Email: \texttt{jliu@stat.harvard.edu}}
}
\affil[4]{Department of Statistics, Harvard University, Cambridge, Massachusetts, USA}
\date{}
\begin{document}
\maketitle

\begin{abstract}
We present a method for fitting monotone curves using cubic B-splines, which is equivalent to putting a
monotonicity constraint on the coefficients. 
We explore different ways of enforcing this constraint and analyze their theoretical and empirical properties. We propose two algorithms for solving the spline fitting problem: one that uses standard optimization techniques and one that trains a Multi-Layer Perceptrons (MLP) generator to approximate the solutions under various settings and perturbations. The generator approach can speed up the fitting process when we need to solve the problem repeatedly, such as when constructing confidence bands using bootstrap. We evaluate our method against several existing methods, some of which do not use the monotonicity constraint, on some monotone curves with varying noise levels. We demonstrate that our method outperforms the other methods, especially in high-noise scenarios. We also apply our method to analyze the polarization-hole phenomenon during star formation in astrophysics. The source code is accessible at \texttt{\url{https://github.com/szcf-weiya/MonotoneSplines.jl}}.
\end{abstract}


\noindent%
\vfill

\newpage

\section{Introduction}

Monotonicity or other shape constraints are commonly seen in many applications, such as monotonic patterns of growth curves in biology and ecology \parencite{kahmGrofitFittingBiological2010}, shapes of certain economic instruments during certain periods \parencite{pattonMonotonicityAssetReturns2010}, dose response functions in medicine, 
and curves related to the item response theory (IRT) in psychometrics \parencite{embretsonItemResponseTheory2013}. 
Various monotone fitting approaches have been proposed to handle such types of data. \textcite{ramsayMonotoneRegressionSplines1988} introduced integrated splines (I-splines),
and constructed monotone splines with non-negative coefficients on the I-splines. \textcite{meyerInferenceUsingShaperestricted2008} recommended using quadratic I-splines 
because a linear combination of the piecewise quadratic I-splines is non-decreasing if and only if their coefficients are non-negative. Similar to I-splines, \textcite{murrayFastFlexibleMethods2016} presented an integral parameterization for monotone polynomials.
\textcite{heMonotoneBsplineSmoothing1998} proposed a monotone smoothing method by minimizing the $L_1$ loss in the space of quadratic B-splines subject to the nonnegative (or nonpositive) first derivative constraint. Since the first derivative of quadratic B-splines is linear, the problem can be solved by linear programming.
The $L_1$ loss is a special median case of the loss for quantile functions, and the algorithm has been summarized in \textcite{heCOBSQualitativelyConstrained1999}, which is later updated by \textcite{ngFastEfficientImplementation2007} with available R package \texttt{COBS}.

Monotonicity is just one typical shape constraint, and there are several other typical shape constraints, such as the sign and the curvature. Researchers have developed a general workflow for general shape-constrained problems. \textcite{turlachShapeConstrainedSmoothing2005} proposed an iterative procedure: firstly, fit an unconstrained smoothing, then verify if the fit satisfies the shape constraints. If not, identify violations of the shape constraints, and add new constraints for the violations to refit. The procedure is iterated until all shape constraints are fulfilled. \textcite{pappOptimizationModelsShapeconstrained2011} and \textcite{pappShapeConstrainedEstimationUsing2014} characterized the monotonicity and curvature conditions with Bernstein polynomials and solved with a conic optimization approach. \textcite{navarro-garciaConstrainedSmoothingOutofrange2023} formulated the (sign, monotonicity, or curvature) constrained smoothing via the non-negative penalized splines approach based on a necessary and sufficient condition for non-negative univariate polynomials. Particularly, for monotone smoothing, they imposed such a condition on the first derivative of splines.

Another well-known approach for preserving monotonicity is the isotonic regression \parencite{barlowIsotonicRegressionProblem1972}. However, the isotonic regressions always under-smooth the data. To fulfill the smoothing requirement, \textcite{mammenEstimatingSmoothMonotone1991} proposed to conduct a smoothing step before (or after) the isotonisation step for isotonic regressions. Recently, \textcite{groeneboomConfidenceIntervalsMonotone2023} proposed to construct consistent bootstrap confidence intervals using the smoothed isotonic (i.e., smoothing after isotonisation) estimator, since the bootstrap based on the ordinary isotonic estimator is inconsistent.

Recently, neural network-based deep learning algorithms have been successfully applied to problems with complex patterns or structures, such as image and video classifications, speech recognition, and text modeling \parencite{jamesIntroductionStatisticalLearning2021}. There are also some researches on imposing the monotonicity constraint on neural networks. \textcite{zhangFeedforwardNetworksMonotone1999} proposed a monotone Multi-Layer Perceptron (MLP) network by replacing the weights $w_i$ between different layers with $e^{w_i}$. An implementation of monotone MLP based on \textcite{zhangFeedforwardNetworksMonotone1999} can be found in \textcite{cannonMonmlpMultilayerPerceptron2017}'s R package \texttt{monmlp}. \textcite{langMonotonicMultilayerPerceptron2005} used a similar idea but considered the hyperbolic tangent activation function and assumed positive weights between different layers. \textcite{mininComparisonUniversalApproximators2010} proposed a min-max neural network and constrained the weights to be non-negative to obtain a monotone model.

Splines are powerful tools for local polynomial representations, among which
the cubic spline is the most popular one. Some researchers even claim that cubic spline is the lowest-order spline for which the knot-discontinuity is not visible to human eyes, and there is scarcely any good reason to go beyond cubic splines \parencite{hastieElementsStatisticalLearning2009}.  Monotone quadratic splines proposed by \textcite{heMonotoneBsplineSmoothing1998}  do not have second derivatives at the knots, so that a commonly used measure of smoothness (and penalty)  in smoothing splines cannot be defined. Curiously, there has not been much literature on monotone fitting using cubic splines. To fill this gap, we here propose {\it monotone cubic B-splines} and provide two approaches for fitting them: one based on existing optimization toolboxes and another achieved by our proposed MLP generator, which takes advantage of the power and flexibility of neural networks. The MLP generator can be further extended to estimate the confidence band efficiently.

This article is organized as follows. Section \ref{sec:monobspl_method} elaborates the proposed monotone cubic B-splines by comparing the fitting errors under different monotonicity conditions (Section~\ref{sec:cond_mono}), giving an explicit form to the solution (Section~\ref{sec:sol}), and discussing the selection of tuning parameters(Section~\ref{sec:monobspl_paras}). Section \ref{sec:two_alg} presents two algorithms for fitting the monotone splines: existing optimization toolboxes and our proposed Multi-Layer Perceptrons (MLP) generator. The MLP generator can be further extended to estimate the confidence band efficiently in Section \ref{sec:conf_band}. Extensive simulations for comparing the monotone splines with other monotone fitting techniques are given in Section \ref{sec:monobspl_sim}. We also apply our monotone splines on an astrophysics project to explore the mystery of star formation in Section~\ref{sec:app}. Limitations and future work are discussed in Section \ref{sec:monobspl_discuss}.



\section{Monotone Cubic B-spline}\label{sec:monobspl_method}

\subsection{Preliminary}\label{sec:monobspl_pre}
An order-$M$ spline with $K$  ordered knots at $\xi_1<\xi_2<\cdots<\xi_K$
can be represented by a linear combination of $K+M$ bases: $f(x)=\sum_{i=1}^{K+M}\gamma_i h_i(x)$, where the set of functions,
$\{h_i(x)\}_{i=1}^{K+M}$, are called {\it bases}.
Although there are many equivalent bases for representing spline functions, the B-spline basis system, which has been discussed in detail in \textcite{deboorPracticalGuideSplines1978}, is attractive numerically \parencite{ramsayFunctionalDataAnalysis2005}. 

The order-$M$ B-spline basis can be defined through a lower order B-spline basis recursively. Let $B_{i,m}(x)$ be the $i$-th B-spline basis function of order $m\in [1:M]$. Let $\xi_0 <\xi_1$ and $\xi_{K+1}>\xi_K$ be two boundary knots. Augment the knot sequence $\{\xi_\ell\}_{\ell=1}^K$ to $\{\tau_i\}_{i=1}^{K+2M}$ by extending two boundary knots:
\begin{equation}
\begin{split}
    &\tau_1\le \tau_2\le\cdots\le\tau_M\le\xi_0\,,\\
    &\tau_{M+1} = \xi_1 < \tau_{M+2} = \xi_2 < \cdots < \tau_{K+M} =\xi_K\,,\\
    &\xi_{K+1}\le \tau_{K+M+1}\le \tau_{K+M+2}\le\cdots\le\tau_{K+2M}\,.
    \end{split}
    \label{eq:def_aug}
\end{equation}
B-spline basis functions are recursively defined as follows,
\begin{align*}
    B_{i, 1}(x) & = \begin{cases}
    1 & \text{ if } \tau_i\le x < \tau_{i+1}\\
    0 & \text{otherwise}
    \end{cases}
    \quad\,,i=1,\ldots,K+2M-1\\
    B_{i, m}(x) &= \frac{x-\tau_i}{\tau_{i+m-1} -\tau_i}B_{i,m-1}(x) + \frac{\tau_{i+m}-x}{\tau_{i+m}-\tau_{i+1}}B_{i+1,m-1}(x)\,,
    \quad i=1,\ldots,K+2M-m\,.
\end{align*}

Given $n$ paired points $\{x_i, y_i\}_{i=1}^n$, spline fitting aims to find some function $f$ by minimizing
\begin{equation}
\sum_{i=1}^n(y_i-f(x_i))^2 + \lambda\int \{f''(t)\}^2dt\,,
\label{eq:def_spline}
\end{equation}
where $\lambda \ge 0$ is the penalty parameter to discourage the roughness. Write $f(x)$ as a cubic B-spline,
$f(x) = \sum_{j=1}^J\gamma_jB_j(x)\,,$
where $J=K+M=K+4$ is the number of basis functions, $B_j$'s are the basis functions, and $\gamma_j$'s are the coefficients. Denote $\bfy = [y_1,\ldots,y_n]^T, \gamma =[\gamma_1,\ldots,\gamma_J]^T.$ Let $\bB$ be a $n\times J$ matrix with entries $\bB_{ij}=B_j(x_i),i=1,\ldots,n;j=1,\ldots,J$. Then we can write $f(x_i) = \bfb^T_{i}\gamma$, where $\bfb_i$ is the $i$-th row vector of $\bB$. Note that
$f''(t) = \sum_{j=1}^J\gamma_jB_j''(t)\,,$
then 
$$
\int [f''(t)]^2dt = \int \sum_{j=1}^J\sum_{k=1}^J\gamma_j\gamma_kB_j''(t)B_k''(t)dt = \sum_{j=1}^J\sum_{k=1}^J\gamma_j\gamma_k\int B_j''(t)B_k''(t)dt = \gamma^T\bOmega\gamma\,,
$$
where $\{\bfOmega\}_{jk}=\int B_j''(s)B_k''(s)ds$ is called the roughness penalty matrix. Now Problem \eqref{eq:def_spline} can be expressed in a matrix form,
\begin{equation}
  \min_{\gamma}\, (\bfy - \bfB\gamma)^T(\bfy - \bfB\gamma) + \lambda \gamma^T\bfOmega\gamma\,.    \label{eq:matrss_cubic_smooth_spline}
\end{equation}
 The solution turns out to be
\begin{equation}
\hat\gamma
= (\bfB^T\bfB+\lambda\bfOmega)^{-1}\bfB^T\bfy\,.    
\label{eq:sol_ridge}
\end{equation}
If $\lambda = 0$, the spline is referred to as a \emph{cubic spline}, and it is called a \emph{smoothing spline} when $\lambda > 0$.

\subsection{Encoding the Monotonicity Constraint}\label{sec:cond_mono}

For quadratic B-splines, the nonnegative (or nonpositive) first derivative constraint can be encoded as a set of linear inequality constraints on the knots (see the proof of Proposition~\ref{prop:conditions} in the \supp). For cubic B-splines, however, such simple  linear constraints at the knots are no longer sufficient to ensure monotonicity. Proposition~\ref{prop:conditions} below describes a set of computing-friendly constraints for a cubic B-spline to be monotone.

\begin{proposition}
Let $\xi_0<\xi_1<\ldots<\xi_{K} < \xi_{K+1}$ be the knots of cubic B-spline basis functions $B_{j,4}(x),j=1,\ldots,J=K+4$, and let $\{\tau_i\}_{i=1}^{K+8}$ be the augmented knots as defined in \eqref{eq:def_aug}.
To ensure a cubic spline function $f(x)=\sum_{j=1}^J\gamma_jB_{j,4}(x)$ to be non-decreasing in $x\in [\xi_0, \xi_{K+1}]$:
  \begin{itemize}
      \item A sufficient condition is that $\gamma_1\le \gamma_2\le\cdots\le \gamma_J$, which can be written in matrix form as, 
      \begin{equation}\label{eq:cond_suff}
      \bA\gamma\le 0\,,\;\textrm{ where }\;
\bA = \begin{bmatrix}
1 & -1 & 0 & \cdots & 0 & 0\\
0 & 1 & -1 &\cdots & 0 & 0\\
0 & 0 & 1 & \cdots & 0 & 0\\
\vdots & \vdots & \vdots & \ddots & \vdots & \vdots\\
0 & 0 & 0 & \cdots & 1 & -1
\end{bmatrix}_{(J-1)\times J}\,;
      \end{equation}
\item A necessary condition is that the first derivative is nonnegative at the knots, i.e., $f'(x) \ge 0\,, \forall x\in\{\xi_i\}_{i=0}^{K+1}$, which can be written in the following matrix form, 
\begin{equation}\label{eq:cond_nece}
\bfB^{(1)}\bD^{-1}\bA\gamma\le 0\,,
\end{equation}
where $\bfD$ is the diagonal matrix of size $(J-1)\times (J-1)$ with entries $\bfD_{jj} = \tau_{j+4} - \tau_{j+1}, j=1,\ldots,J-1$, and $\bfB^{(1)}$ is the matrix of size $(K+2)\times(J-1)$ with entries $\bfB^{(1)}_{i, j} = B_{j+1,3}(\xi_{i-1})$, i.e., the evaluation of the basis of one order lower $B_{j,3}(x),j=1,\ldots,J-1$ at $\xi_i, i=1,\ldots, K+2$.
\item Furthermore, a sufficient and necessary condition is that the first derivative is nonnegative at two boundary knots $\{\xi_0, \xi_{K+1}\}$ and the points $\{x: f''(x)=0\}$ with zero second-derivative, specifically,
\begin{equation}
f'(x) \ge 0\,,\forall x\in \underbrace{\{\xi_0, \xi_{K+1}\}\cup \bigcup_{\pi_i\in [0, 1],i=0,\ldots,K}\{\pi_i\xi_i+(1-\pi_i)\xi_{i+1}\}}_{\cX_0}\,,    
\label{eq:suff_and_nece}
\end{equation}

where
$$
\pi_i =\frac{A_{i+4}}{A_{i+4}-A_{i+3}}\,,A_{i+3} = \dfrac{1}{\tau_{i+5}-\tau_{i+3}}\left[\dfrac{\gamma_{i+3}-\gamma_{i+2}}{\tau_{i+6}-\tau_{i+3}} -\dfrac{\gamma_{i+2}-\gamma_{i+1}}{\tau_{i+5}-\tau_{i+2}} \right]\,.
$$
Let
$\bB^{(1)}_0$ be the matrix of size $n_0\times (J-1)$ with entries $\{\bB^{(1)}_0\}_{ij} = B_{j+1, 3}(x_{i})$, where $x_i\in \cX_0,  i=1,\ldots, n_0\triangleq\vert\cX_0\vert$ and $j=1,\ldots,J-1$, then the condition \eqref{eq:suff_and_nece} can be written in the following matrix from
\begin{equation}\label{eq:cond_suff_nece}
\bB^{(1)}_0\bD^{-1}\bA\gamma\le 0\,.
\end{equation}
\item Particularly, if $f''(x)\ge 0$, the sufficient condition $\bfA\gamma \le 0$ is also necessary for $f$ being non-decreasing. 
  \end{itemize}
  \label{prop:conditions}
\end{proposition}
\begin{remark}
    There is at most a point with zero second-derivative in each interval $(\xi_i, \xi_{i+1})$, which is given by $\pi_i\xi_i + (1-\pi_i)\xi_{i+1}$. But the point might not lie in the interval $(\xi_i, \xi_{i+1})$, so we restrict $\pi_i\in[0, 1]$. Alternatively, we can write 
    $\bar\pi_i = \pi_iI(0 \le \pi_i \le 1) + I(\pi_i > 1)$,
    then evaluate $f'(x)\ge 0$ at $\{\xi_0,\xi_{K+1}\}\cup \{\bar\pi\xi_i+(1-\bar\pi_i)\xi_{i+1}\}_{i=0}^K$.
\end{remark}

Figure~\ref{fig:demo_conditions} illustrates those conditions in Proposition~\ref{prop:conditions} with a simple spline function $f(x) =\sum_{j=1}^4\gamma_jB_j(x), x\in[0, 1]$. There are only two boundary knots $\{0, 1\}$ (no internal knots) when $J=4$, then all three conditions do not rely on the knot locations, and hence we can directly compare the conditions by checking the space of $\gamma$. When the gap between $\gamma_3$ and $\gamma_4$ becomes smaller, the sufficient and necessary condition tends to be closer to the sufficient condition.
\begin{figure}[H]
    \centering
    \begin{subfigure}{0.5\textwidth}
    \includegraphics[width=\textwidth]{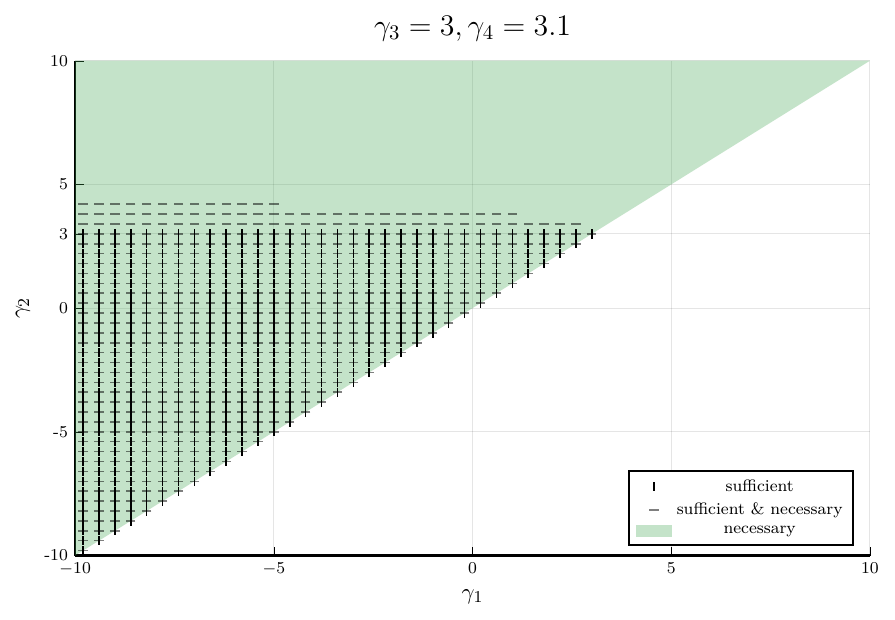}        
    \end{subfigure}%
    \begin{subfigure}{0.5\textwidth}
    \includegraphics[width=\textwidth]{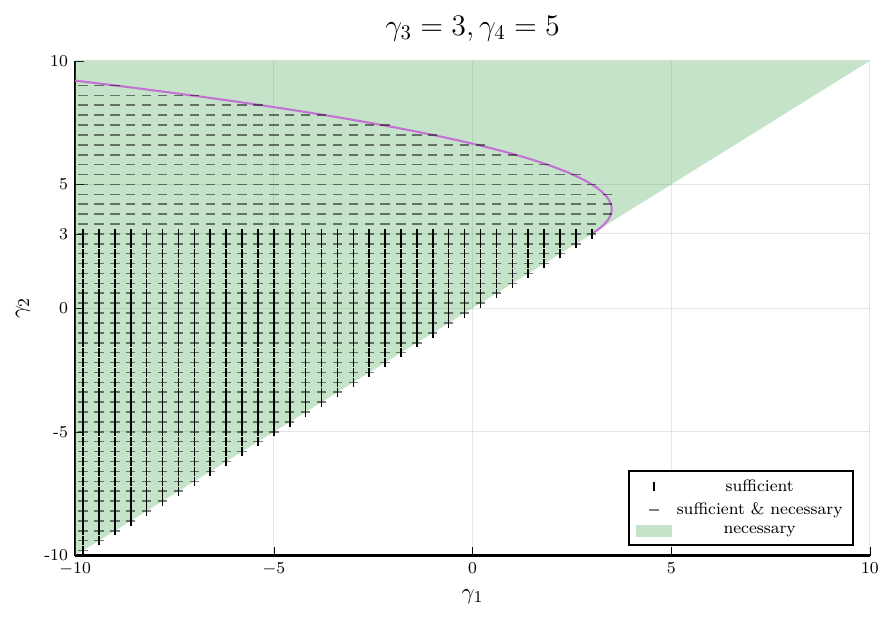}        
    \end{subfigure}
    \caption{Illustration of conditions for non-decreasing spline functions with $J=4$ basis functions and knots $\xi=\{0, 1\}$. Each panel shows the regions in the space of $(\gamma_1, \gamma_2)$ when $\gamma_3$ and $\gamma_4$ are fixed. The vertical hatching region is the sufficient condition, the horizontal hatching region denotes the sufficient and necessary condition, and the green shaded area shows the necessary condition. The upper boundary of the sufficient and necessary condition is $\gamma_1=\gamma_2 - (\gamma_2 - \gamma_3)^2/(\gamma_4-\gamma_3)$.}
    \label{fig:demo_conditions}
\end{figure}
Generally, we can illustrate the relationship between the conditions in Proposition~\ref{prop:conditions} using the diagram in Figure~\ref{fig:conditions}. The sufficient condition \eqref{eq:cond_suff} is the most restrictive, then the sufficient and necessary condition \eqref{eq:cond_suff_nece}, and finally, the necessary condition \eqref{eq:cond_nece}. 
\begin{figure}[H]
    \centering
    \begin{tikzpicture}
\coordinate (smallest) at (0, 0);
\coordinate (largest) at (0, 1);
\coordinate (middle) at (0, 0.5);
\coordinate (s1) at (0.3, 0.4);
\coordinate (s1p1) at (1.2, 1.2);
\coordinate (s1p2) at (1.6, -0.6);
\coordinate (s2) at (-1, 0.2);
\coordinate (s3) at (1, -0.25);
\node[ellipse,draw, minimum height = 3em] at (smallest) {$\bA\gamma \le 0$};
\node[ellipse, draw, minimum height=9em, minimum width=24em, label={[label distance=-2em]90:$\bfB^{(1)}\bfD^{-1}\bA\gamma \le 0$}] at (largest) {};
\node [ellipse, draw, dashed, minimum height= 6em, minimum width=15em, label={[label distance=-2em]90:$\bfB^{(1)}_{0}\bfD^{-1}\bA\gamma \le 0$}] at (middle) {};
\end{tikzpicture}
    \caption{Spaces of coefficient $\gamma$ under different conditions. The necessary and sufficient  condition is represented by the dashed ellipse, the sufficient condition is denoted by the smallest ellipse, and the necessary condition is shown by the largest ellipse. 
    }
    \label{fig:conditions}
\end{figure}
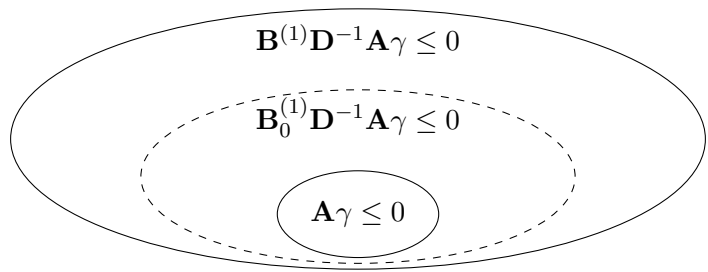

All three conditions in Proposition~\ref{prop:conditions} can be cast into a unified matrix form $\bS\gamma\le 0$,
where $\bS \in\{\bA, \bfB^{(1)}\bfD^{-1}\bfA, \bfB^{(1)}_0\bfD^{-1}\bfA\}$.
Both $\bfB^{(1)}\bfD^{-1}\bfA$ (necessary condition) and $\bfB^{(1)}_0\bfD^{-1}\bfA$ (sufficient and necessary condition) depend on the knot locations $\{\tau_i\}_{i=1}^{K+8}$. Moreover, $\bfB_0^{(1)}$ even depends on the coefficient vector $\gamma$. In other words, the sufficient condition is the simplest one, so we adopt the sufficient condition and formulate monotone splines as follows:
\begin{equation}
\begin{split}
    \min_{\gamma}\quad & (\bfy - \bB\gamma)^T(\bfy - \bB\gamma) + \lambda\gamma^T\bOmega\gamma\,,\\
    \subto\quad & \alpha\bA\gamma \le 0\,,
\end{split}    
\label{eq:matrss_cubic_smooth_spline_monotone}
\end{equation}
where $\alpha=1$ implies a non-decreasing function while $\alpha=-1$ results in a non-increasing spline. Without loss of generality, we focus on the non-decreasing scenario $\alpha=1$. 
We call the resulting fit the  \emph{Monotone Cubic Spline (MCS)} if
there is no smoothness penalty, i.e., $\lambda=0$; otherwise, we call it \emph{Monotone Smoothing Spline (MSS)} (i.e., if $\lambda>0$).

We derive an estimation error bound for monotone cubic splines to demonstrate the ability of the sufficient but not necessary condition to fit any monotone functions.
Without loss of generality, we restrict $x\in [0, 1]$. Let $0 = \xi_0 < \xi_1 < \cdots < \xi_{K+1}=1$ be nearly equally spaced knots of the B-spline, where the number of (internal) knots $K$ grows along with the number of observations $n$, i.e., $K =J-4= k(n)$ for some function $k(\cdot )$.
\begin{theorem}\label{thm:gap_suff_f}
Suppose $f(x)=\sum_{j=1}^J\gamma_jB_j(x)$ is a non-decreasing function in $[\xi_0, \xi_{K+1}]$, i.e., $f'(x)\ge 0$.  Given $n$ observations $\{(x_i, y_i)\}_{i=1}^n$, where $y_i=f(x_i)+\varepsilon_i$. Consider the monotone spline fitting $\hat f(x) = \sum_{j=1}^J\hat\gamma_jB_j(x)$, where $\hat\gamma$ is the solution to Problem \eqref{eq:matrss_cubic_smooth_spline_monotone} ($\lambda=0$).
Under Assumptions~\ref{ass:1},\ref{ass:2},\ref{ass:3},
\begin{assumption}[Bounded second derivative]\label{ass:1}
There is a constant $L$ such that $\vert f''(x)\vert \le L$; 
\end{assumption}
\begin{assumption}[Nearly uniform $x_i$'s]\label{ass:2}
$x_i$'s are nearly uniformly located in $[\xi_0, \xi_{K+1}]$. Specifically, there are at most $\frac{n}{K+1}(1+\eta_1)$ points between any two adjacent knots, where $\eta_1$ controls the bias of number of points since there are $\frac{n}{K+1}$ points in each interval on average;
\end{assumption}
\begin{assumption}[Nearly equally spaced knots]\label{ass:3}
The knots are nearly equally spaced. Specifically, the interval width between any two adjacent knots is at most $\frac{(1+\eta_2)}{K+1}$, where $\eta_2$ controls the difference from the average interval length $\frac{1}{K+1}$.
\end{assumption}
we have
\begin{itemize}
    \item No error: If $\varepsilon_i = 0$,
$$
\frac 1n \sum_{i=1}^n(f(x_i)-\hat f(x_i))^2 \le \frac{36(1+\eta_1)(1+\eta_2)^2L^2J}{(J-3)^3} = O(J^{-2})\,.
$$
    \item Nonzero error: If $\varepsilon_i$'s are i.i.d. sub-Gaussian errors with parameter $\sigma>0$. Then if $J=C n^{1/3}$ for some constant $C > 0$, for any $M\ge 0$, it holds with at least probability $1-2J^{-M^2}$ that
    \begin{align*}
    \frac{1}{n}\sum_{i=1}^n(f(x_i)-\hat f(x_i))^2
    &\le\frac{36(1+\eta_1)(1+\eta_2)^2L^2J}{(J-3)^3} + \frac{32}{c_1}\sigma^2(1+M)^2(1+\eta_1)\frac{\log J}{(J-3)^2}\\
    &=O\left(\frac{\log J}{J^2}\right)\,,     
    \end{align*}
    where $c_1>0$ is a constant such that the minimum eigenvalue of $\frac 1n\bfB^T\bfB \ge c_1/J$ \parencite{shenLocalAsymptoticsRegression1998}.
\end{itemize}
\end{theorem}
Theorem~\ref{thm:gap_suff_f} implies that monotone splines based on the sufficient condition can achieve a small approximation error, which can be further reduced with more basis functions in order $J^{-2}$, to the monotone splines based on the sufficient and necessary condition. Furthermore, besides functions represented by B-splines, we can also obtain the approximation error to arbitrary monotone functions, as stated in Theorem~\ref{thm:gap_suff_g}.
\begin{theorem}\label{thm:gap_suff_g}
Suppose $g$ is a strictly increasing function, i.e., $g'(x) > 0$, and consider $n$ observations $\{(x_i, g(x_i))\}_{i=1}^n$.
Let $\hat g_n=\sum_{j=1}^J\hat\gamma_jB_j(x)$ be the monotone spline fitting under the sufficient condition, where $\hat\gamma$ is the solution to Problem \eqref{eq:matrss_cubic_smooth_spline_monotone}  with $\lambda=0$, and $\tilde g_n(x) = \sum_{j=1}^J\tilde \gamma_jB_j(x)$ be the monotone spline fitting based on the sufficient and necessary condition, where $\tilde\gamma$ is the solution to Problem~\eqref{eq:matrss_cubic_smooth_spline_monotone} by replacing the condition \eqref{eq:cond_suff} with \eqref{eq:cond_suff_nece}. And denote $\check g_n = \sum_{j=1}^J\check \gamma_jB_j(x)$ as the unconstrained B-spline fitting, where $\check \gamma$ is the solution to Problem~\eqref{eq:matrss_cubic_smooth_spline} ($\lambda=0$). Under Assumptions~\ref{ass:1},\ref{ass:2},\ref{ass:3}, when $n$ is sufficiently large
\begin{itemize}
    \item the monotone spline fitting $\tilde g_n$ is identical to the unconstrained B-spline fitting $\check g_n$, i.e., $\check g_n = \tilde g_n$,
    and we have
    \begin{equation}
    \frac 1n\sum_{i=1}^n(\tilde g_n(x_i) - g(x_i))^2 = O(J^{-8})\,.
    \label{eq:error_rate_J8}
    \end{equation}
    \item the monotone spline fitting $\hat g_n$ based on the sufficient condition satisfies
    \begin{equation}
    \frac{1}{n}\sum_{i=1}^n(\hat g_n(x_i) - g(x_i))^2 = O(J^{-2})\,.  
    \label{eq:error_rate_J2}
    \end{equation}
\end{itemize}
\end{theorem}
\begin{remark}
    The error bound in Equation~\eqref{eq:error_rate_J8} for the monotone spline fitting $\tilde g_n$ based on the (complicated) sufficient and necessary condition is derived from the bias bound of unconstrained spline fitting (e.g., \textcite{shenLocalAsymptoticsRegression1998}), which is quite tight. However, the error bound in Equation~\eqref{eq:error_rate_J2} for the monotone spline fitting $\hat g_n$ based on the simple sufficient condition is relatively loose, where we take the asymptotic results \parencite{yangContractionUniformConvergence2019} from the isotonic regression as an internal step since the monotonic coefficients $\hat \gamma_1\le \cdots \le \hat \gamma_J$ can be viewed as an isotonic fitting to the following isotonic regression,
    $$
    \hat\gamma_1,\ldots,\hat\gamma_J = 
    \argmin_{\beta_1\le \cdots\le \beta_J}\sum_{j=1}^J(\gamma_j-\beta_j)^2\,,
    $$
    so there might be some scarification in the error bound due to the internal isotinisation step. Note that Theorem~\ref{thm:gap_suff_f} also depends on such an internal step, so the error bounds therein might be improved.
    Investigating more tight error bounds or the min-max lower bounds might be a potentially interesting direction.
\end{remark}

The proofs of those theorems (and theorems in the following section) are given in the \supp.

\subsection{Characterization of Solutions}\label{sec:sol}
Theorem \ref{thm:solution} 
describes the solutions of monotone splines. If a solution $\hat\gamma$ has no ties, i.e., $\bA\hat\gamma < 0$ strictly holds, then the solution is the same as the unconstrained splines. If the solution has ties, then the solution can be written as a least-square-like form using unique elements of the solution.
\begin{theorem}\label{thm:solution}
  Let $\hat\gamma$ be the solution to Problem \eqref{eq:matrss_cubic_smooth_spline_monotone} when $\alpha=1$.
  \begin{itemize}
      \item If there is no ties in $\hat\gamma$, i.e., $\hat\gamma_1 <\cdots < \hat\gamma_J$, then
      $\hat\gamma = (\bfB^T\bfB+\lambda\bOmega)^{-1}\bfB^T\bfy\,.$
      \item If there exists ties in $\hat\gamma$, such as
$\hat\gamma_1<\cdots<\hat\gamma_{k_1}=\cdots=\hat\gamma_{k_2} < \cdots < \hat\gamma_{k_{m-1}} = \hat\gamma_{k_{m}} <\cdots< \hat\gamma_J\,,$
where $1\le k_1\le k_2\le\cdots \le k_{m-1}\le k_{m}\le J$, and let $\hat\beta$ be the sub-vector $\hat\gamma$ with unique entries, then
$$
\hat\gamma=\bfG^T\hat\beta = \bG^T(\bG\bB^T\bB\bG^T+\lambda\bG\bOmega\bG^T)^{-1}\bG\bB^T\bfy\,,
$$
where
\begin{equation}
    \bG = \begin{bmatrix}
\bI_{k_1-1} & & &&&\\
& \one^T_{k_2-k_1+1} & &&&\\
& & \ddots &&&\\
 &  &  & \bI_{k_{m-1}-k_{m-2}-1} &&\\
 &  & &  & \one^T_{k_{m}-k_{m-1}+1} &\\
 &  & &  & & \bI_{J-k_{m}}
\end{bmatrix}\,,
\label{eq:md_g}
\end{equation}
in which $\one$ is the all-ones vector, and $\bI$ is the identity matrix.
If $\bfG=\bfI$, it reduces to the above no-tie case.
  \end{itemize}
\end{theorem}
With the solution given in Theorem~\ref{thm:solution}, we can explicitly compare the mean square error (MSE) between the monotone cubic spline and the classical cubic spline. Theorem~\ref{thm:mse} implies that the monotone cubic spline can achieve a better MSE when the noise level is large, which would be further validated in the simulations of Section~\ref{sec:monobspl_sim}.

\begin{theorem}
Suppose observations $\{(x_i, y_i)\}_{i=1}^n$ are generated from $y = f(x)+\varepsilon, \varepsilon\sim N(0, \sigma^2)$.
Let $\bfB$ with entries $\bfB_{ij}=B_j(x_i)$ be the evaluated B-spline matrix and denote $\bff = [f(x_1),\ldots,f(x_n)]^T$.
Consider the MSE of the monotone cubic spline $\hat\bfy = \bfB\hat\gamma$, where $\hat\gamma$ is the solution to Problem \eqref{eq:matrss_cubic_smooth_spline_monotone} with $\lambda=0$, 
and the MSE of the cubic spline $\hat\bfy^\ls = \bfB\hat\gamma^\ls$, where $\hat\gamma^\ls$ is the solution to Problem \eqref{eq:matrss_cubic_smooth_spline} with $\lambda=0$,
$$
\MSE(\hat\bfy) =\bbE \Vert \bB\hat\gamma-\bff\Vert_2^2
\,,\qquad \MSE(\hat\bfy^\ls) = \bbE\Vert\bB\hat\gamma^\ls-\bff\Vert_2^2\,.
$$
    If $\sigma^2\ge \frac{\bff^T(\bH-\bH_g)\bff}{J-g}$, where $\bH=\bfB(\bfB^T\bfB)^{-1}\bfB^T, \bH_g = \bB\bG^T(\bfG\bfB^T\bfB\bfG^T)^{-1}\bfG\bfB^T$ and $\bG$ of size $g\times J$ is defined in Equation~\eqref{eq:md_g}, the monotone cubic spline can achieve a better MSE since $\bfH-\bfH_g$ is a positive semidefinite matrix.
\label{thm:mse}
\end{theorem}




\subsection{Selection of Parameters}\label{sec:monobspl_paras}

The tuning parameters of cubic splines include the number and placement of the knots. However, selecting the placement and number of knots can be a combinatorially complex task. A simple but adaptive way is to only determine the number of knots and places the knots at appropriate quantiles of the predictor variables \parencite{hastieGeneralizedAdditiveModels1990}. 
Specifically, we choose $K$ interior knots as the $j/(K+1), j=0, 1,\ldots,K, K+1$ quantile of the predictor variable, where $j=0$ and $j=K+1$ represent two boundary points. Since the number of interior knots $K$ and the number of basis functions $J$ satisfy $J=K+4$, where $4$ comes from the order of cubic spline, it turns out to select the number of basis functions.

In addition to the popular cross-validation (CV), there are other widely used criteria for model selection, which can be quickly calculated, such as Akaike information criterion (AIC), Bayesian information criterion (BIC), and generalized cross-validation (GCV),
\begin{align*}
  \text{AIC} &= n\log\sum_{i=1}^n(y_i-\hat f(x_i))^2 + 2\df\,,\\
  \text{BIC} &= n\log\sum_{i=1}^n(y_i-\hat f(x_i))^2 + \df\log n\,,\\
  \text{GCV} &=\frac{\sum_{i=1}^n(y_i-\hat f(x_i))^2}{(1-\df/n)^2}\,.
\end{align*}
All of them involve the degrees of freedom ($\df$). 
For monotone cubic splines, the degree of freedom can be derived based on the results of \textcite{chenDegreesFreedomProjection2020}, and the proof is given in the \supp.
\begin{proposition}
 \label{coro:df_mono_cubic}
The degrees of freedom for the monotone cubic B-spline $\hat\bfy=\bfB\hat\gamma$ is
\begin{equation}
  \df = \bbE[U_\bfy]\,,
  \label{eq:df_cubic_mono_spl}
\end{equation}
where $U_\bfy$ (depends on $\bfy$) is the number of unique coefficients in $\hat\gamma$.
\end{proposition}

On the other hand, the smoothing splines avoid the knot selection problem entirely by taking all unique $x_{i}$'s as the knots and controlling the complexity only by the regularization parameter $\lambda$. Actually, in practice, it is unnecessary to use all unique $x_i$'s, and any reasonable thinning strategy can save in computations and have a negligible effect on the fitness \parencite{hastieElementsStatisticalLearning2009}. In other words, for smoothing splines, we only need to tune the regularization parameter $\lambda$ and treat the number of basis functions as fixed.
Practically, the GCV principle is usually used to find the best $\lambda$, which can alleviate the potential high computational burden of CV. Thus, we also use the GCV criterion to determine the parameter of the monotone smoothing splines.

\section{Two Algorithms}\label{sec:two_alg}
This section introduces and compares two algorithms for fitting monotone cubic B-splines: 
\begin{itemize}
    \item Algorithm \ref{alg:point_opt}: the optimization (abbreviated as OPT hereafter) approach based on existing toolboxes;
    \item Algorithm \ref{alg:point_mlp}: the Multi-Layer Perceptrons (MLP) generator, which takes advantage of the flexibility of neural networks.
\end{itemize}

For simplicity, we focus on the increasing case $\alpha=1$, but it is straightforward to apply the results of the increasing case to the decreasing case $\alpha=-1$. The optimization problem in Equation \eqref{eq:matrss_cubic_smooth_spline_monotone} is a classical convex second-order cone problem by rewriting
\begin{align}
    \min &\; z\\
    \subto\; & \left\Vert \begin{bmatrix}
        \bfy - \bfB\gamma \\
        \sqrt{\lambda} \bfL^T\gamma
    \end{bmatrix}\right\Vert_2 \le z\label{eq:cone}\,,\\
    & \bA\gamma \le 0\,,\label{eq:mono_constraint}
\end{align}
where $\bfOmega = \bfL\bfL^T$ is the Cholesky's decomposition. Inequality \eqref{eq:cone} implies a cone
$$
\cQ = \{(z, u)\in \IR\times \IR^{n+J}\mid z\ge \Vert u\Vert_2\}\,,
$$
so we can adopt many mature optimization toolboxes to solve such a problem, such as \textcite{domahidiECOSSOCPSolver2013}'s ECOS (Embedded Conic Solver) and \textcite{grantCVXMatlabSoftware2014}'s disciplined convex programming system CVX. 
\begin{algorithm}
    \caption{Point Estimate: OPT solution}
    \label{alg:point_opt}
    \begin{algorithmic}[1]
    \REQUIRE Dataset $\bZ = \{(x_i, y_i)\}_{i=1}^n$, penalty parameter $\lambda$.
    \STATE Feed the problem into the existing optimization toolbox ECOS (or others)
    \RETURN $\hat\gamma$
    \end{algorithmic}
\end{algorithm}

Recall that without the monotonicity constraint \eqref{eq:mono_constraint}, the solution is expressed in Equation \eqref{eq:sol_ridge},
which is a function of $\bfy$ and $\lambda$ since usually $\bfB$ and $\bfOmega$ are treated as given. Furthermore, if we let $\bfy$ be given, then $\hat\gamma$ is a function in the penalty parameter $\lambda$. Then for monotone splines with the monotonicity constraint, a natural question is whether we can find a function of $\lambda$ to provide the solution for each $\lambda$. If we find the formula $G(\lambda)$, we can obtain the solution at a new  $\lambda$ by evaluating the function $G$ at $\lambda$ instead of re-running the optimization program by specifying the penalty parameter $\lambda$.

Inspired by \textcite{shinGenerativeMultiplepurposeSampler2022}'s Generative Multiple-purpose Sampler (GMS) (with the Generative Bootstrap Sampler (GBS) as a special case for bootstrap), we take a new viewpoint at the constrained solution by explicitly treating $\hat\gamma$ as a function of $\bfB$, $\bfy$, and the penalty parameter $\lambda$, denoted by $G(\bfB, \bfy, \lambda)$, which is also required to be a monotonic vector to fulfill the monotonicity constraint. Usually, the basis matrix $\bfB$ is fixed, so we ignore it in the functional argument of $G$, that is, $G(\bfy, \lambda)$. For an estimate of the coefficient $\gamma$, $\bfy$ is also fixed in this section, so it can be further written as $G(\lambda)$, but we will let $\bfy$ be random to consider the confidence band in the next section. Recently, the neural network has become a powerful tool for function representation, so we adopt the Multi-Layer Perceptrons (MLP) to construct our family of functions, that is, $\cG = \{G_\phi: \IR^{n+1}\rightarrow \IR^J,\phi\in \Phi\}$, where $\Phi$ represents the space of parameters that characterize the function family and
\begin{align*}
G_\phi(\bfy, \lambda) &= \texttt{sort}\circ \texttt{MLP}_\phi(\bfy, \lambda)=\texttt{sort}\circ L \circ g_K\circ\cdots \circ g_1(\bfy, \lambda)\,,
\end{align*}
where $g_k:\IR^{n_k}\rightarrow \IR^{n_{k+1}}$ is the feed-forward mapping with an activation function $\sigma:\IR\rightarrow\IR$, which operates element-wise when the input is a vector or a matrix,
$g_k(x) = \sigma(\bfw_kx+\bfb_k)\,,$
and $L:\IR^{n_{K+1}}\rightarrow \IR^J$ is a linear function mapping the final hidden layer $g_K\circ \cdots g_1(x)$ to the $J$-dimensional output space of $g$, and $\texttt{sort}:\IR^J\rightarrow \IR^J$ returns the vector in an ascending order.
\begin{remark}
    Generally, an MLP is trained by back-propagation, which requires the functions to be differentiable. However, \texttt{sort} is not a standard function and is not differentiable. On the other hand, for a vector $v$ of length $m$, the \texttt{sort} operation can be written as
    $\texttt{sort}(v) = \bfP v\,,$
    where $\bfP$ is an $m\times m$ permutation matrix. Practically, the deep learning frameworks, such as \texttt{PyTorch}\footnote{\texttt{\url{https://pytorch.org/docs/stable/notes/autograd.html}}} and \texttt{Flux}\footnote{\texttt{\url{https://fluxml.ai/Zygote.jl/latest/adjoints/}}}, would define
    $\nabla \texttt{sort}(v) = \bfP\,.$
    Besides, some researchers discuss the differentiable variants of the sort operation, such as \textcite{blondelFastDifferentiableSorting2020} and \textcite{groverStochasticOptimizationSorting2019}.
\end{remark}
In other words, we want to take advantage of the flexibility of neural networks to construct a generator $G(\lambda)$ to approximate the solution for each $\lambda$. Then the target function becomes
$$
\hat G=\argmin_{G\in \cG} \Vert \bfy -\bfB G(\bfy, \lambda) \Vert_2^2 + \lambda \Vert\bfL G(\bfy, \lambda)\Vert_2^2\,, \quad \forall \lambda\in [\lambda_l,\lambda_u]\,.
$$
 We consider a less ambitious but more robust and practically almost equivalent formulation by integrating $\lambda$ out,
\begin{equation}
\hat G = \argmin_{G\in \cG}\bbE_\lambda \left[\Vert \bfy -\bfB G(\bfy, \lambda) \Vert_2^2 + \lambda \Vert\bfL G(\bfy, \lambda)\Vert_2^2\right]\,,  
\label{eq:g_point_est}
\end{equation}
where $\lambda\sim U([\lambda_l,\lambda_u])$. Practically, we generate Monte Carlo samples to approximate it,
$$
\hat G = \argmin_{G\in \cG}\frac 1M \sum_{i=1}^M \left[\Vert \bfy -\bfB G(\bfy, \lambda_i) \Vert_2^2 + \lambda \Vert\bfL G(\bfy, \lambda_i)\Vert_2^2\right]\,.
$$
The MLP generator $G(\lambda) \triangleq G(\lambda\mid\bfy)$ is summarized in Algorithm \ref{alg:point_mlp}.
\begin{algorithm}
    \caption{Point Estimate: MLP Generator $G(\lambda\mid \bfy)$}
    \label{alg:point_mlp}
    \begin{algorithmic}[1]
    \REQUIRE Dataset $\bZ = \{(x_i, y_i)\}_{i=1}^n$, batch size $M$.
    \WHILE{not converged}
    \STATE Sample $\{\lambda_i\}_{i=1}^M$ where each $\lambda_i\sim U[\lambda_l, \lambda_u]$.
    \STATE Define the loss
    $$
    \cL = \frac{1}{M} \sum_{i=1}^M\left\{\Vert\bfy - \bfB G(\bfy, \lambda_i)\Vert_2^2 + \lambda_i\Vert\bfL G(\bfy, \lambda_i)\Vert_2^2\right\}\,.
    $$
    \STATE Update neural network $G$ to minimize $\cL$.
    \ENDWHILE
    \end{algorithmic}
\end{algorithm}

Figure \ref{fig:demo_mlp_generator} shows a demo using MLP to fit the data generated from a cubic curve with noise $\sigma=0.2$. The left panel displays the training loss $\bbE_\lambda[\cL(\lambda)]$, together with the losses evaluated at the boundary of tuning parameters $\lambda\in[\lambda_l, \lambda_u]$, $\cL(\lambda_l)$ and $\cL(\lambda_u)$. The right panel shows that for each $\lambda$, the solid fitted curve obtained from the OPT solution and the dashed curves obtained from the MLP generator coincide quite well, which indicates the MLP generator achieves a pretty good approximation.
\begin{figure}
    \centering
    \begin{subfigure}{0.475\textwidth}
    \includegraphics[width=\textwidth]{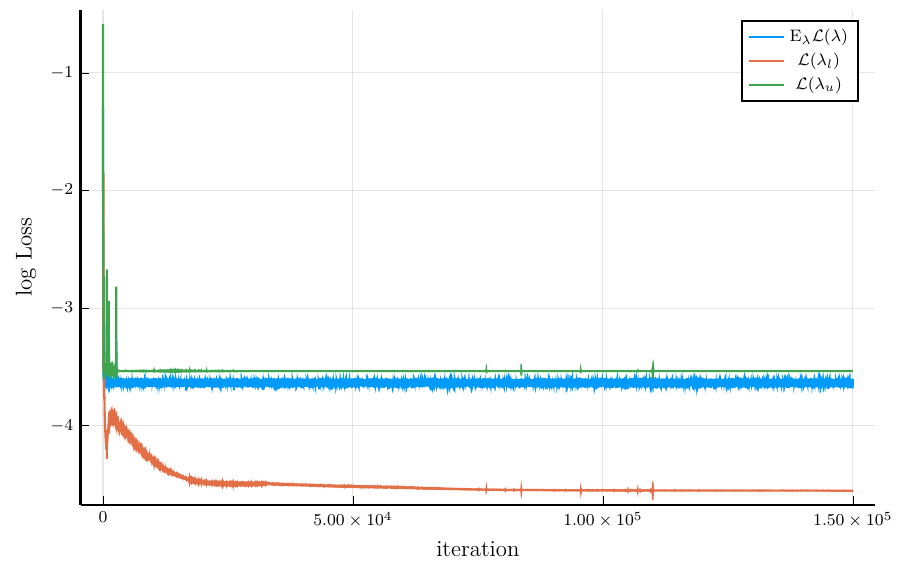}
    \end{subfigure}%
    \begin{subfigure}{0.525\textwidth}
    \includegraphics[width=\textwidth]{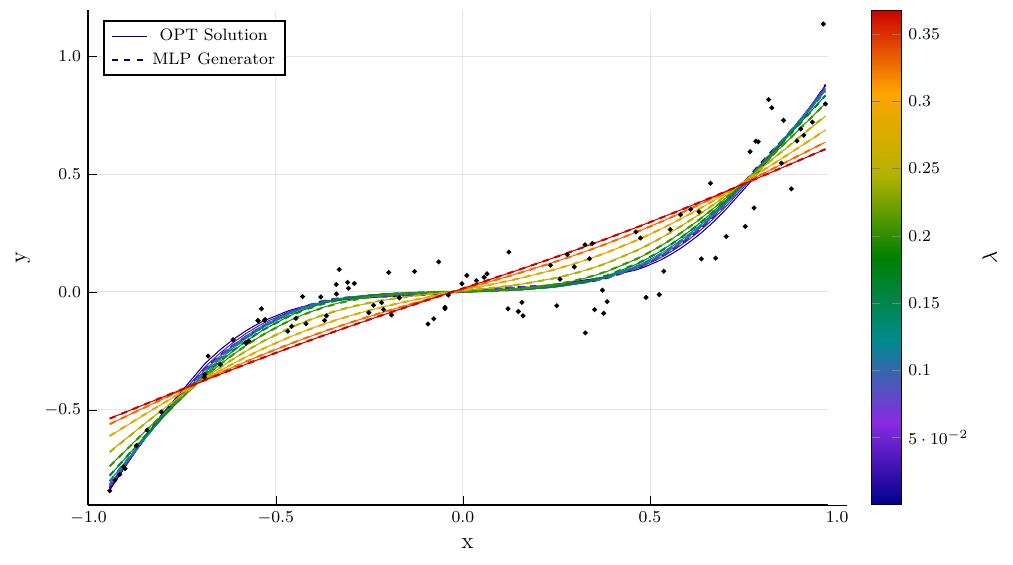}
    \end{subfigure}%
    \caption[Demo of MLP Generator]{Demo of MLP Generator on a cubic curve with noise 0.2. (a) Training loss $\bbE_\lambda[\cL(\lambda)]$ and two losses evaluated at the boundary of the penalty parameters $\lambda\in[\lambda_l, \lambda_u]$, $\cL(\lambda_l)$ and $\cL(\lambda_u)$. (b) The fitted curves via OPT solution (solid curves) and MLP generator (dashed curves) for each $\lambda$.}
    \label{fig:demo_mlp_generator}
\end{figure}


Since the MLP generator solution from Algorithm \ref{alg:point_mlp} is actually an approximation to the OPT solution from Algorithm \ref{alg:point_opt}, we compare these two solutions to measure the performance of the MLP generator. Firstly, we consider the difference between these two solutions, $\Vert \bfB G(\bfy, \lambda) - \bfB \hat\gamma(\lambda)\Vert$.
A relative one would be more informative, which alleviates the magnitude effect of the curve itself,
$$
\text{Relative Gap} = \frac{\Vert \bfB G(\bfy, \lambda)-\bfB\hat\gamma(\lambda)\Vert_2^2 }{\Vert \bfB\hat\gamma(\lambda)\Vert_2^2}\,.
$$
We also compare their fitness to the noise observation $\bfy$, $\Vert \bfy - \bfB G(\bfy, \lambda)\Vert$ and $\Vert \bfy - \bfB \hat\gamma(\lambda)\Vert$,
$$
\text{Fitness Ratio} = \frac{\Vert \bfy-\bfB G(\bfy,\lambda)\Vert_2^2}{\Vert \bfy-\bfB\hat\gamma(\lambda)\Vert_2^2}\,.
$$

We conduct 5 repeated experiments on four different curves and three different noise levels $\sigma=0.1, 0.2, 0.5$. The data are generated from 
\begin{equation}
    \begin{split}
    x_i&\sim U[-1, 1]\,,\\
y_i&=f_j(x_i)+N(0, \sigma^2), i=1,\ldots,n=100\,,\\
f_1(x) &= \frac{\exp(5x)}{1+\exp(5x)}\triangleq S(5x)\,, f_2(x) = e^x\,, f_3(x)=x^3\,, f_4(x)=\sin\left(\frac{\pi}{2}x\right)\,.
    \end{split}
    \label{eq:mlp_data}
\end{equation}
We choose the studied region $[\lambda_l,\lambda_u]$ of penalty parameter as $\lambda\in [\exp(-8), \exp(-2)]$, which is wide enough to contain the minimizer of the cross-validation error 
(see Figure S1b in the \supp). 
Table~\ref{tab:mlp_acc} summarizes the mean relative gap and mean fitness ratio, together with their standard errors, among 5 repeated experiments. Both relative gap and fitness ratio are measured at 10 even-spaced $\lambda$ in $[\lambda_l, \lambda_u]$, and Table~\ref{tab:mlp_acc} reports the values at $\lambda_l,\lambda_u$, and the average (column ``Avg.'') over 10 $\lambda$'s.
\begin{table}
  \caption[Accuracy of MLP Generator]{Mean Relative Gap and Mean Fitness Ratio (with standard error in parentheses) between MLP Generator solution and optimization solution among 5 repetitions. The ``Avg.'' represents the average measurement over 10 even-spaced $\lambda$ in $[\lambda_l, \lambda_u]$.}
  \label{tab:mlp_acc}
  \resizebox{\textwidth}{!}{
    \begin{tabular}{cccccccc}
\toprule
\multirow{2}{*}{noise} & \multirow{2}{*}{curve}&\multicolumn{3}{c}{Relative Gap}&\multicolumn{3}{c}{Fitness Ratio}\tabularnewline
\cmidrule(lr){3-5}
\cmidrule(lr){6-8}
&&$\lambda_l$&$\lambda_u$&Avg.&$\lambda_l$&$\lambda_u$&Avg.\tabularnewline
\midrule
\multirow{4}{*}{$\sigma = 0.1$}&$S(5x)$& 4.97e-04 (5.7e-04)& 9.84e-07 (1.0e-06)& 9.77e-05 (2.5e-04)& 1.03e+00 (1.7e-02)& 1.00e+00 (4.2e-03)& 1.01e+00 (1.1e-02)\tabularnewline
&$e^x$& 1.19e-04 (8.5e-05)& 7.19e-07 (1.5e-07)& 2.18e-05 (4.7e-05)& 1.04e+00 (2.3e-02)& 1.01e+00 (3.5e-03)& 1.01e+00 (1.4e-02)\tabularnewline
&$x^3$& 3.72e-04 (2.0e-04)& 2.72e-06 (1.2e-06)& 5.09e-05 (1.2e-04)& 1.01e+00 (6.4e-03)& 9.95e-01 (1.9e-03)& 1.00e+00 (4.4e-03)\tabularnewline
&$\sin(\pi x/2)$& 1.04e-03 (8.2e-04)& 2.68e-06 (2.2e-06)& 1.73e-04 (4.2e-04)& 1.09e+00 (6.5e-02)& 1.01e+00 (8.6e-03)& 1.02e+00 (3.4e-02)\tabularnewline
\midrule
\multirow{4}{*}{$\sigma = 0.2$}&$S(5x)$& 3.09e-04 (2.5e-04)& 3.05e-06 (2.7e-06)& 6.39e-05 (1.3e-04)& 1.02e+00 (8.0e-03)& 1.00e+00 (2.4e-03)& 1.00e+00 (5.9e-03)\tabularnewline
&$e^x$& 1.01e-04 (2.5e-05)& 5.52e-07 (2.6e-07)& 1.61e-05 (3.2e-05)& 1.02e+00 (7.3e-03)& 1.00e+00 (1.4e-03)& 1.00e+00 (6.6e-03)\tabularnewline
&$x^3$& 6.47e-04 (8.0e-04)& 2.18e-06 (1.5e-06)& 1.09e-04 (3.4e-04)& 1.01e+00 (6.8e-03)& 9.99e-01 (1.8e-03)& 1.00e+00 (3.6e-03)\tabularnewline
&$\sin(\pi x/2)$& 8.81e-04 (7.4e-04)& 1.87e-06 (1.7e-06)& 1.42e-04 (3.6e-04)& 1.02e+00 (1.7e-02)& 1.00e+00 (2.2e-03)& 1.00e+00 (9.2e-03)\tabularnewline
\midrule
\multirow{4}{*}{$\sigma = 0.5$}&$S(5x)$& 7.67e-04 (4.9e-04)& 2.76e-05 (4.6e-05)& 1.37e-04 (2.8e-04)& 1.01e+00 (5.1e-03)& 1.00e+00 (2.3e-03)& 1.00e+00 (2.8e-03)\tabularnewline
&$e^x$& 2.97e-04 (1.5e-04)& 6.73e-07 (5.9e-07)& 4.56e-05 (1.0e-04)& 1.01e+00 (4.9e-03)& 1.00e+00 (3.6e-04)& 1.00e+00 (3.8e-03)\tabularnewline
&$x^3$& 7.02e-04 (2.8e-04)& 2.86e-06 (1.2e-06)& 9.72e-05 (2.2e-04)& 1.00e+00 (4.6e-04)& 1.00e+00 (3.8e-04)& 1.00e+00 (1.1e-03)\tabularnewline
&$\sin(\pi x/2)$& 1.03e-03 (1.5e-03)& 3.40e-06 (4.9e-06)& 1.84e-04 (5.9e-04)& 1.01e+00 (6.8e-03)& 1.00e+00 (5.4e-04)& 1.00e+00 (3.7e-03)\tabularnewline
\bottomrule
\end{tabular}
}
\end{table}
We can find that the fitness ratios are pretty close to 1 and the relative gaps are at quite low values, both of which indicate that the fitting from the MLP generator is a good approximation to the OPT solution.

To demonstrate the efficiency of the two algorithms, we compare the running time of the OPT approach and the MLP generator, which are shown in Figure \ref{fig:run_time}. We take a cubic curve with noise level $\sigma=0.2$, and vary the sample size from $\{50, 100, 200, 500, 1000, 2000, 5000\}$. The number of iterations 50000 in the training step of the MLP generator is large enough to guarantee convergence. Figure \ref{fig:fitness_ratio_cubic} displays the resulting average fitness ratio, which is close to 1.0 for each studied sample size $n$. In other words, the training of the MLP generator has been sufficient to achieve a good approximation. The MLP training runs on an \texttt{Nvidia-A100-SXM4-40GB} GPU using \texttt{PyTorch}, then the evaluation of the trained MLP and the OPT approach run on the same machine using an \texttt{AMD-EPYC-7742} CPU.

\begin{figure}
    \centering
    \begin{subfigure}{0.5\textwidth}
        \includegraphics[width=\textwidth]{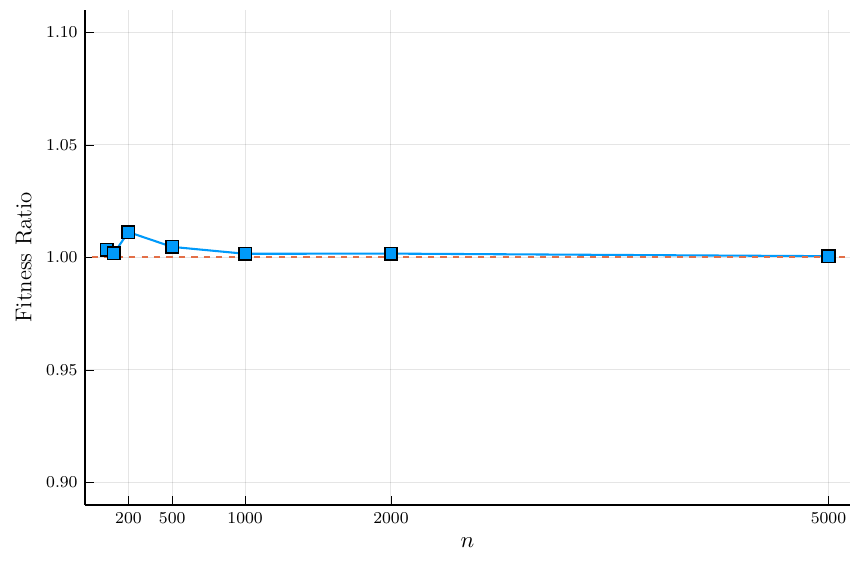}
        \caption{}
        \label{fig:fitness_ratio_cubic}        
    \end{subfigure}%
    \begin{subfigure}{0.5\textwidth}
        \includegraphics[width=\textwidth]{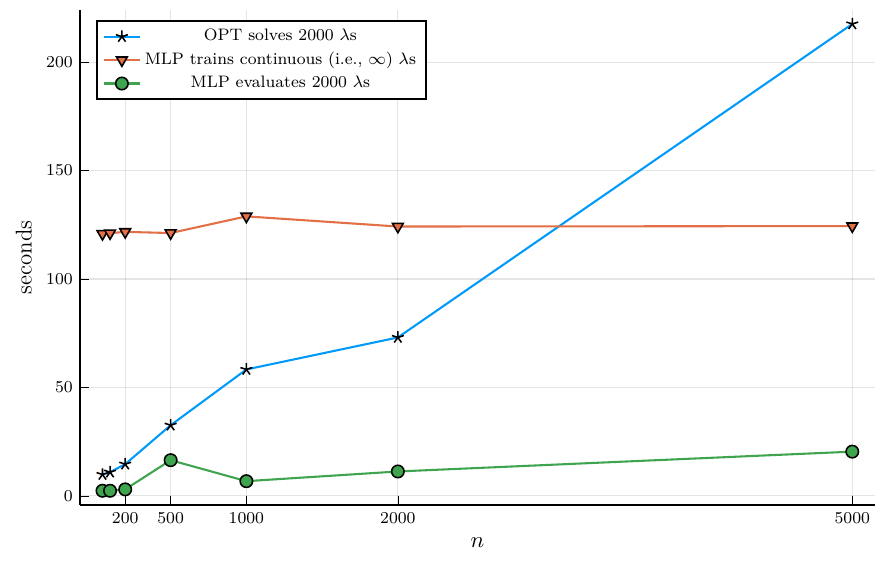}
        \caption{}
        \label{fig:run_time}
    \end{subfigure}
    \caption{(a) Average fitness ratio for the trained MLP generator. (b) Running time of the MLP generator (training and evaluation) and the OPT approach.}
\end{figure}
Figure \ref{fig:run_time} shows the running time for the OPT approach solving 2000 optimization problems with 2000 different $\lambda$'s, the training step of the MLP generator, and the evaluation step of the trained MLP step on the same 2000 different $\lambda$'s. In practice, to save computational time, we only count the time cost of running 10 $\lambda$'s, then multiply 200 to get the time for 2000 $\lambda$'s since each optimization problem takes a similar time cost.
The running time of the OPT approach increases nearly linearly along the sample size $n$. In contrast, the training time of the MLP generator is not affected by the sample size, and the evaluation, which just plugs a $\lambda$ into the trained MLP generator $\hat G$, is much cheaper. More importantly, the trained MLP generator works for continuous $\lambda$ located in the range $[\lambda_l, \lambda_u]$, while the OPT approach needs to run one optimization problem for each candidate penalty parameter $\lambda$. Since each optimization problem costs similar (if not the same) time, then it is expected to take $k$ times the running time shown in Figure \ref{fig:run_time} to solve $2000k$ optimization problems. On the other hand, the MLP generator can evaluate those parameters $\lambda$'s in a much shorter time. Thus, the MLP generator can save time by avoiding repeating to run the optimization problems. The more evaluations, the more time can be saved.

In practice, we might not evaluate so many $\lambda$'s in an interval $[\lambda_l, \lambda_u]$, but the next section will need to run so many, and even more, optimization problems to calculate the confidence band by bootstrap samples. 

\section{Confidence Band}\label{sec:conf_band}

Another issue is how reliable the fitted monotone curve is, thus we shall investigate the confidence band of the fitted curve in this section.
Without the monotonicity constraint, the confidence band of smoothing splines $\hat f = \bfB\hat\gamma$ can be explicitly derived since both $\hat\gamma$ and its variance-covariance matrix $\hat\bSigma$ can be explicitly calculated (see Chapter 15.5 of \textcite{ramsayFunctionalDataAnalysis2005}). Then the standard error of a prediction $\hat y_0 = \bfb(x_0)^T\hat\gamma$ is $se(\hat y_0) = \bfb(x_0)^T\hat\bSigma\bfb(x_0)$.
It follows that the 95\% confidence interval can be estimated as
$$
(\hat y_0 - 1.96\times se(\hat y_0), \hat y_0 + 1.96\times se(\hat y_0))\,.
$$
Then the (point-wise) confidence band is formed as the set of the confidence interval at each $x$ point. 



With the monotonicity constraint, the estimation of confidence bands becomes more difficult. Fortunately, we can resort to the bootstrap approach to estimate the confidence bands. Generally, there are two types of bootstrap. One is the nonparametric bootstrap as summarized in Algorithm \ref{alg:ci_opt_nonpara}, where the essential is to generate bootstrap sample $\{\bfx^\star, \bfy^\star\}$ by sampling the original sample $\{\bfx, \bfy\}$ with replacement. Usually, it is time-consuming to conduct a nonparametric bootstrap. \textcite{shinGenerativeMultiplepurposeSampler2022} developed a Generative Bootstrap Sampler (GBS) to reduce the computational time by avoiding the repeated solving procedure for each bootstrap sample. However, the smoothing splines cannot fit into the GBS framework (see the discussion in the \supp), not to say our monotone splines with the monotonicity constraint.


\begin{algorithm}
    \caption{Confidence Band: Nonparametric Bootstrap}
    \label{alg:ci_opt_nonpara}
    \begin{algorithmic}[1]
    \REQUIRE Dataset $\bZ = \{(x_i, y_i)\}_{i=1}^n$, significance level $\alpha$.
    \FOR{$b$ from 1 to $B$}
    \STATE Sample $n$ points with replacement from $\bZ$, obtain bootstrap dataset $\bZ^\star$
    \STATE Fit a monotone B-spline $\hat f^{(b)}(x)$ for $\bZ^\star$.
    \ENDFOR
    \STATE Form a $1-\alpha$ pointwise confidence band from the percentiles at each $x$, i.e., the $\alpha/2$ and $1-\alpha/2$ quantiles of $\{\hat f^{(b)}(x)\}_{b=1}^B$.
    \end{algorithmic}
\end{algorithm}

Another type of bootstrap is the parametric bootstrap, as summarized in Algorithm \ref{alg:ci_opt}. Suppose we have obtained a fitting $\hat \bfy$ with error $\bfe = \bfy - \hat\bfy$, the bootstrap sample is constructed by
$(\bfx, \hat\bfy + \bfe^\star)$, where $\bfe^\star \sim N(0, \hat\sigma^2\bfI)$ and $\hat\sigma^2$ is the sample variance of $\bfe$.
\begin{algorithm}
    \caption{Confidence Band: Parametric Bootstrap}
    \label{alg:ci_opt}
    \begin{algorithmic}[1]
    \REQUIRE Dataset $\bZ = \{(x_i, y_i)\}_{i=1}^n$. Bootstrap repetitions $B$. Penalty parameter $\lambda$.
    \STATE Conduct Algorithm \ref{alg:point_opt} on $(\bfx, \bfy)$, and obtain $\hat\bfy$.
    \STATE Calculate $\bfe = \bfy- \hat\bfy$.
    \FOR{$b$ from 1 to $B$}
    \STATE Conduct Algorithm \ref{alg:point_opt} on parametric bootstrap sample $(\bfx, \hat\bfy+\bfe^\star)$, and obtain fitting $\hat f^{(b)}(x) = \bfb(x)^T\hat\gamma^{(b)}$
    \ENDFOR
    \STATE Construct $1-\alpha$ pointwise percentile bootstrap confidence band at each $x$, i.e., the $\alpha/2$ and $1-\alpha/2$ quantiles of $\{\hat f^{(b)}(x)\}_{b=1}^B$.
    \end{algorithmic}
\end{algorithm}

Note that both Algorithms \ref{alg:ci_opt_nonpara} and \ref{alg:ci_opt} require a specified $\lambda$, so we need to repeat those algorithms for each $\lambda$ if we want to investigate the effect of the penalty parameter $\lambda\in[\lambda_l,\lambda_u]$. On the other hand, the MLP generator in the previous section can be further developed to estimate the confidence bands. Recall that we can train an MLP generator $G(\lambda)$ to be the solution $\hat\gamma$ given $\{\bfx,\bfy\}$. Now we want to extend it to be $G(\bfy^\star, \lambda)$, which would (approximately) be the solution for a bootstrap sample $(\bfx, \hat\bfy +\bfe^\star)$ and $\lambda$. 

We replace the fixed $\bfy$ in Equation \eqref{eq:g_point_est} with the random $\hat\bfy(\lambda) +\bfe$, where the randomness comes from $\bfe\sim N(0, \hat\sigma^2\bfI)$ and $\hat\sigma^2$ is the sample variance of $\bfy-\hat\bfy(\lambda)$. Take the expectation $\bbE_{\bfe\mid \lambda}$ to integrate out $\bfe$,
\begin{equation}
\hat G = \argmin_{G\in \cG} \bbE_\lambda\bbE_{\bfe\mid \lambda} \left[\Vert \hat\bfy(\lambda) +\bfe - \bfB G(\hat\bfy(\lambda) +\bfe, \lambda)\Vert_2^2 + \lambda \Vert\bfL G(\hat\bfy(\lambda)+\bfe, \lambda)\Vert_2^2\right]\,.
\label{eq:g_ci_band}
\end{equation}
In practice, both expectations $\bbE_\lambda$ and $\bbE_{\bfe\mid \lambda}$ can be approximated by their Monte Carlo estimates, as summarized in Algorithm \ref{alg:ci_mlp}.
\begin{algorithm}
    \caption{Confidence Band: MLP Generator $G(\bfy, \lambda)$}
    \label{alg:ci_mlp}
    \begin{algorithmic}[1]
    \REQUIRE Dataset $\bZ = \{(x_i, y_i)\}_{i=1}^n$, batch size $M$, bootstrap repetition $B$, pre-trained $G(\bfy, \lambda)$ from Algorithm \ref{alg:point_mlp}.
    \COMMENT{\texttt{Training Step}}
    \STATE Calculate prediction $\hat\bfy(\lambda) = \bfB G(\bfy, \lambda)$, and the standard error $\hat\sigma(\lambda) = \text{std}(\bfy - \hat \bfy(\lambda))$.
    \WHILE{not converged}
    \FOR{$j$ from 1 to $M$}
    \STATE Sample $\lambda_j\sim U[\lambda_l, \lambda_u]$.
    \STATE Calculate the standard error $\hat\sigma(\lambda_j) = \text{std}(\bfy - \hat \bfy(\lambda_j))$.
    \STATE Sample $\{\bfe_i\}_{i=1}^M$ where each $\bfe_i\sim N(0, \hat\sigma^2\bfI)$.
    \STATE Evaluate 
    $$
    l_j = \frac 1M\sum_{i=1}^M\left[\Vert \hat\bfy(\lambda_j) +\bfe_i - \bfB G(\hat\bfy(\lambda_j) +\bfe_i, \lambda_j)\Vert_2^2 + \lambda_j \Vert\bfL G(\hat\bfy(\lambda_j)+\bfe_i, \lambda_j)\Vert_2^2\right]\,.
    $$
    \ENDFOR
    \STATE Define the loss $\cL = \frac 1M\sum_{j=1}^M l_j$.
    \STATE Update neural network $G$ to minimize $\cL$.
    \ENDWHILE
    \REQUIRE Set of penalty parameters $\Lambda$, significance level $\alpha$. \COMMENT{\texttt{Evaluation Step}}
    \FOR{$\lambda$ in $\Lambda$}
    \FOR{$b$ from 1 to $B$}
    \STATE Evaluate $\hat\gamma^{(r)} = G(\hat\bfy+\bfe^\star,\lambda)$, then the fitting spline is $\hat f^{(b)} =\bfb(x)^T\hat\gamma^{(b)}$.
    \ENDFOR 
    \STATE Construct $1-\alpha$ pointwise percentile bootstrap confidence band at each $x$, i.e., the $\alpha/2$ and $1-\alpha/2$ quantiles of $\{\hat f^{(b)}(x)\}_{b=1}^B$.
    \ENDFOR
    \end{algorithmic}
\end{algorithm}

To compare the confidence band estimated from the parametric bootstrap with the OPT approach (Algorithm~\ref{alg:ci_opt}) and the confidence band obtained from the MLP generator (Algorithm~\ref{alg:ci_mlp}), we consider the Jaccard index.
The Jaccard index measures the similarity between two finite sets $A$ and $B$. Specifically, it is defined as the size of the intersection divided by the size of the union,
$$
\cJ(A, B) = \frac{\vert A\cap B\vert}{\vert A\cup B\vert}\,.
$$
We propose the Jaccard index for two confidence intervals $\CI_1$ and $\CI_2$ to measure the similarity (overlap) 
$$
\cJ(\CI_1, \CI_2) = \frac{\vert \CI_1\cap \CI_2\vert}{\vert \CI_1\cup \CI_2\vert}\,,
$$
where the size of confidence interval $\vert \cdot\vert$ is defined as the width of the interval. Furthermore, we define the Jaccard index for two confidence bands $\CB_j = \{\CI_j(x_i),i=1,\ldots,n\}, j=1,2$ as the average of Jaccard index of confidence intervals at each $x_i$, 
$$
\cJ(\CB_1, \CB_2) = \frac 1n \sum_{i=1}^n \cJ(\CI_1(x_i), \CI_2(x_i)) = 
\frac{1}{n}
\sum_{i=1}^{n}
\frac{\vert \CI_1(x_i)\cap \CI_2(x_i)\vert}{\vert \CI_1(x_i) \cup \CI_2(x_i)\vert}\,.
$$
If $n$ is large enough and $x_i$ is uniformly distributed, then the Jaccard index for confidence bands can be interpreted as the proportion of the overlap area.


\begin{figure}
    \centering
    \begin{subfigure}{0.33\textwidth}
    \includegraphics[width=\textwidth]{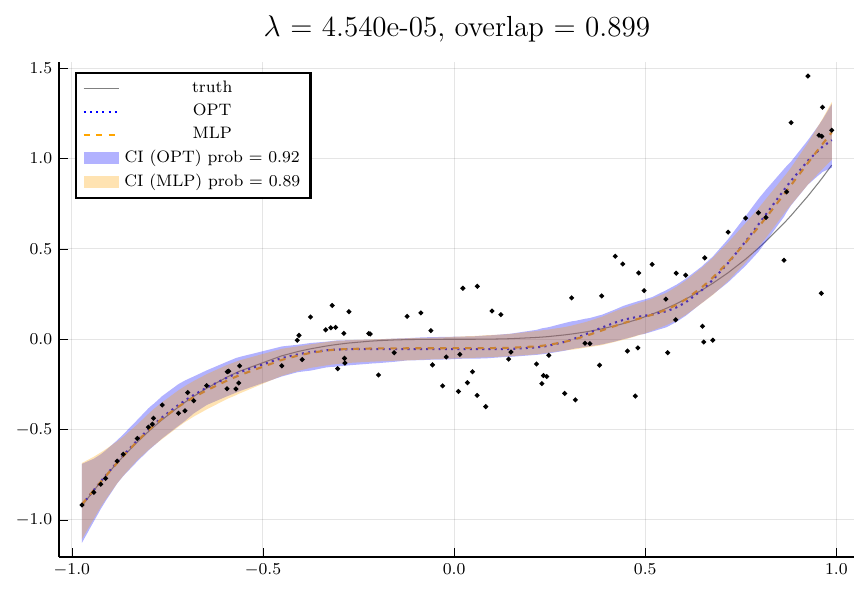}
    \end{subfigure}%
    \begin{subfigure}{0.33\textwidth}
    \includegraphics[width=\textwidth]{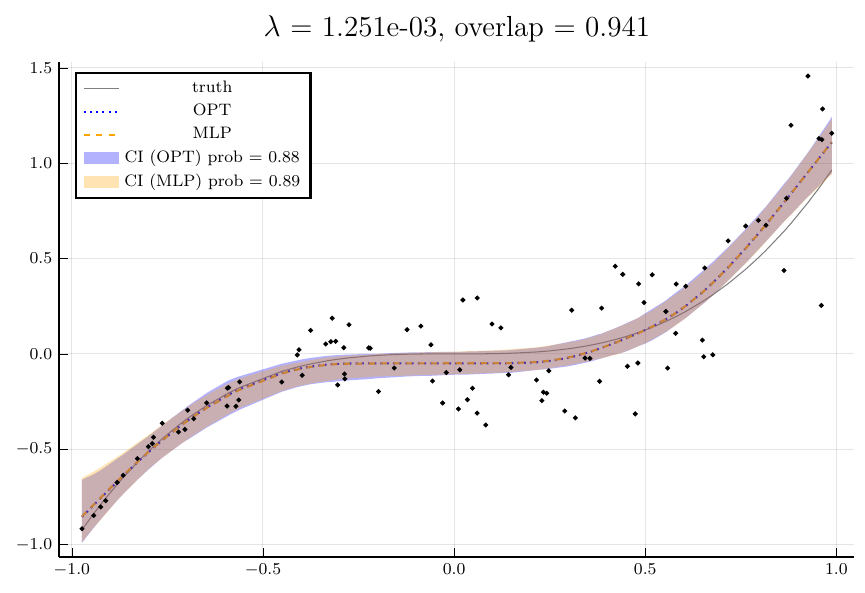}
    \end{subfigure}%
    \begin{subfigure}{0.33\textwidth}
    \includegraphics[width=\textwidth]{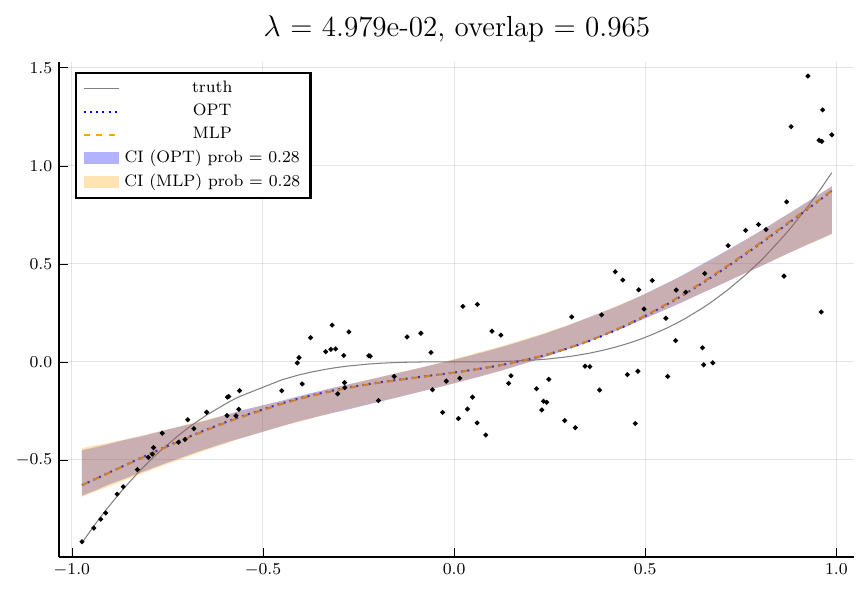}
    \end{subfigure}
    \caption[Demo of MLP Confidence Band cubic]{Demo of 95\% confidence bands of a cubic curve with noise $\sigma=0.2$ under three different penalty parameters. The blue shaded band denotes the confidence band by the OPT approach, and the orange shaded band represents the confidence band by the MLP generator. ``overlap'' in the title of each subfigure gives the average Jaccard index between these two confidence bands. In the legend, the ``prob'' is the coverage probability of the corresponding estimated confidence band. The fitted curve by the OPT approach and the MLP generator, together with the truth, are also shown.}
    \label{fig:demo_mlp_CI}
\end{figure}


Figure \ref{fig:demo_mlp_CI} displays the confidence bands by the OPT approach and the MLP generator for a cubic curve with noise $\sigma=0.2$ under three different penalty parameters. In each setting, although the fitness to the truth becomes worse along the increasing penalty parameter $\lambda$, two confidence bands always overlap quite well and all average Jaccard indexes are close to 1.0. Since the MLP generator can be viewed as an approximation for the OPT approach, Figure~\ref{fig:demo_mlp_CI} implies that the approximation of MLP generator for confidence bands is pretty good.

Now we perform repeated experiments to demonstrate the performance of the MLP generator.
Figure \ref{fig:jaccard_vs_lambda} shows the Jaccard index along the penalty parameter $\lambda$ for the same four curves investigated in Section~\ref{sec:two_alg}. The simulation data generating scheme is also the same, which is defined in Equation~\eqref{eq:mlp_data}.
\begin{figure}
    \centering
    \includegraphics[width=\textwidth]{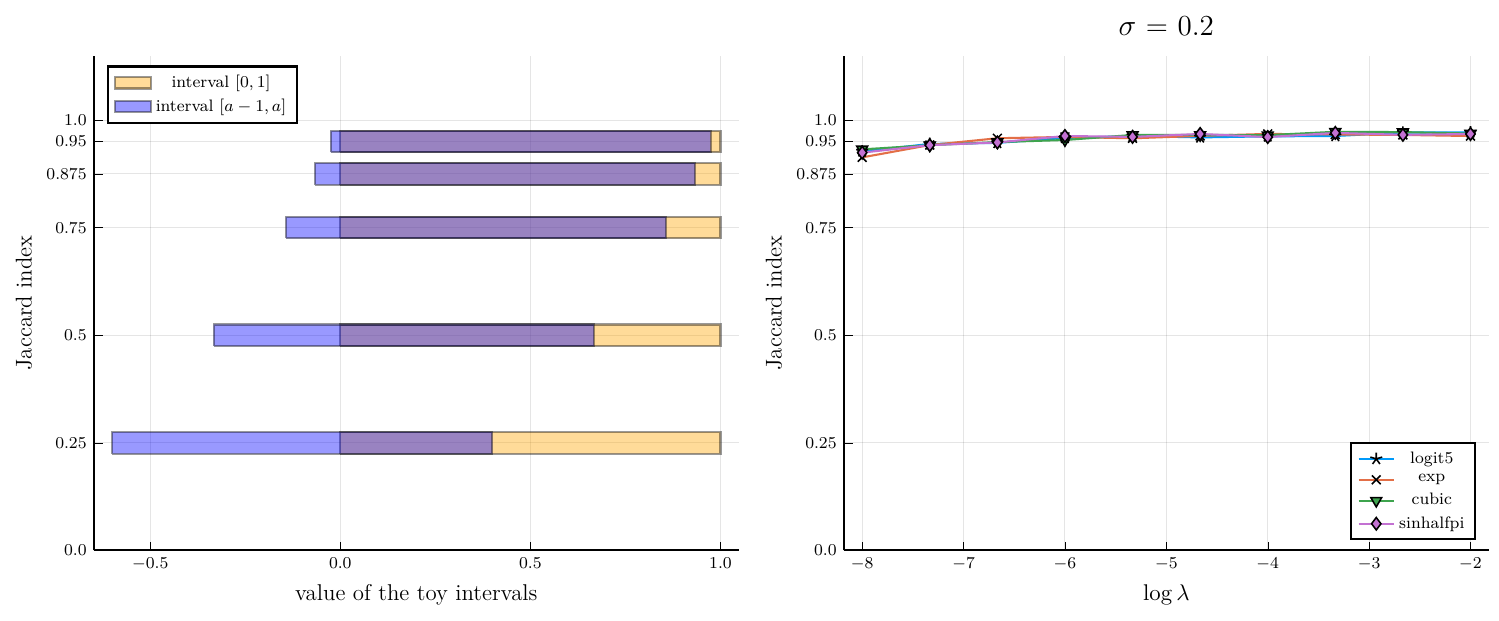}
    \caption{(Left) The visual representation of overlapped intervals at several Jaccard indexes. (Right) Average Jaccard index of confidence bands obtained from MLP generator and OPT approach along tuning parameters for four monotone curves among 5 repetitions. The ticks on the $y$-axis of the left subfigure and the right subfigure are kept on the same level.}
    \label{fig:jaccard_vs_lambda}
\end{figure}
The left panel visualizes the overlap extent of the Jaccard index using a toy example. Suppose we have two unit intervals, $[0, 1]$ and $[a-1, a], a\in [0, 1]$, denoted by the orange color and the blue color, respectively. Then the Jaccard index is $\frac{a}{2-a}\,.$
Intuitively, the Jaccard index larger than 0.9 can be considered an excellent overlap, and a larger than 0.75 would be a good overlap.
Although in practice, these two intervals are not always equal in size,
it provides insights into how close two confidence intervals are given a Jaccard index. The right panel shows the average Jaccard index along the penalty parameter $\lambda$ among five repetitions for four curves under the noise level $\sigma = 0.2$. Overall, all Jaccard indexes are larger than 0.875, and in most cases, they can achieve 0.95. It is slightly smaller for the small penalty parameter $\lambda$. Thus, the MLP generator can achieve pretty good approximations to the confidence band obtained by the OPT approach.

It is necessary to note that we are more concerned about the approximation accuracy of confidence bands instead of the coverage probability. Let $\hat\theta$ be a scalar point estimate and $q$ be the set of estimates based on bootstrap samples, and $q_\alpha$ be the $\alpha$ quantile of $q$. In addition to the classical percentile CI $(q_{0.025}, q_{0.975})$, there are many variants of bootstrap confidence intervals for better coverage probability, see more discussion in the \supp, where we also show the coincidence of those two confidence bands by comparing their coverage probability.
Although we focus on the classical percentile CI, the comparisons can be seamlessly moved to other bootstrap CIs.

The Jaccard index and coverage probability when $\sigma=0.1$ and $\sigma=0.5$ are displayed in Figures 
S1 and S2 of the \supp. 
Both show good coincidences between the OPT approach and the MLP generator.

Now we check the running time of the MLP generator $G(\bfy, \lambda)$ for the confidence band. Firstly, we ensure the training iterations are enough for convergence, and Figure \ref{fig:ci_fitness_ratio_cubic} reports the average fitness ratio and Jaccard index, both of which are good enough.
As in the previous point estimate of MLP generator $G(\lambda)$ in Figure~\ref{fig:run_time}, Figure \ref{fig:ci_run_time} shows that the running time of the training step of MLP does not increase along the sample size $n$, but the OPT approach would take a longer time for a larger sample size $n$. Note that the number of bootstrap samples is 2000. If we increase the number of bootstrap samples, say $2000k$, then it would roughly take $k$ times the running time for the OPT approach. While for the MLP generator, the training step is not affected by the number of bootstrap samples, and it only needs more evaluation time, which is much cheaper. So, the more bootstrap samples, the more time we can save.
\begin{figure}
    \centering
    \begin{subfigure}{0.5\textwidth}
        \includegraphics[width=\textwidth]{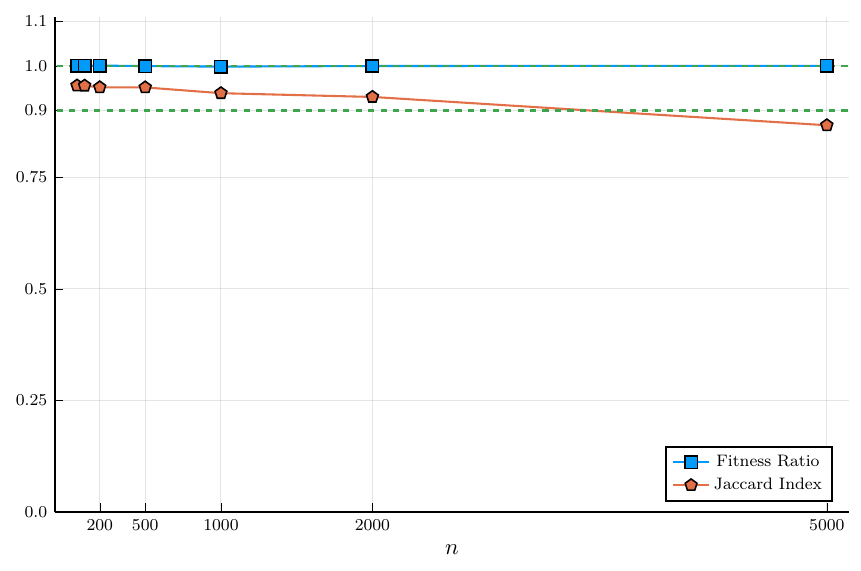}
        \caption{}
        \label{fig:ci_fitness_ratio_cubic}        
    \end{subfigure}%
    \begin{subfigure}{0.5\textwidth}
        \includegraphics[width=\textwidth]{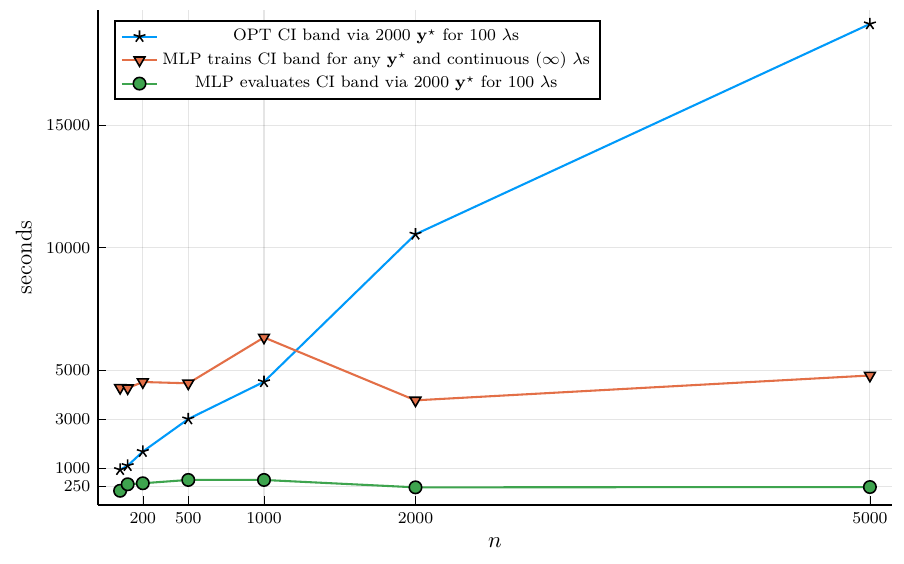}
        \caption{}
        \label{fig:ci_run_time}
    \end{subfigure}
    \caption{(a) Average fitness ratio and Jaccard index for the trained MLP generator to show that the number of iterations in training is enough. (b) Running time of the MLP generator (training and evaluation) and the OPT approach for the confidence band. }
\end{figure}

\section{Simulations}\label{sec:monobspl_sim}

In this section, we conduct several simulations to compare the performance of the proposed monotone splines with several competitors\footnote{The source code which can reproduce the results: \texttt{\url{https://github.com/szcf-weiya/MonotoneSplines.jl}}}:
\begin{itemize}
\item \textcite{heMonotoneBsplineSmoothing1998}'s monotone quadratic spline (MQS), with implementation in R package \texttt{COBS} given by \textcite{ngFastEfficientImplementation2007}.
\item Quadratic spline (QS): the corresponding unconstrained version of MQS.
\item Isotonic regression, using the implementation in R package \texttt{Iso::pava} by the Pool-Adjacent-Violators Algorithm (PAVA).
\item Two different strategies for combing the isotonic regression with smoothing techniques,
\begin{itemize}
  \item IS: isotonic regression followed by smoothing,
  \item SI: smoothing followed by isotonic regression,
\end{itemize}
which have been proved to be asymptotically equivalent in some sense \parencite{mammenEstimatingSmoothMonotone1991}.
\item The locally estimated scatterplot smoothing (LOESS), which will be used as the smoothing step in SI and IS.
\item Cubic spline (CS): the ordinary cubic spline, where the number of knots is selected by 2-fold cross-validation. 
\item Monotone cubic spline (MCS): the number of knots is set as the same of the corresponding cubic spline.
\item Smoothing spline (SS): the smoothness penalty parameter is determined by the generalized cross-validation (GCV) principle.
\item Monotone smoothing spline (MSS): the smoothness penalty parameter is the same as the one in the corresponding smoothing spline.
\item \textcite{murrayFastFlexibleMethods2016}'s monotone polynomial fitting (MonoPoly).
\item \textcite{cannonMonmlpMultilayerPerceptron2017}'s Monotone Multi-Layer Perceptron (MONMLP) with two hidden layers, where 32 nodes in the first hidden layer and 2 nodes in the second hidden layer. 
\item \textcite{navarro-garciaConstrainedSmoothingOutofrange2023}'s constrained penalized splines (\texttt{cpsplines}).
\item \textcite{groeneboomConfidenceIntervalsMonotone2023}'s smoothed least squares estimator (SLSE), which can also be viewed as a kind of IS (isotonic regression followed by smoothing), but it adopted the kernel smoothing technique.
\end{itemize}

We consider the following five types of monotone curves, where the logistic and growth curves were used in \textcite{heMonotoneBsplineSmoothing1998}, and the error function curve was illustrated in \textcite{pappOptimizationModelsShapeconstrained2011}, \textcite{pappShapeConstrainedEstimationUsing2014} and \textcite{navarro-garciaConstrainedSmoothingOutofrange2023}.
\begin{itemize}
  \item Logistic curve: $f(x) = e^x/(1+e^x), x\in (-5, 5)$.
  \item Growth curve: $f(x) = 1/(1-0.42\log(x)), x\in(0, 10)$.
  \item Cubic (polynomial) curve: $f(x) = x^3, x\in (-1,1)$.
  \item Step curve: $f(x) = \sum_{i=1}^TI(x > t_i), x\in(-1,1)$ with $T$ random points $t_i\in (-1, 1)$.
  \item Error function curve: $f(x) = 5+\sum_{i=1}^4\mathrm{erf}(15i(x-i/5)), x\in [0, 1]$, where $\mathrm{erf}(x) = \frac{2}{\sqrt\pi}\int_0^xe^{-t^2}dt$ is the error function.
\end{itemize}

For each curve, we generate $n=100$ points $\{(x_i,y_i),i=1,\ldots,n\}$ from
$$
y_i = f(x_i) + \varepsilon_i\,,
$$
where $\varepsilon_i$'s are independently sampled from $N(0, \sigma^2)$. 

Figure~\ref{fig:demo_logit_erf} provides two demo figures to illustrate the behaviors of selected methods on the logistic and error function curves. Some methods might overfit and return a wiggly fitting, such as the cubic spline on the logistic curve in Figure~\ref{fig:demo_logit}; some methods result in less smooth and step-like fitting, like the isotonic regression in Figure~\ref{fig:demo_logit}; others might underfit by imposing too strong constraints, like LOESS on the error curve function in Figure~\ref{fig:demo_erf}. Roughly, most fitting curves are relatively close to the truth. 
\begin{figure}
    \centering
    \begin{subfigure}{0.5\textwidth}
    \includegraphics[width=\textwidth]{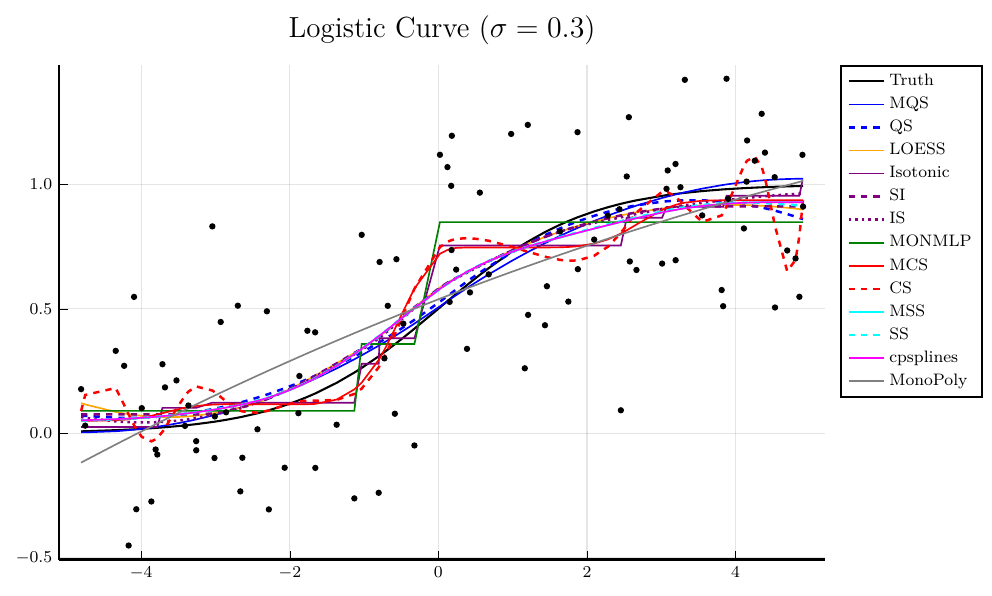}
    \caption{}
    \label{fig:demo_logit}
    \end{subfigure}%
    \begin{subfigure}{0.5\textwidth}
    \includegraphics[width=\textwidth]{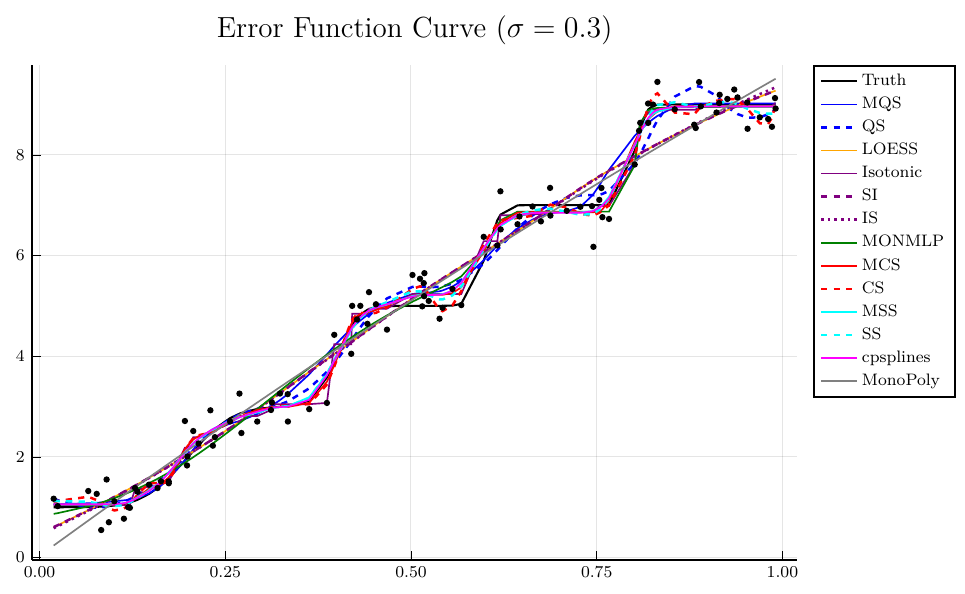}
    \caption{}
    \label{fig:demo_erf}
    \end{subfigure}
    \caption[Demo of Fitting on Logistic and Growth Curve]{Demo of fitting of different methods on the logistic and error function curves.}
    \label{fig:demo_logit_erf}
\end{figure}

To comprehensively compare their performance, we conduct repeated experiments to measure their average performance.
We adopt the $L_p, p\in \{1,2, \infty\}$ distance between the predicted values $\hat f(x_i)$ and the underlying true values $f(x_i)$,
$$
L_p =
\begin{cases}
    \left(\sum_{i=1}^n\vert \hat f(x_i) - f(x_i)\vert^p\right)^{1/p} & p = 1,2\\
    \max_{i=1}^n\vert \hat f(x_i) - f(x_i)\vert & p =\infty
\end{cases}\,.
$$
The $L_p$ distances have been used in \textcite{pappOptimizationModelsShapeconstrained2011} and \textcite{navarro-garciaConstrainedSmoothingOutofrange2023}, and the scaled $\frac{1}{\sqrt{n}}L_2$ is equivalent to the root mean squared error (RMSE) in \textcite{heMonotoneBsplineSmoothing1998}.

Based on 100 repetitive experiments, Table~\ref {tab:erf} reports the mean (scaled) $L_p, p\in\{1,2,\infty\}$ distances, together with the standard errors, on the error function curve. To account for the randomness in experiments, both the smallest one and the ones whose errors are no more than one standard error above the error of the smallest one (referred to as \emph{one-standard-error range}) are highlighted in bold. The rank of the mean distances are also noted as superscripts. Our proposed MSS and \textcite{navarro-garciaConstrainedSmoothingOutofrange2023}'s \texttt{cpsplines} are quite close and outperform others for all distances. When $\sigma = 0.15$, MSS ranks first for all distances, but \texttt{cpsplines} is within the one-standard-error range of MSS; on the other hand, when $\sigma=0.3$, \texttt{cpsplines} ranks first for $L_1$ and $L_2$ distances, but MSS is within the one-standard-error range of \texttt{cpsplines}, while MSS is the top-1 in $L_\infty$ and \texttt{cpsplines} is still within the one-standard-error range of MSS. MCS is also not bad, whose ranks are always within the top 6.

\begin{table}
  \caption{Average (scaled) $L_p$ distances, $p\in \{1,2,\infty\}$, over 100 experiments on the error function curve, together with the standard error of the average in parentheses. Both the smallest one and the ones whose errors are no more than one standard error above the error of the smallest one are highlighted in bold. The superscripts indicate the rank of methods.}
  \label{tab:erf}
\resizebox{\textwidth}{!}{
  \begin{tabular}{ccccc}
\toprule
Noise $\sigma$ & Method&$\frac 1n L_1$&$\frac{1}{\sqrt n}L_2$&$L_\infty$\tabularnewline
\midrule
\multirow{14}{*}{0.15}&Cubic Spline (CS)& 8.53e-02 (3.0e-03)\textsuperscript{6}& 1.09e-01 (4.0e-03)\textsuperscript{5}& 3.29e-01 (1.2e-02)\textsuperscript{5}\tabularnewline
&Monotone CS (MCS)& 7.95e-02 (4.9e-03)\textsuperscript{5}& 1.09e-01 (6.5e-03)\textsuperscript{6}& 3.88e-01 (1.8e-02)\textsuperscript{6}\tabularnewline
&Smoothing Spline (SS)& 6.26e-02 (9.0e-04)\textsuperscript{3}& 8.03e-02 (1.1e-03)\textsuperscript{3}& \textbf{2.47e-01} (6.1e-03)\textsuperscript{2}\tabularnewline
&Montone SS (MSS)& \textbf{5.27e-02} (8.9e-04)\textsuperscript{1}& \textbf{7.02e-02} (1.1e-03)\textsuperscript{1}& \textbf{2.42e-01} (7.5e-03)\textsuperscript{1}\tabularnewline
&Quadratic Spline (QS)& 1.20e-01 (2.9e-03)\textsuperscript{9}& 1.65e-01 (3.9e-03)\textsuperscript{9}& 5.52e-01 (1.5e-02)\textsuperscript{8}\tabularnewline
&\textcite{heMonotoneBsplineSmoothing1998}: MQS& 9.59e-02 (1.8e-03)\textsuperscript{7}& 1.45e-01 (2.8e-03)\textsuperscript{7}& 5.94e-01 (1.8e-02)\textsuperscript{9}\tabularnewline
&LOESS& 2.97e-01 (2.1e-03)\textsuperscript{10}& 3.54e-01 (2.1e-03)\textsuperscript{10}& 7.86e-01 (5.1e-03)\textsuperscript{11}\tabularnewline
&Isotonic& 6.33e-02 (8.2e-04)\textsuperscript{4}& 8.40e-02 (1.0e-03)\textsuperscript{4}& 2.85e-01 (6.2e-03)\textsuperscript{4}\tabularnewline
&\textcite{mammenEstimatingSmoothMonotone1991}: SI (LOESS+Isotonic)& 2.97e-01 (2.1e-03)\textsuperscript{11}& 3.54e-01 (2.1e-03)\textsuperscript{11}& 7.86e-01 (5.1e-03)\textsuperscript{12}\tabularnewline
&\textcite{mammenEstimatingSmoothMonotone1991}: IS (Isotonic+LOESS)& 2.97e-01 (2.1e-03)\textsuperscript{12}& 3.54e-01 (2.1e-03)\textsuperscript{12}& 7.84e-01 (5.0e-03)\textsuperscript{10}\tabularnewline
&\textcite{murrayFastFlexibleMethods2016}: MonoPoly& 3.25e-01 (2.1e-03)\textsuperscript{13}& 3.89e-01 (2.1e-03)\textsuperscript{13}& 8.57e-01 (7.0e-03)\textsuperscript{13}\tabularnewline
&\textcite{cannonMonmlpMultilayerPerceptron2017}: MONMLP& 1.20e-01 (8.3e-03)\textsuperscript{8}& 1.60e-01 (1.0e-02)\textsuperscript{8}& 4.76e-01 (2.6e-02)\textsuperscript{7}\tabularnewline
&\textcite{navarro-garciaConstrainedSmoothingOutofrange2023}: cpsplines& \textbf{5.32e-02} (8.2e-04)\textsuperscript{2}& 7.15e-02 (1.1e-03)\textsuperscript{2}& \textbf{2.47e-01} (6.6e-03)\textsuperscript{3}\tabularnewline
&\textcite{groeneboomConfidenceIntervalsMonotone2023}: SLSE& 3.33e-01 (2.3e-03)\textsuperscript{14}& 4.05e-01 (2.6e-03)\textsuperscript{14}& 1.03e+00 (1.3e-02)\textsuperscript{14}\tabularnewline
\midrule
\multirow{14}{*}{0.3}&Cubic Spline (CS)& 1.47e-01 (3.7e-03)\textsuperscript{6}& 1.87e-01 (4.6e-03)\textsuperscript{6}& 5.50e-01 (1.5e-02)\textsuperscript{6}\tabularnewline
&Monotone CS (MCS)& 1.25e-01 (5.1e-03)\textsuperscript{5}& 1.65e-01 (6.1e-03)\textsuperscript{5}& 5.11e-01 (1.7e-02)\textsuperscript{4}\tabularnewline
&Smoothing Spline (SS)& 1.13e-01 (1.9e-03)\textsuperscript{4}& 1.43e-01 (2.1e-03)\textsuperscript{3}& 4.13e-01 (1.2e-02)\textsuperscript{3}\tabularnewline
&Montone SS (MSS)& \textbf{9.78e-02} (1.8e-03)\textsuperscript{2}& \textbf{1.26e-01} (1.9e-03)\textsuperscript{2}& \textbf{3.94e-01} (1.1e-02)\textsuperscript{1}\tabularnewline
&Quadratic Spline (QS)& 1.73e-01 (3.8e-03)\textsuperscript{9}& 2.24e-01 (5.0e-03)\textsuperscript{9}& 6.49e-01 (1.6e-02)\textsuperscript{9}\tabularnewline
&\textcite{heMonotoneBsplineSmoothing1998}: MQS& 1.50e-01 (3.2e-03)\textsuperscript{7}& 2.00e-01 (3.8e-03)\textsuperscript{7}& 6.42e-01 (1.6e-02)\textsuperscript{8}\tabularnewline
&LOESS& 3.00e-01 (2.1e-03)\textsuperscript{10}& 3.59e-01 (2.0e-03)\textsuperscript{10}& 7.89e-01 (5.5e-03)\textsuperscript{11}\tabularnewline
&Isotonic& 1.12e-01 (1.7e-03)\textsuperscript{3}& 1.49e-01 (1.8e-03)\textsuperscript{4}& 5.22e-01 (1.2e-02)\textsuperscript{5}\tabularnewline
&\textcite{mammenEstimatingSmoothMonotone1991}: SI (LOESS+Isotonic)& 3.00e-01 (2.1e-03)\textsuperscript{11}& 3.59e-01 (2.0e-03)\textsuperscript{11}& 7.89e-01 (5.5e-03)\textsuperscript{12}\tabularnewline
&\textcite{mammenEstimatingSmoothMonotone1991}: IS (Isotonic+LOESS)& 3.01e-01 (2.0e-03)\textsuperscript{12}& 3.60e-01 (2.1e-03)\textsuperscript{12}& 7.87e-01 (5.3e-03)\textsuperscript{10}\tabularnewline
&\textcite{murrayFastFlexibleMethods2016}: MonoPoly& 3.34e-01 (2.2e-03)\textsuperscript{13}& 3.99e-01 (2.1e-03)\textsuperscript{13}& 8.69e-01 (6.6e-03)\textsuperscript{13}\tabularnewline
&\textcite{cannonMonmlpMultilayerPerceptron2017}: MONMLP& 1.65e-01 (7.9e-03)\textsuperscript{8}& 2.14e-01 (9.4e-03)\textsuperscript{8}& 6.33e-01 (1.9e-02)\textsuperscript{7}\tabularnewline
&\textcite{navarro-garciaConstrainedSmoothingOutofrange2023}: cpsplines& \textbf{9.76e-02} (1.8e-03)\textsuperscript{1}& \textbf{1.26e-01} (2.0e-03)\textsuperscript{1}& \textbf{3.96e-01} (1.1e-02)\textsuperscript{2}\tabularnewline
&\textcite{groeneboomConfidenceIntervalsMonotone2023}: SLSE& 3.36e-01 (3.3e-03)\textsuperscript{14}& 4.12e-01 (4.1e-03)\textsuperscript{14}& 1.06e+00 (1.9e-02)\textsuperscript{14}\tabularnewline
\bottomrule
\end{tabular}
}
\end{table}

For the logistic curve, the results are summarized in Table~\ref{tab:logit}. Our proposed MSS outperforms others for all distances when the noise is relatively large $\sigma=1.0, 1.5$. When the noise level is small $\sigma = 0.1$, \textcite{mammenEstimatingSmoothMonotone1991}'s SI outperforms others, but MSS is within the top 4 for all distances. The proposed MCS is also not bad, which is always better than the unconstrained cubic spline (CS), and it is within the top 6 in most cases.

\begin{table}
  \caption{Average (scaled) $L_p$ distances, $p\in \{1,2,\infty\}$, over 100 experiments on the logistic curve, together with the standard error of the average in parentheses. Both the smallest one and the ones whose errors are no more than one standard error above the error of the smallest one are highlighted in bold. The superscripts indicate the rank of methods.}
  \label{tab:logit}
\resizebox{\textwidth}{!}{
  \begin{tabular}{ccccc}
\toprule
Noise $\sigma$ & Method&$\frac 1n L_1$&$\frac{1}{\sqrt n}L_2$&$L_\infty$\tabularnewline
\midrule
\multirow{14}{*}{0.1}&Cubic Spline (CS)& 2.32e-02 (6.8e-04)\textsuperscript{9}& 2.96e-02 (8.8e-04)\textsuperscript{9}& 8.04e-02 (4.3e-03)\textsuperscript{11}\tabularnewline
&Monotone CS (MCS)& 2.13e-02 (5.7e-04)\textsuperscript{6}& 2.64e-02 (7.4e-04)\textsuperscript{6}& 6.20e-02 (3.8e-03)\textsuperscript{6}\tabularnewline
&Smoothing Spline (SS)& 2.11e-02 (6.1e-04)\textsuperscript{5}& 2.61e-02 (7.6e-04)\textsuperscript{5}& 6.12e-02 (3.0e-03)\textsuperscript{5}\tabularnewline
&Montone SS (MSS)& 2.01e-02 (5.5e-04)\textsuperscript{3}& 2.49e-02 (6.6e-04)\textsuperscript{3}& 5.74e-02 (2.7e-03)\textsuperscript{4}\tabularnewline
&Quadratic Spline (QS)& 2.23e-02 (6.4e-04)\textsuperscript{7}& 2.73e-02 (7.2e-04)\textsuperscript{7}& 6.56e-02 (2.5e-03)\textsuperscript{8}\tabularnewline
&\textcite{heMonotoneBsplineSmoothing1998}: MQS& 2.44e-02 (7.1e-04)\textsuperscript{10}& 2.98e-02 (8.5e-04)\textsuperscript{10}& 6.46e-02 (3.0e-03)\textsuperscript{7}\tabularnewline
&LOESS& 1.99e-02 (5.3e-04)\textsuperscript{2}& 2.44e-02 (6.1e-04)\textsuperscript{2}& 5.60e-02 (1.9e-03)\textsuperscript{3}\tabularnewline
&Isotonic& 3.14e-02 (4.8e-04)\textsuperscript{13}& 4.11e-02 (5.8e-04)\textsuperscript{13}& 1.28e-01 (4.4e-03)\textsuperscript{13}\tabularnewline
&\textcite{mammenEstimatingSmoothMonotone1991}: SI (LOESS+Isotonic)& \textbf{1.93e-02} (5.5e-04)\textsuperscript{1}& \textbf{2.34e-02} (6.3e-04)\textsuperscript{1}& \textbf{4.82e-02} (1.4e-03)\textsuperscript{1}\tabularnewline
&\textcite{mammenEstimatingSmoothMonotone1991}: IS (Isotonic+LOESS)& 2.10e-02 (5.6e-04)\textsuperscript{4}& 2.55e-02 (6.5e-04)\textsuperscript{4}& 5.40e-02 (1.6e-03)\textsuperscript{2}\tabularnewline
&\textcite{murrayFastFlexibleMethods2016}: MonoPoly& 7.83e-02 (3.6e-04)\textsuperscript{14}& 8.91e-02 (3.4e-04)\textsuperscript{14}& 1.76e-01 (2.7e-03)\textsuperscript{14}\tabularnewline
&\textcite{cannonMonmlpMultilayerPerceptron2017}: MONMLP& 2.59e-02 (7.7e-04)\textsuperscript{12}& 3.34e-02 (1.0e-03)\textsuperscript{12}& 9.20e-02 (4.7e-03)\textsuperscript{12}\tabularnewline
&\textcite{navarro-garciaConstrainedSmoothingOutofrange2023}: cpsplines& 2.48e-02 (4.8e-04)\textsuperscript{11}& 3.14e-02 (5.8e-04)\textsuperscript{11}& 8.02e-02 (2.8e-03)\textsuperscript{10}\tabularnewline
&\textcite{groeneboomConfidenceIntervalsMonotone2023}: SLSE& 2.28e-02 (7.6e-04)\textsuperscript{8}& 2.82e-02 (9.1e-04)\textsuperscript{8}& 6.86e-02 (2.6e-03)\textsuperscript{9}\tabularnewline
\midrule
\multirow{14}{*}{1.0}&Cubic Spline (CS)& 2.34e-01 (1.2e-02)\textsuperscript{14}& 2.90e-01 (1.5e-02)\textsuperscript{14}& 8.04e-01 (6.0e-02)\textsuperscript{13}\tabularnewline
&Monotone CS (MCS)& 1.61e-01 (5.8e-03)\textsuperscript{6}& 1.96e-01 (6.7e-03)\textsuperscript{5}& 4.76e-01 (3.1e-02)\textsuperscript{9}\tabularnewline
&Smoothing Spline (SS)& 1.58e-01 (5.9e-03)\textsuperscript{4}& 1.92e-01 (7.1e-03)\textsuperscript{4}& 4.18e-01 (2.2e-02)\textsuperscript{4}\tabularnewline
&Montone SS (MSS)& \textbf{1.44e-01} (5.1e-03)\textsuperscript{1}& \textbf{1.74e-01} (5.7e-03)\textsuperscript{1}& \textbf{3.59e-01} (1.6e-02)\textsuperscript{1}\tabularnewline
&Quadratic Spline (QS)& 1.63e-01 (5.3e-03)\textsuperscript{7}& 2.01e-01 (6.1e-03)\textsuperscript{7}& 4.44e-01 (1.7e-02)\textsuperscript{5}\tabularnewline
&\textcite{heMonotoneBsplineSmoothing1998}: MQS& 1.69e-01 (6.4e-03)\textsuperscript{9}& 2.08e-01 (7.2e-03)\textsuperscript{9}& 4.45e-01 (2.0e-02)\textsuperscript{6}\tabularnewline
&LOESS& 1.72e-01 (5.3e-03)\textsuperscript{11}& 2.12e-01 (6.5e-03)\textsuperscript{10}& 4.83e-01 (2.2e-02)\textsuperscript{10}\tabularnewline
&Isotonic& 1.89e-01 (5.6e-03)\textsuperscript{13}& 2.50e-01 (7.2e-03)\textsuperscript{13}& 8.83e-01 (4.8e-02)\textsuperscript{14}\tabularnewline
&\textcite{mammenEstimatingSmoothMonotone1991}: SI (LOESS+Isotonic)& 1.57e-01 (5.6e-03)\textsuperscript{3}& 1.88e-01 (6.5e-03)\textsuperscript{3}& 3.84e-01 (1.7e-02)\textsuperscript{2}\tabularnewline
&\textcite{mammenEstimatingSmoothMonotone1991}: IS (Isotonic+LOESS)& 1.65e-01 (5.7e-03)\textsuperscript{8}& 2.04e-01 (6.8e-03)\textsuperscript{8}& 4.72e-01 (2.1e-02)\textsuperscript{8}\tabularnewline
&\textcite{murrayFastFlexibleMethods2016}: MonoPoly& 1.61e-01 (5.2e-03)\textsuperscript{5}& 1.98e-01 (5.9e-03)\textsuperscript{6}& 4.47e-01 (1.9e-02)\textsuperscript{7}\tabularnewline
&\textcite{cannonMonmlpMultilayerPerceptron2017}: MONMLP& 1.83e-01 (5.4e-03)\textsuperscript{12}& 2.29e-01 (6.8e-03)\textsuperscript{12}& 5.54e-01 (2.9e-02)\textsuperscript{12}\tabularnewline
&\textcite{navarro-garciaConstrainedSmoothingOutofrange2023}: cpsplines& 1.53e-01 (8.3e-03)\textsuperscript{2}& 1.85e-01 (9.8e-03)\textsuperscript{2}& 3.92e-01 (3.2e-02)\textsuperscript{3}\tabularnewline
&\textcite{groeneboomConfidenceIntervalsMonotone2023}: SLSE& 1.69e-01 (6.3e-03)\textsuperscript{10}& 2.13e-01 (8.3e-03)\textsuperscript{11}& 5.45e-01 (3.1e-02)\textsuperscript{11}\tabularnewline
\midrule
\multirow{14}{*}{1.5}&Cubic Spline (CS)& 3.52e-01 (1.7e-02)\textsuperscript{14}& 4.39e-01 (2.2e-02)\textsuperscript{14}& 1.16e+00 (8.0e-02)\textsuperscript{13}\tabularnewline
&Monotone CS (MCS)& 2.22e-01 (7.5e-03)\textsuperscript{5}& 2.73e-01 (9.6e-03)\textsuperscript{5}& 7.44e-01 (5.3e-02)\textsuperscript{12}\tabularnewline
&Smoothing Spline (SS)& 2.40e-01 (1.2e-02)\textsuperscript{11}& 2.92e-01 (1.4e-02)\textsuperscript{10}& 6.31e-01 (4.1e-02)\textsuperscript{7}\tabularnewline
&Montone SS (MSS)& \textbf{2.04e-01} (7.9e-03)\textsuperscript{1}& \textbf{2.46e-01} (9.4e-03)\textsuperscript{1}& \textbf{5.15e-01} (2.9e-02)\textsuperscript{1}\tabularnewline
&Quadratic Spline (QS)& \textbf{2.11e-01} (7.1e-03)\textsuperscript{2}& \textbf{2.54e-01} (8.5e-03)\textsuperscript{2}& \textbf{5.40e-01} (2.5e-02)\textsuperscript{2}\tabularnewline
&\textcite{heMonotoneBsplineSmoothing1998}: MQS& 2.34e-01 (1.0e-02)\textsuperscript{9}& 2.85e-01 (1.2e-02)\textsuperscript{8}& 5.90e-01 (3.3e-02)\textsuperscript{5}\tabularnewline
&LOESS& 2.46e-01 (8.1e-03)\textsuperscript{12}& 2.98e-01 (9.9e-03)\textsuperscript{11}& 6.63e-01 (3.1e-02)\textsuperscript{9}\tabularnewline
&Isotonic& 2.74e-01 (7.8e-03)\textsuperscript{13}& 3.74e-01 (1.1e-02)\textsuperscript{13}& 1.49e+00 (8.1e-02)\textsuperscript{14}\tabularnewline
&\textcite{mammenEstimatingSmoothMonotone1991}: SI (LOESS+Isotonic)& 2.17e-01 (7.1e-03)\textsuperscript{4}& 2.59e-01 (8.5e-03)\textsuperscript{4}& 5.47e-01 (2.8e-02)\textsuperscript{3}\tabularnewline
&\textcite{mammenEstimatingSmoothMonotone1991}: IS (Isotonic+LOESS)& 2.39e-01 (7.9e-03)\textsuperscript{10}& 3.00e-01 (1.0e-02)\textsuperscript{12}& 7.34e-01 (3.6e-02)\textsuperscript{11}\tabularnewline
&\textcite{murrayFastFlexibleMethods2016}: MonoPoly& \textbf{2.12e-01} (6.3e-03)\textsuperscript{3}& 2.57e-01 (7.8e-03)\textsuperscript{3}& 5.73e-01 (2.8e-02)\textsuperscript{4}\tabularnewline
&\textcite{cannonMonmlpMultilayerPerceptron2017}: MONMLP& 2.33e-01 (7.5e-03)\textsuperscript{8}& 2.82e-01 (8.3e-03)\textsuperscript{7}& 6.35e-01 (3.5e-02)\textsuperscript{8}\tabularnewline
&\textcite{navarro-garciaConstrainedSmoothingOutofrange2023}: cpsplines& 2.29e-01 (1.7e-02)\textsuperscript{6}& 2.79e-01 (2.0e-02)\textsuperscript{6}& 6.17e-01 (6.1e-02)\textsuperscript{6}\tabularnewline
&\textcite{groeneboomConfidenceIntervalsMonotone2023}: SLSE& 2.32e-01 (8.3e-03)\textsuperscript{7}& 2.87e-01 (1.1e-02)\textsuperscript{9}& 7.16e-01 (4.4e-02)\textsuperscript{10}\tabularnewline
\bottomrule
\end{tabular}
}
\end{table}



The results for other three curves (cubic, step, and growth) are given in Table S1-S3 of the \supp. The findings are quite similar: our proposed approaches might not be the best in all cases, but they always demonstrate comparable performance with top ranks, particularly when the noise level is relatively large.

\section{Real Data Application}\label{sec:app}
If the Big Bang is considered the starting point of the universe as a whole, star formation can be seen as the opposite process -- the current contraction of gas on a local scale. Both, however, remain enigmatic \parencite{liMagneticFieldsMolecular2021}. It has been proposed that the alignment between magnetic fields (B-field) and the gas angular momentum (AM), BAM alignment, may be a significant factor in star formation \parencite{wangPolarizationHolesIndicator2023}. However, BAM alignment is not directly observable. We \parencite{wangPolarizationHolesIndicator2023} recently proposed that different BAM alignments will result in distinct B-field morphologies, from which we can determine whether nature has any preference for BAM alignment. Using numerical magnetohydrodynamic simulations, we have shown that the closer the BAM alignment is to parallelism, the faster the B-field dispersion increases with density. 

B-field morphologies, again, cannot be directly observed, but their orientation dispersion can. The observation tool used to investigate B-field dispersion is the polarization fraction ($P$) of the thermal dust emission. The higher the dispersion, the lower the $P$. In other words, $P$ will decrease with increasing density ($N$), a phenomenon known as ``polarization holes'' \parencite{liLinkMagneticFields2014a}. The closer the BAM alignment is to being parallel, the ``deeper'' the polarization hole \parencite{wangPolarizationHolesIndicator2023}. An example of a polarization hole is given in Figure~\ref{fig:ph}; the holes' monotonically decreasing nature can be modeled with our proposed monotone B-spline approach.

Given observed cloud data via telescopes, one specific task is to unveil the mystery of polarization holes by constructing proper astrophysical models. For an astrophysical model, we can generate different polarization-hole patterns under different parameters, then determine the best model parameters by finding the closest generated polarization-hole pattern to the observed cloud data. Specifically, Figure~\ref{fig:ph} displayed 256 paired points $\{(N_i, P_i)\}_{i=1}^n$ for one polarization-hole pattern
and the fittings by our proposed monotone splines and the corresponding unconstrained splines. Both monotone fitting approaches can avoid wiggly overfitting of their corresponding unconstrained splines, and they give a more robust summarization of the paired data, which is quite important in the downstream analysis for determining the best polarization-hole pattern from many patterns. 
\begin{figure}
    \centering
    \includegraphics{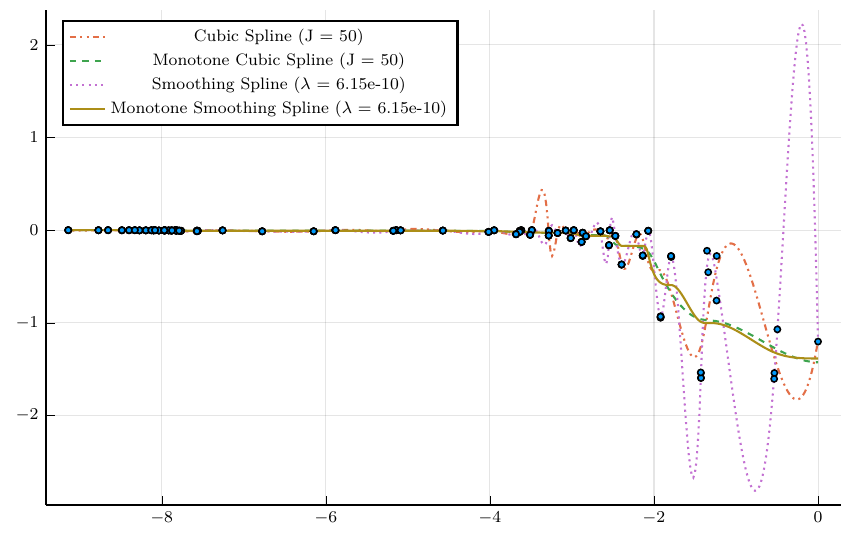}
    \caption{Fittings on a polarization-hole pattern. The $x,y$-axes denote the cloud column density $(N)$ and the polarization fraction ($P$) in the logarithm scale, respectively. The monotone splines and its corresponding unconstrained splines use the same parameters. The number of basis functions $J$ is selected to be large enough to summarize the relationship. The penalty parameter $\lambda$ for the smoothing spline is the default value selected by the generalized cross-validation principle from the function \texttt{smooth.spline} in R software.}
    \label{fig:ph}
\end{figure}

\section{Discussions}\label{sec:monobspl_discuss}

We propose monotone splines, including monotone cubic B-splines and monotone smoothing splines, by imposing the monotonic coefficients constraint. This constraint is a sufficient but not necessary condition for the splines to be monotonic. We discuss different conditions for the monotonicity, and investigate the estimation error and characterize the solutions.
To fit the monotone splines, we propose the MLP generator as an alternative to existing optimization toolboxes. The MLP generator can help save time in bootstrap tasks.

Extending the splines for univariate data to multidimensional data is a future potential research direction. For multidimensional data, the tensor product basis functions are defined for multidimensional splines \parencite{hastieElementsStatisticalLearning2009}, but the dimension of the basis grows exponentially, which is a manifestation of the curse of dimensionality. For computational and conceptual complexity, there are some restricted classes of multidimensional splines, such as additive splines. Specifically, the additive splines assume $f\in\IR^d$ has the form $f(X)=\alpha+f_1(X_1)+\cdots +f_d(X_d)$, where the functions $f_j$ are univariate splines. On the other hand, \textcite{dengIsotonicRegressionMultidimensional2020} proposed multiple isotonic regressions with the monotonicity defined on graphs, where $a < b$ if vertex $a$ is a descendant of vertex $b$ on a graph. Although our proposed monotone splines seem difficult to extend to general multidimensional splines, it would be promising to incorporate the monotonicity into additive splines, where the monotonicity of multidimensional functions can be defined analogously to \textcite{chipmanMBARTMultidimensionalMonotone2022}. 

\section*{Acknowledgement}
The main results in this article are developed from Lijun Wang's Ph.D. thesis when he was at the Chinese University of Hong Kong under the supervision of Xiandan Fan. Lijun Wang was supported by the Hong Kong Ph.D. Fellowship
Scheme from the University Grant Committee. Xiaodan Fan was supported by two grants from the Research Grants Council (14303819, C4012-20E) of the Hong Kong SAR, China. Jun S. Liu was supported by the NSF grant DMS-2015411 and NIH R01 HG011485-01.

\section*{Supplementary Material}

The \supp{} contains technical proofs of propositions and theorems, additional simulation results, and some further discussions.

\printbibliography

\appendix
\newpage

\begin{abstract}
    This supplementary material contains technical proofs of propositions and theorems, additional simulation results, and some further discussions.
\end{abstract}

\section{Basic Property of B-splines}

The basic properties of B-spline can be summarized in Proposition \ref{prop:bspl}, which are adapted from Exercise 5.2 of \textcite{hastieElementsStatisticalLearning2009}. 
\begin{proposition}[\cite{hastieElementsStatisticalLearning2009}]\label{prop:bspl}
Suppose $B_{i, M}(x)$ is an order-$M$ B-spline, then
\begin{enumerate}[label=(\roman*)]
    \item $B_{i, M}(x)=0$ for $x\not\in [\tau_i, \tau_{i+M}]$, i.e., the support is at most $M+1$ knots.
    \item $B_{i, M}(x)>0$ for $x\in (\tau_i,\tau_{i+M})$, i.e., B-splines are positive in the interior of the support.
    \item $\sum_{i=1}^{K+M}B_{i,M}(x)=1\,,\forall x\in [\xi_0, \xi_{K+1}]$. Let $\bB$ be an $n\times (K+M)$ matrix, where the $i$-th column is the evaluated $B_{i,M}(x)$ at $n$ points, then $\bB\one_{K+M}=\one_n$.
\end{enumerate}
\end{proposition}

\subsection{(i)}

\begin{proof}
Firstly,
$$
B_{i,1}(x) = \begin{cases}
1 & x\in [\tau_i, \tau_{i+1})\\
0 & \text{othewise}
\end{cases}
$$
then when $x\not \in [\tau_i,\tau_{i+1}]$, we have $B_{i,1}(x)=0$. 

Next, suppose when $m=k$, we have $B_{i,k}(x)=0$ for $x\not\in [\tau_i,\tau_{i+k}]$, then when $m=k+1$,
$$
B_{i,k+1}(x) = \frac{x-\tau_i}{\tau_{i+k}-\tau_i} B_{i,k}(x) + \frac{\tau_{i+k+1}-x}{\tau_{i+k+1}-\tau_{i+1}}B_{i+1,k}(x)\,,
$$
by the assumption, 
\begin{align*}
    x\not\in [\tau_i, \tau_{i+k}],&\quad B_{i,k}(x)=0\\
    x\not\in [\tau_{i+1}, \tau_{i+k+1}],&\quad B_{i+1,k}(x)=0\,,
\end{align*}
then if $x\in[\tau_0,\tau_i)$ or $x\in (\tau_{i+k+1},\tau_{K+2M}]$, $B_{i,k}(x)=B_{i+1,k}(x)=0$, thus if $x\not\in [\tau_i, \tau_{i+k+1}]$, $B_{i,k+1}(x)=0$.
By induction, the proof is complete.

\end{proof}

\subsection{(ii)}

\begin{proof}
First of all,
$$
B_{i,1}(x) = \begin{cases}
1 & x\in [\tau_i, \tau_{i+1})\\
0 & \text{othewise}
\end{cases}
$$
then when $x\in (\tau_i,\tau_{i+1})$, we have $B_{i,1}(x)> 0$. 

Next, suppose when $m=k$, we have $B_{i,k}(x)> 0$ for $x\in (\tau_i,\tau_{i+k})$. Consider when $m=k+1$,
$$
B_{i,k+1}(x) = \frac{x-\tau_i}{\tau_{i+k}-\tau_i} B_{i,k}(x) + \frac{\tau_{i+k+1}-x}{\tau_{i+k+1}-\tau_{i+1}}B_{i+1,k}(x)\,,
$$
by the assumption
\begin{align*}
    x\in (\tau_i, \tau_{i+k}),&\quad B_{i,k}(x) > 0\\
    x\in (\tau_{i+1}, \tau_{i+k+1}),&\quad B_{i+1,k}(x) > 0\,,
\end{align*}
and by the conclusion from (i), we have
$$
\begin{cases}
\dfrac{x-\tau_i}{\tau_{i+k}-\tau_i} > 0\,,\quad B_{i,k}(x) > 0\,,\quad \quad \dfrac{\tau_{i+k+1}-x}{\tau_{i+k+1}-\tau_{i+1}}>0\,,\quad B_{i+1,k}(x)= 0 & \text{if }x\in (\tau_i, \tau_{i+1})\\
\dfrac{x-\tau_i}{\tau_{i+k}-\tau_i} > 0\,,\quad B_{i,k}(x) > 0\,,\quad \quad \dfrac{\tau_{i+k+1}-x}{\tau_{i+k+1}-\tau_{i+1}}>0\,,\quad B_{i+1,k}(x)\ge 0 &\text{if } x = \tau_{i+1}\\
\dfrac{x-\tau_i}{\tau_{i+k}-\tau_i} > 0\,,\quad B_{i,k}(x) > 0\,,\quad \dfrac{\tau_{i+k+1}-x}{\tau_{i+k+1}-\tau_{i+1}} > 0\,,\quad B_{i+1,k}(x) > 0 & \text{if }x\in (\tau_{i+1}, \tau_{i+k})\\
\dfrac{x-\tau_i}{\tau_{i+k}-\tau_i} >0\,,\quad B_{i,k}(x) \ge 0\,,\quad \dfrac{\tau_{i+k+1}-x}{\tau_{i+k+1}-\tau_{i+1}}>0\,,\quad B_{i+1,k}(x) > 0 & \text{if }x = \tau_{i+k}\\
\dfrac{x-\tau_i}{\tau_{i+k}-\tau_i} >0\,,\quad B_{i,k}(x) = 0\,,\quad \dfrac{\tau_{i+k+1}-x}{\tau_{i+k+1}-\tau_{i+1}}>0\,,\quad B_{i+1,k}(x) > 0 & \text{if }x\in (\tau_{i+k}, \tau_{i+k+1})
\end{cases}
$$
then if $x\in (\tau_i, \tau_{i+k+1})$, $B_{i, k+1}(x) > 0$.

By induction, if $x\in (\tau_i,\tau_{i+M})$, $B_{i,M}(x) > 0$.
\end{proof}

\subsection{(iii)}

\begin{proof}
Firstly, when order is 1,
$$
\sum_{i=1}^{K+2M-1}B_{i,1}(x)=1\,.
$$

Next, suppose when order is $m$, we have
$$
\sum_{i=1}^{K+2M-m} B_{i,m}(x)=1\,,
$$
then consider order is $m+1$, where $m+1\le M$,
\begin{align*}
    \sum_{i=1}^{K+2M-m-1}B_{i,m+1}(x) &= \sum_{i=1}^{K+2M-m-1}\left[\frac{x-\tau_i}{\tau_{i+m}-\tau_i} B_{i,m}(x) + \frac{\tau_{i+m+1}-x}{\tau_{i+m+1}-\tau_{i+1}}B_{i+1,m}(x)\right]\\
    &=\left[ \sum_{i=1}^{K+2M-m}\frac{x-\tau_i}{\tau_{i+m}-\tau_i} B_{i,m}(x) - \frac{x-\tau_{K+2M-m}}{\tau_{K+2M}-\tau_{K+2M-m}}B_{K+2M-m,m}(x) \right] + \\
    &\qquad \left[ \sum_{j=2}^{K+2M-m} \frac{\tau_{j+m}-x}{\tau_{j+m}-\tau_{j}}B_{j,m}(x) \right]\\
    &=\left[ \sum_{i=1}^{K+2M-m}\frac{x-\tau_i}{\tau_{i+m}-\tau_i} B_{i,m}(x) - \frac{x-\tau_{K+2M-m}}{\tau_{K+2M}-\tau_{K+2M-m}}B_{K+2M-m,m}(x) \right] + \\
    &\qquad \left[ \sum_{j=1}^{K+2M-m} \frac{\tau_{j+m}-x}{\tau_{j+m}-\tau_{j}}B_{j,m}(x) - \frac{\tau_{m+1}-x}{\tau_{m+1}-\tau_1}B_{1,m}(x)\right]\,,
\end{align*}
since $m+1\le M$, then $\tau_{m+1}\le \tau_M, \tau_{K+2M-m}\ge \tau_{K+M+1}$, it follows that $\tau_{m+1}=\tau_1, \tau_{K+2M}=\tau_{K+2M-m}$, then 
$$
B_{1,m}(x) = 0, \quad B_{K+2M-m, m}(x) = 0\,,
$$
thus
\begin{align*}
\sum_{i=1}^{K+2M-m-1}B_{i,m+1}(x) &= \sum_{i=1}^{K+2M-m}\frac{x-\tau_i}{\tau_{i+m}-\tau_i} B_{i,m}(x) + \sum_{j=1}^{K+2M-m} \frac{\tau_{j+m}-x}{\tau_{j+m}-\tau_{j}}B_{j,m}(x)\\ 
&= \sum_{i=1}^{K+2M-m}\left[\frac{x-\tau_i}{\tau_{i+m}-\tau_i} + \frac{\tau_{i+m}-x}{\tau_{i+m}-\tau_{i}}\right] B_{i,m}(x) \\
&= \sum_{i=1}^{K+2M-m}B_{i,m}(x)=1\,.
\end{align*}
Thus, by induction, $\sum_{i=1}^{K+M}B_{i, M}(x)=1$.
\end{proof}

\section{GBS not work for splines}

In GBS (or GMS), the preliminary assumption is that the loss function can be written as the sum of the loss of each observation:
$$
\hat\theta = \argmin_{\theta} L_\bfy(\theta), \qquad \text{with}\quad L_\bfy(\theta)\equiv \frac 1n \sum_{i=1}^n \ell(\theta;y_i)\,,
$$
where $\ell(\cdot)$ is a suitable loss function.
Even though a penalty $u(\theta)$ can be imposed, but it should be independent of the observation $\bfy$,
$$
L_\bfy(\theta, \lambda) = \frac 1n \sum_{i=1}^n \ell(\theta;y_i) + \lambda u(\theta) = \frac 1n\sum_{i=1}^n\left\{\ell(\theta;y_i)+\lambda u(\theta)\right\} \triangleq \frac 1n \sum_{i=1}^n \tilde\ell (\theta; y_i,\lambda)\,,
$$
then again the total loss function can be decomposed as sum of ``new'' individual loss $\tilde\ell(\theta;y_i,\lambda)$, which absorbs the penalty term.

Then the nonparametric bootstrap can be written as solving
$$
\hat\theta^\star = \argmin_{\theta}\frac 1n \sum_{i=1}^n w_i\ell(\theta;y_i)\,,
$$
where $\bfw=(w_1,\ldots, w_n)$ is sampled from multi-nominal distribution $\text{Multinom}(n, \one_n/n)$, i.e., $w_i$ counts the times that point $i$ is observed among $n$ trials and each point can be observed with equal probabilities $p_1=\cdots=p_n=1/n$. 

The loss function for the smoothing spline is
$$
L(\bfy, \bfx, \lambda) = \Vert \bfy -\bfB(\bfx)\gamma\Vert_2^2 +\lambda \gamma^T\bfOmega(\bfx)\gamma = \sum_{i=1}^n (y_i-\bfb^T_i\gamma)^2+\lambda \gamma^T\bfOmega(\bfx)\gamma\,.
$$
If $\lambda = 0$, we can treat $(\bfb_i, y_i)$ as a unit, where $\bfb_i$ denotes the $i$-th row of $\bB$, then the loss function can be written as sum of individual losses,
$$
L(\bfy, \bfx, \lambda = 0) = \sum_{i=1}^n (y_i-\bfb^T_i\gamma)^2\,,
$$
then the loss for the bootstrap sample $\{\bfx^\star, \bfy^\star\}$ is
$$
L(\bfy^\star, \bfx^\star, \lambda = 0) = \sum_{i=1}^n w_i (y_i-\bfb^T_i\gamma)^2\,,
$$
which indicates the basis matrix for the bootstrap sample is
$$
\tilde \bfB = [\underbrace{\bfb_1, \ldots,\bfb_1}_{nw_1},\underbrace{\bfb_2, \ldots,\bfb_2}_{nw_2},\ldots,
\underbrace{\bfb_n, \ldots,\bfb_n}_{nw_n}]^T\,.
$$
However, $\tilde\bfB$ cannot be a basis matrix for $\bfx^\star$ since it is no longer a lower 4-banded matrix as $\bfB$. When $\lambda > 0$, the loss cannot even be decomposed as the sum of individual losses since the smoothness penalty involves $\bfx$ itself, so the GMS framework is not suitable for smoothing splines.

\section{Coverage Probability}

In addition to the classical percentile CI $(q_{0.025}, q_{0.975})$, there are many variants of bootstrap confidence intervals for better coverage probability, such as
\begin{itemize}
    \item the bias-corrected percentile CI $(2\hat\theta-q_{0.975}, 2\hat\theta-q_{0.025})$ \parencite{efronNonparametricStandardErrors1981}
    \item \textcite{efronBetterBootstrapConfidence1987}'s $\mathrm{BC}_a$ introduced an "acceleration constant" $a$, and if $a=0$, it reduces to the above bias-corrected percentile CI
    \item the calibrated CI via double bootstrap \parencite{martinDoubleBootstrap1992}
    \item the studentized CI via double bootstrap \parencite{hallTheoreticalComparisonBootstrap1988}
\end{itemize}
Although we focus on the classical percentile CI, the comparisons can be seamlessly moved to other bootstrap CIs.

The accuracy would be affected by the size of the studied range of penalty parameters. But practically, we are more concerned about a range that contains the minimizer of the cross-validation error (or some other criteria). Here the selected range of $\lambda$ is wide enough to contain the minimizer of the cross-validation error, as shown in Figure \ref{fig:cv_minimizer}. 
\begin{figure}
    \centering
    \begin{subfigure}{0.5\textwidth}
    \includegraphics[width=\textwidth]{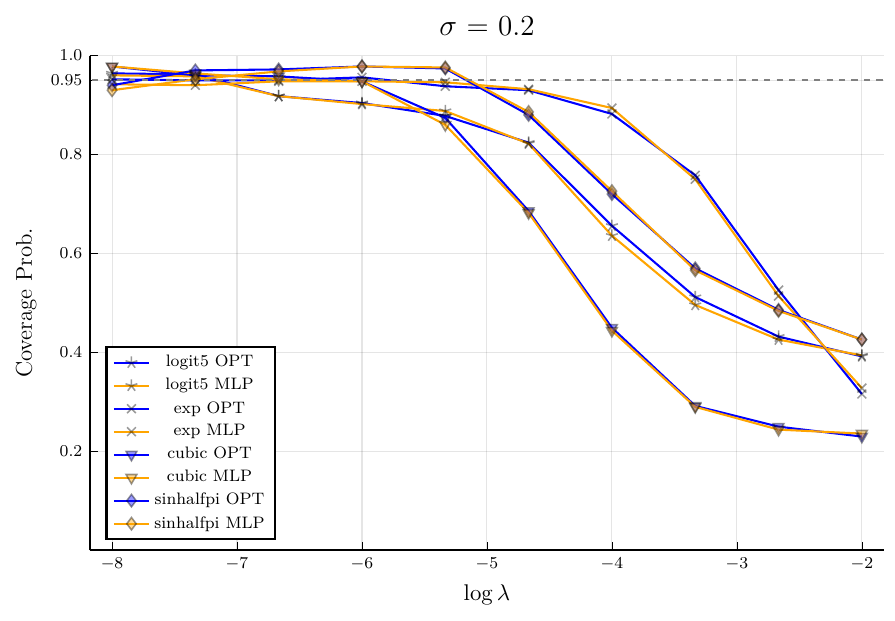}
    \caption{}
    \label{fig:covp_vs_lambda}
    \end{subfigure}%
    \begin{subfigure}{0.5\textwidth}
    \includegraphics[width=\textwidth]{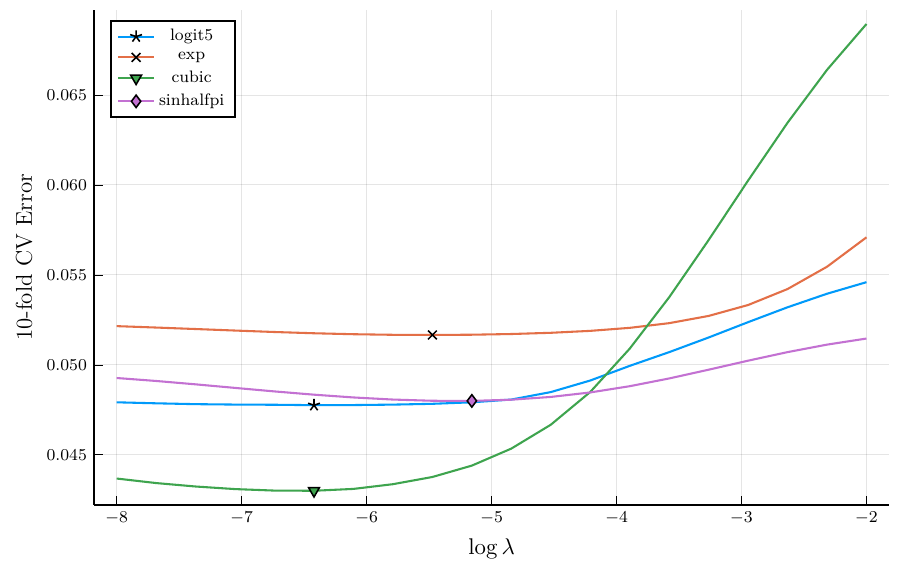}
    \caption{}
    \label{fig:cv_minimizer}
    \end{subfigure}
    \caption{(a) Coverage probability along the penalty parameter $\lambda$ among 5 repetitions. (b) 10-fold cross-validation error for each curve in one specific repetition.}
\end{figure}
Although a better coverage probability is not our direct goal, the coincidence of coverage probabilities from the OPT approach and the MLP generator would be another measurement for checking the approximation performance. Figure \ref{fig:covp_vs_lambda} displays the average coverage probabilities among five repetitions for each curve. Firstly, the blue and orange colors represent the coverage probabilities based on the OPT approach and the MLP generator. Different curves are denoted by different symbols. The coverage probabilities of the blue curve and the orange curve are pretty close, which indicates that the MLP generator can indeed achieve a good approximation to the OPT solution. The dashed horizontal line denotes the coverage probability $1-\alpha$, where $\alpha=0.05$ is the nominal significance level. Roughly, the coverage probabilities are close to 0.95 when $\log\lambda < -5$, and then they decrease. In other words, the coverage probability of the confidence band becomes worse with increasing $\lambda$. This can be explained by the cross-validation (CV) error curve. The minimizers of CV error curves are roughly on the left side of $\log\lambda=-5$. With a larger penalty parameter, the smoothing penalty would cause the final fitting to underfit. In an extreme case, it becomes a straight line, so it is not surprising that the coverage probability can not retain around 0.95. On the other hand, the CV error curve indicates that the selected range of $\lambda$ is wide enough to consider the situations of underfitting and overfitting.

\section{Proof of Proposition 1}


The first derivative of B-spline function $\sum\gamma_jB_{j,m}(x)$ turns out to be a spline of one order lower, and it can be calculated by differencing the coefficients.
\begin{lemma}[\cite{deboorPracticalGuideSplines1978}]
The first derivative of a spline function $\sum_{j=s}^r\gamma_jB_{j,m}(x)$ is
\begin{equation}
D\left(\sum_{j=r}^s\gamma_jB_{j,m}(x)\right) = (m-1)\sum_{j=r}^{s+1}\frac{\gamma_j-\gamma_{j-1}}{\tau_{j+m-1}-\tau_j}B_{j,m-1}(x)\,,\label{eq:bspl_diff}
\end{equation}
where $\gamma_{r-1}\triangleq 0$ and $\gamma_{s+1} \triangleq 0$.
\end{lemma}
As a consequence, we can obtain the derivative of the linear combination of all $J$ B-spline basis functions. 
\begin{corollary}
The first derivative of a spline function $\sum_{j=1}^J\gamma_jB_{j,m}(x)$ is
\begin{equation}
D\left(\sum_{j=1}^J\gamma_jB_{j,m}(x)\right) = (m-1)\sum_{j=2}^{J}\frac{\gamma_j-\gamma_{j-1}}{\tau_{j+m-1}-\tau_j}B_{j,m-1}(x)\,.\label{eq:bspl_diff2}
\end{equation}
\end{corollary}
\begin{proof}
\begin{align}
    B'_{j,m}(x) = (m-1)\left(\frac{-B_{j+1,m-1}(x)}{\tau_{j+m}-\tau_{j+1}} + \frac{B_{j,m-1}(x)}{\tau_{j+m-1}-\tau_j}\right)
\end{align}
then
\begin{align*}
    &D\left(\sum_{j=1}^J\gamma_jB_{j,m}(x)\right) \\
    &= \alpha_1\frac{B_{1,m-1}(x)}{\tau_m-\tau_1} + \frac{(\alpha_2-\alpha_1)B_{2,m-1}(x)}{\tau_{m+1}-\tau_2} + \cdots + \frac{(\alpha_J-\alpha_{J-1})B_{J,m-1}(x)}{\tau_{J+m-1}-\tau_J} - \frac{\alpha_JB_{J+1,m-1}(x)}{\tau_{J+m}-\tau_{J+1}}
\end{align*}
Note that
$$
\tau_1 =\tau_2 = \cdots=\tau_m = \xi_0 < \xi_1 <\cdots < \xi_K < \xi_{K+1} = \tau_{J+1} = \cdots =\tau_{J+m}\,,
$$
so the first term and the last term are zero, thus
$$
D\left(\sum_{j=1}^J\gamma_jB_{j,m}(x)\right) = \sum_{j=2}^J\frac{\gamma_j-\gamma_{j-1}}{\tau_{j+m-1}-\tau_j}B_{j,m-1}(x)\,.
$$
\end{proof}
Note that the limits of summation in Equations~\eqref{eq:bspl_diff}~\eqref{eq:bspl_diff2} are different. Compared to Equation~\eqref{eq:bspl_diff}, two boundary terms become to zero in Equation~\eqref{eq:bspl_diff2} due to $\tau_1=\tau_m$ and $\tau_{J+1}=\tau_{J+m}$.

For quadratic splines $m=3$, since $B_{j,2}(t)$ reduces to linear functions, then the nonnegative (or nonpositive) first derivative constraints in the whole domain can be reduced to the constraints on the knots. But for cubic splines $m=4$, we cannot characterize monotonicity as linear constraints at the knots in the same way. 

\subsection{Sufficient condition}\label{sec:proof_suff_cond}
If $\bA\gamma \le 0$, then the first derivative is larger than zero, thus $\bA\gamma\le 0$ is a sufficient condition.

\subsection{Necessary condition}
    If $f(x)$ is non-decreasing, then $f'(x)\ge 0$. Evaluate $f'(x)$ on $\{\xi_i\}_{i=0}^{K+1}$ and write in matrix form, we have $\bB^{(1)}\bD^{-1}\bA\gamma\le 0$, which would be a necessary condition.

\subsection{Sufficient and necessary condition}

Note that the sufficient and necessary condition for nondecreasing spline function is
$$
f'(x) \ge \min_{x\in[\xi_0,\xi_{K+1}]} f'(x) = \min_{x\in\{\xi_0, \xi_{K+1}\}\cup\{x:f''(x)=0\}} f'(x)\,.
$$
Now we find the roots of $f''(x)=0$ in the intervals $[\xi_0, \xi_{K+1}]$.

First of all, the second derivative of a spline function $f(x)=\sum_{j=1}^J\gamma_jB_{j,m}(x)$ is 
\begin{align*}
    D^2\left(\sum_{j=1}^J\gamma_jB_{j, m}(x)\right) &= (m-1)D\left(\sum_{j=2}^J\frac{\gamma_j-\gamma_{j-1}}{\tau_{j+m-1}-\tau_j}B_{j,m-1}(x)\right)\\
&=(m-1)\sum_{j=2}^J\frac{\gamma_j-\gamma_{j-1}}{h_j^{(m-1)}}D(B_{j,3}(x))\\
&=(m-1)\sum_{j=2}^J\frac{\gamma_j-\gamma_{j-1}}{h_j^{(m-1)}}(m-2)\left(\frac{-B_{j+1,m-2}(x)}{\tau_{j+m-1}-\tau_{j+1}}+\frac{B_{j, m-2}(x)}{\tau_{j+m-2}-\tau_j}\right)\,,
\end{align*}
where $h_i^{(j)}\triangleq \tau_{i+j}-\tau_i$. Now for cubic splines $m=4$, we have
\begin{align*}
&=6\sum_{j=2}^J\frac{\gamma_j-\gamma_{j-1}}{h_j^{(3)}}\left(\frac{-B_{j+1,2}(x)}{h_{j+1}^{(2)}} + \frac{B_{j,2}(x)}{h_{j}^{(2)}}\right)\\
&=6\left\{\sum_{j=2}^J\frac{\gamma_j-\gamma_{j-1}}{h_j^{(2)}h_j^{(3)}}B_{j,2}(x) - \sum_{j=2}^J\frac{\gamma_j-\gamma_{j-1}}{h_j^{(3)}h_{j+1}^{(2)}}B_{j+1,2}(x)\right\}\\
&=6\left\{\sum_{j=2}^J\frac{\gamma_j-\gamma_{j-1}}{h_j^{(2)}h_j^{(3)}}B_{j,2}(x) - \sum_{j=2}^{J-1}\frac{\gamma_j-\gamma_{j-1}}{h_j^{(3)}h_{j+1}^{(2)}}B_{j+1,2}(x)\right\}\\
&=6\left\{\sum_{j=2}^J\frac{\gamma_j-\gamma_{j-1}}{h_j^{(2)}h_j^{(3)}}B_{j,2}(x) - \sum_{k=3}^{J}\frac{\gamma_{k-1}-\gamma_{k-2}}{h_{k-1}^{(3)}h_{k}^{(2)}}B_{k,2}(x)\right\}\\
&=6\left\{\frac{\gamma_2-\gamma_1}{h_2^{(2)}h_2^{(3)}} + \sum_{j=3}^J\left[\frac{\gamma_j-\gamma_{j-1}}{h_j^{(2)}h_j^{(3)}}-\frac{\gamma_{j-1}-\gamma_{j-2}}{h_j^{(2)}h_{j-1}^{(3)}}\right]B_{j,2}(x)\right\}\\
&=6\sum_{j=3}^J\left[\frac{\gamma_j-\gamma_{j-1}}{h_j^{(2)}h_j^{(3)}}-\frac{\gamma_{j-1}-\gamma_{j-2}}{h_j^{(2)}h_{j-1}^{(3)}}\right]B_{j,2}(x)\\
&=6\sum_{j=3}^J\frac{1}{h_j^{(2)}}\left[\frac{\gamma_j-\gamma_{j-1}}{h_j^{(3)}}-\frac{\gamma_{j-1}-\gamma_{j-2}}{h_{j-1}^{(3)}}\right]B_{j,2}(x)\\
&\triangleq 6\sum_{j=3}^JA_{j}B_{j,2}(x)\,.
\end{align*}
Note that
\begin{align*}
    B_{i, 2}(x) &= \frac{x-\tau_i}{\tau_{i+1} -\tau_i} B_{i, 1}(x) + \frac{\tau_{i+2}-x}{\tau_{i+2}-\tau_{i+1}}B_{i, i+1}(x)\\
    &=\begin{cases}
        \dfrac{x - \tau_i}{\tau_{i+1}-\tau_i} & \tau_i\in [\tau_i, \tau_{i+1})\\
        \dfrac{\tau_{i+2}-x}{\tau_{i+2}-\tau_{i+1}} & \tau_i\in [\tau_{i+1}, \tau_{i+2})
    \end{cases}\,,
\end{align*}
then in the interval $x\in [\tau_i, \tau_{i+1}], i = 4,\ldots,J$, to have
\begin{equation}
f''(x) = 6\sum_{j\in\{i, i-1\}}A_jB_{j,2}(x)=A_i\frac{x-\tau_i}{h_i} + A_{i-1}\frac{\tau_{i+1}-x}{h_i} = 0\,,  \label{eq:bspl_2nd_deriv}
\end{equation}
we obtain
$$
x = \frac{A_i}{A_i-A_{i-1}}\tau_i +\frac{-A_{i-1}}{A_i-A_{i-1}}\tau_{i+1}\,,
$$
the minimizer lies in $[\tau_i, \tau_{i+1}]$ only when both $A_i$ and $-A_{i-1}$ are positive.

\subsection{Another sufficient and necessary condition}
Since $f''(x)$ is a linear combination of linear functions in each interval, then $f''(x)\ge 0$ if and only if $f''(\tau_i)\ge 0$. And it follows that $f'(x)\ge f'(\xi_0)$, then if $f$ is a non-decreasing function, we must have $f'(\xi_0) \ge 0$.
Thus, when $f''(x)\ge 0$, we have the following necessary condition,
$$
\begin{cases}
f''(\tau_{i+1}) = A_i & \ge 0\,, i=3,\ldots,J\,,\\
f'(\xi_0) & \ge 0\,.
\end{cases}
$$
Note that
\begin{align*}
    f'(\tau_i) &=3\sum_{j\in\{i-1,i-2\}}\frac{\gamma_j-\gamma_{j-1}}{\tau_{j+3}-\tau_j}B_{j,3}(x) \\
&= 3\left[\frac{\gamma_{i-1}-\gamma_{i-2}}{\tau_{i+2}-\tau_{i-1}}\cdot\frac{\tau_{i}-\tau_{i-1}}{\tau_{i+1}-\tau_{i-1}} + \frac{\gamma_{i-2}-\gamma_{i-3}}{\tau_{i+1}-\tau_{i-2}}\cdot\frac{\tau_{i+1}-\tau_{i}}{\tau_{i+1}-\tau_{i-1}}\right]\\
&= 3\left[\frac{h_{i-1}(\gamma_{i-1}-\gamma_{i-2})}{h_{i-1}^{(3)}h_{i-1}^{(2)}}+\frac{h_{i}(\gamma_{i-2}-\gamma_{i-3})}{h_{i-2}^{(3)}h_{i-1}^{(2)}}\right]\,,
\end{align*}
we have 
$$
f'(\xi_0) = f'(\tau_4) = 3\left[\frac{h_{3}(\gamma_{3}-\gamma_{2})}{h_{3}^{(3)}h_{3}^{(2)}}+\frac{h_{4}(\gamma_{2}-\gamma_{1})}{h_{2}^{(3)}h_{3}^{(2)}}\right] = 3\frac{\gamma_2-\gamma_1}{h_4}\,.
$$
then the condition can be written as
$$
\begin{cases}
    \dfrac{\gamma_j-\gamma_{j-1}}{h_j^{(3)}}-\dfrac{\gamma_{j-1}-\gamma_{j-2}}{h_{j-1}^{(3)}}\ge 0\;\;\forall j=3,\ldots,J\,,\\
    \gamma_2-\gamma_1\ge 0\,.
\end{cases}
$$
In matrix form, it becomes $\bfC\gamma \ge 0$ with
$$
\bfC = 
\begin{bmatrix}
    -1 & 1 & 0 & 0 & \cdots& 0 & 0 & 0\\
    \frac{1}{h_2^{(3)}} & \;\;-\left(\frac{1}{h_2^{(3)}} +\frac{1}{h_3^{(3)}}\right) & \frac{1}{h_3^{(3)}} & 0 & \cdots& 0 & 0 & 0\\
    0 & \frac{1}{h_3^{(3)}} & -\left(\frac{1}{h_3^{(3)}}+\frac{1}{h_4^{(3)}}\right) & \frac{1}{h_4^{(3)}} & \cdots& 0 & 0 & 0\\
    \vdots & \vdots & \vdots & \vdots &\ddots & -\left(\frac{1}{h_{J-2}^{(3)}} + \frac{1}{h_{J-1}^{(3)}}\right) & \frac{1}{h_{J-1}^{(3)}} & 0\\
    0 & 0 & 0 & 0 & \cdots & \frac{1}{h_{J-1}^{(3)}} & -\left(\frac{1}{h_{J-1}^{(3)}} + \frac{1}{h_{J}^{(3)}}\right) & \frac{1}{h_{J}^{(3)}}
\end{bmatrix}\,.
$$
By Gaussian elimination on rows of $\bfC$, 
\begin{enumerate}
    \item multiply 2nd row by $h_2^{(3)}$ and plus the 1st row;
    \item multiply 3rd row by $h_3^{(3)}$ and plus the (updated) 2nd row;
    \item multiply 4th row by $h_4^{(3)}$ and plus the (updated) 3rd row;
    \item ...
\end{enumerate}
$\bfC\gamma\ge 0$ becomes $\bfC'\gamma\ge 0$ with
$$
\bfC' = \begin{bmatrix}
    -1 & 1 & 0 & 0 & \cdots & 0 & 0 & 0\\
    0 & -\frac{h_2^{(3)}}{h_3^{(3)}} & \frac{h_2^{(3)}}{h_3^{(3)}} & 0 & \cdots & 0 & 0 & 0\\
    0 & 0 & -\frac{h_3^{(3)}}{h_4^{(3)}} & \frac{h_3^{(3)}}{h_4^{(3)}} & \cdots & 0 & 0 & 0\\
    \vdots & \vdots & \vdots & \vdots &\ddots & - \frac{h_{J-2}^{(3)}}{h_{J-1}^{(3)}} & \frac{h_{J-2}^{(3)}}{h_{J-1}^{(3)}} & 0\\
    0 & 0 & 0 & 0 & 0 & 0 & - \frac{h_{J-1}^{(3)}}{h_{J}^{(3)}} & \frac{h_{J-1}^{(3)}}{h_{J}^{(3)}}
\end{bmatrix}\,,
$$
then normalize the $j$-th row by multiplying $\frac{h_{j+1}^{(3)}}{h_{j}^{(3)}}$, $j=2,\ldots, J-1$, and denote the resulting matrix as $\bfC''$. Now the condition becomes $\bfC''\gamma\ge 0$.

Note that $\bfC'' = -\bfA$ in Section~\ref{sec:proof_suff_cond}, so we reach the same condition. In other words, $\bfA\gamma\le 0$ is also necessary for $f$ to be non-decreasing when $f''(x)\ge 0$.

\section{Proof of Theorem 1}

\subsection{No error: $\sigma^2 = 0$}
\begin{proof}
Denote
\begin{align*}
    \hat\gamma = \argmin_{\bA\gamma\le 0}\Vert \bfy-\bfB\gamma\Vert_2^2\,,
\end{align*}
and let $\iso(\gamma)$ be the isotonic fitting to $\gamma$, i.e.,
$$
\iso(\gamma)  = \argmin_{\bA x\le 0} \Vert \gamma - x\Vert_2^2\,.
$$
Since $\bA\iso(\gamma)\le 0$, that is, $\iso(\gamma)$ is a feasible point for 
$$
\min_{\gamma:\bA\gamma\le 0}\Vert \bfy-\bfB\gamma\Vert_2\,,
$$
then
\begin{equation}
\Vert \bfy - \bfB\hat\gamma\Vert_2 \le \Vert \bfy-\bfB\iso(\gamma)\Vert_2\,.    
\label{eq:error_suff_vs_iso}
\end{equation}
Particularly, when $\bfe = 0$, we have $\bfy = \bfB\gamma$, \begin{equation}
\Vert \bfB\gamma - \bfB\hat\gamma\Vert_2\le\Vert \bfB\gamma-\bfB\tilde\gamma\Vert_2\,.
\label{ieq:1}
\end{equation}
Note that
\begin{equation}
\Vert\bfB \gamma - \bfB\tilde\gamma\Vert_2^2 = (\gamma - \tilde\gamma)^T\bfB^T\bfB(\gamma - \tilde\gamma)\le \lambda_1\Vert \gamma - \tilde\gamma\Vert_2^2\,,
\label{ieq:2}
\end{equation}
where $\lambda_1$ is the maximum eigenvalues of $\bfB^T\bfB$.

Since $\bfB$ is non-negative matrix, and so is $\bfB^T\bfB$, then the largest eigenvalue satisfies (see Theorem 8.1.22 of \textcite{hornMatrixAnalysis2012})
\begin{align*}
\min_j\sum_i (\bfB^T\bfB)_{ij} & \le \lambda_1\le \max_i \sum_{j}(\bfB^T\bfB)_{ij} \,,\\
\min_j (\one^T\bfB)_j=\min_j(\one^T\bfB^T\bfB)_j & \le \lambda_1\le \max_i (\bfB^T\bfB\one)_i = \max_i(\bfB^T\one)_i\,,
\end{align*}
where
$$
(\bfB^T\one)_j=\sum_{i=1}^n B_j(x_i)\,.
$$
Since there are at most $\frac{n}{J-3}(1+\eta_1)$ points in each interval and each cubic B-spline basis function are nonzero in at most four regions, then
\begin{equation}\label{eq:bound_sum_bji}
\sum_{i=1}^n B_j(x_i) \le \frac{n}{J-3}(1+\eta_1)\cdot 4\max_x B_j(x) \le \frac{n}{J-3}(1+\eta_1)\cdot 4\,.    
\end{equation}

The first derivative is
$$
f'(x) = 3\sum_{j=2}^J\frac{\gamma_j-\gamma_{j-1}}{\tau_{j+3}-\tau_j}B_{j,3}(x)\,,
$$
where
$$
    B_{j,3}(x) =\begin{cases}
        \dfrac{x-\tau_j}{\tau_{j+2}-\tau_j}\cdot \dfrac{x-\tau_j}{\tau_{j+1}-\tau_j} & x\in [\tau_j, \tau_{j+1})\\
        \dfrac{x-\tau_j}{\tau_{j+2}-\tau_j}\cdot \dfrac{\tau_{j+2}-x}{\tau_{j+2}-\tau_{j+1}} + \dfrac{\tau_{j+3}-x}{\tau_{j+3}-\tau_{j+1}}\cdot \dfrac{x-\tau_{j+1}}{\tau_{j+2}-\tau_{j+1}} & x\in[\tau_{j+1}, \tau_{j+2})\\
        \dfrac{\tau_{j+3}-x}{\tau_{j+3}-\tau_{j+1}}\cdot\dfrac{\tau_{j+3}-x}{\tau_{j+3}-\tau_{j+2}} & x\in [\tau_{j+2}, \tau_{j+3})
    \end{cases}\,.
$$
Thus for $x\in [\tau_i, \tau_{i+1})$, only 3 basis functions are nonzero,
\begin{align}
\frac 13f'(x)
&= \sum_{i=j,j+1,j+2}\frac{\gamma_j-\gamma_{j-1}}{\tau_{j+3}-\tau_j}B_{j,3}(x)\notag\\
&=\frac{\gamma_i-\gamma_{i-1}}{h_i^{(3)}}B_{i,3}(x) + \frac{\gamma_{i-1}-\gamma_{i-2}}{h_{i-1}^{(3)}}B_{i-1,3}(x) + \frac{\gamma_{i-2}-\gamma_{i-3}}{h_{i-2}^{(3)}}B_{i-2,3}(x)\label{eq:fp_i}\,.
\end{align}

Note that
\begin{equation}
B_{i-1,3}(x) = \frac{(x-\tau_{i-1})(\tau_{i+1}-x)}{h_{i-1}^{(2)}h_i} + \dfrac{(\tau_{i+2}-x)(x-\tau_i)}{h_i^{(2)}h_{i}}\\
\triangleq \frac{1}{h_i}[C_1(x) +C_2(x)]\,,\label{eq:b_i_1}
\end{equation}
where
\begin{align*}
C_1(x)&=\frac{(x-\tau_{i-1})(\tau_{i+1}-x)}{h_{i-1}^{(2)}}=\frac{(x-\tau_i+\tau_i-\tau_{i-1})(\tau_{i+1}-x)}{h_{i-1}^{(2)}} \\
&= \frac{(x-\tau_i)(\tau_{i+1}-x)}{h_{i-1}^{(2)}} +\frac{h_{i-1}(\tau_{i+1}-x)}{h_{i-1}^{(2)}}\\
C_2(x)&=\frac{(\tau_{i+2}-x)(x-\tau_i)}{h_i^{(2)}} = \frac{(\tau_{i+2}-\tau_{i+1}+\tau_{i+1}-x)(x-\tau_i)}{h_i^{(2)}} \\
&= \frac{h_{i+1}(x-\tau_i)}{h_i^{(2)}} + \frac{(\tau_{i+1}-x)(x-\tau_i)}{h_{i}^{(2)}}\,.
\end{align*}
Since $x\in [\tau_i, \tau_{i+1})$, then $\tau_{i+1}-x>0, x-\tau_i\ge 0$, then we have
\begin{equation}\label{eq:bound_b_i_1}
    \begin{split}
    C_1(x) \ge \frac{(x-\tau_i)(\tau_{i+1}-x)}{h_{i-1}^{(2)}}\triangleq c_1(x)\,,\\
C_2(x) \ge \frac{(\tau_{i+1}-x)(x-\tau_i)}{h_i^{(2)}}\triangleq c_2(x)\,.        
    \end{split}
\end{equation}
The first derivatives at the knots are
\begin{align*}
f'(\tau_i) &= \frac{h_{i-1}(\gamma_{i-1}-\gamma_{i-2})}{h_{i-1}^{(3)}h_{i-1}^{(2)}} + \frac{h_i(\gamma_{i-2}-\gamma_{i-3})}{h_{i-2}^{(3)}h_{i-1}^{(2)}}\,,\\
f'(\tau_{i+1}) &= \frac{h_i(\gamma_i-\gamma_{i-1})}{h_i^{(3)}h_i^{(2)}} + \frac{h_{i+1}(\gamma_{i-1}-\gamma_{i-2})}{h_{i-1}^{(3)}h_i^{(2)}}\,.
\end{align*}
Particularly, since $h_3=0, h_4=h_2^{(3)}=h_3^{(2)}; h_{J+1}=0, h_J = h_J^{(2)}=h_J^{(3)}$, the first derivatives at two boundary knots are
\begin{align*}
    f'(\tau_4) &= \frac{h_4(\gamma_2-\gamma_1)}{h_2^{(3)}h_3^{(2)}} = \frac{\gamma_2-\gamma_1}{h_4}\,,\\
    f'(\tau_{J+1}) &= \frac{h_J(\gamma_J-\gamma_{J-1})}{h_J^{(3)}h_J^{(2)}} = \frac{\gamma_J-\gamma_{J-1}}{h_J}\,.
\end{align*}

Since $f'(x)\ge 0$, then we have $f'(\tau_{i})\ge 0, f'(\tau_{i+1})\ge 0$, and hence
$$
\gamma_2-\gamma_1\ge 0\qquad \gamma_J-\gamma_{J-1}\ge 0\,.
$$
Also note that if $\gamma_{i-1}-\gamma_{i-2} < 0$, we must have $\gamma_{i-2}-\gamma_{i-3} > 0$ and $\gamma_i - \gamma_{i-1} > 0$. 
Consider $I = \{i: \gamma_{i-1}-\gamma_{i-2} < 0\}$, then for each $i\in I$, let Equation~\eqref{eq:fp_i} $\ge 0$, we obtain
$$
\gamma_{i-2}-\gamma_{i-1} \le\dfrac{\dfrac{\tau_{i+2}-\tau_{i-1}}{\tau_{i+3}-\tau_i}(\gamma_i-\gamma_{i-1})B_{i,3}(x) + \dfrac{\tau_{i+2}-\tau_{i-1}}{\tau_{i+1}-\tau_{i-2}}(\gamma_{i-2}-\gamma_{i-3})B_{i-2,3}(x)}{B_{i-1,3}(x)}\,,\;\forall x\in [\tau_i, \tau_{i+1})\,,
$$
thus
\begin{equation}
\gamma_{i-2}-\gamma_{i-1}\le \min_{x\in [\tau_i, \tau_{i+1})}\dfrac{(\gamma_i-\gamma_{i-1})\dfrac{h_{i-1}^{(3)}}{h_i^{(3)}}B_{i,3}(x)+(\gamma_{i-2}-\gamma_{i-3})\dfrac{h_{i-1}^{(3)}}{h_{i-2}^{(3)}}B_{i-2,3}(x)}{B_{i-1,3}(x)}\,.
\label{eq:neg_gamma}
\end{equation}

Note that for two general functions, if $g(x)\ge h(x)$, then we also have $\min g(x) \ge \min h(x)$. Thus, replace $B_{i-3,3}(x)$ with its lower bound \eqref{eq:bound_b_i_1} in the denominator of Inequality~\eqref{eq:neg_gamma}
\begin{equation}
\gamma_{i-2}-\gamma_{i-1}\le \min_{x\in [\tau_i, \tau_{i+1})}\dfrac{(\gamma_i-\gamma_{i-1})\dfrac{h_{i-1}^{(3)}}{h_i^{(3)}}\dfrac{(x-\tau_i)^2}{h_ih_i^{(2)}}+(\gamma_{i-2}-\gamma_{i-3})\dfrac{h_{i-1}^{(3)}}{h_{i-2}^{(3)}}\dfrac{(\tau_{i+1}-x)^2}{h_{i-1}^{(2)}h_i}}{\dfrac{1}{h_i}\left[\dfrac{1}{h_{i-1}^{(2)}}+\dfrac{1}{h_i^{(2)}}\right](x-\tau_i)(\tau_{i+1}-x)}\,.    
\label{eq:neg_gamma_denom}
\end{equation}
For the numerator of \eqref{eq:neg_gamma_denom}, 
$$
a_1(x-\tau_i)^2 + a_2(\tau_{i+1}-x)^2\ge 2\sqrt{a_1a_2}(x-\tau_i)(\tau_{i+1}-x)\,,
$$
and the equality is obtained when
$$
\sqrt{a_1}(x-\tau_i) = \sqrt{a_2}(\tau_{i+1}-x)\,,
$$
that is 
$$
x = \frac{\sqrt{a_1}}{\sqrt{a_1}+\sqrt{a_2}}\tau_i +\frac{\sqrt{a_2}}{\sqrt{a_1}+\sqrt{a_2}}\tau_{i+1}
$$
which lies in $[\tau_i, \tau_{i+1}]$. Thus, \eqref{eq:neg_gamma_denom} becomes
\begin{align}
\gamma_{i-2}-\gamma_{i-1}&\le \frac{2\sqrt{\dfrac{h_{i-1}^{(3)}}{h_i^{(3)}h_ih_i^{(2)}}\dfrac{h_{i-1}^{(3)}}{h_{i-2}^{(3)}h_{i-1}^{(2)}h_i} } }{\dfrac{1}{h_i}\left[\dfrac{1}{h_{i-1}^{(2)}}+\dfrac{1}{h_i^{(2)}}\right]} = 2\dfrac{\dfrac{h_{i-1}^{(3)}}{\sqrt{h_i^{(3)}h_i^{(2)}h_{i-2}^{(3)}h_{i-1}^{(2)}}}}{\left[\dfrac{1}{h_{i-1}^{(2)}}+\dfrac{1}{h_i^{(2)}}\right]}\sqrt{(\gamma_{i-2}-\gamma_{i-3})(\gamma_{i}-\gamma_{i-1})}\notag\\
&=2\dfrac{h_{i-1}^{(3)}\sqrt{h_{i-1}^{(2)}h_i^{(2)}}}{(h_i^{(2)}+h_{i-1}^{(2)})\sqrt{h_i^{(3)}h_{i-2}^{(3)}}}\sqrt{(\gamma_{i-2}-\gamma_{i-3})(\gamma_i-\gamma_{i-1})}\notag\\
&\triangleq C_i\sqrt{(\gamma_{i-2}-\gamma_{i-3})(\gamma_i-\gamma_{i-1})}\label{eq:gamma_diff_i_2}\,.
\end{align}
By \eqref{eq:bspl_2nd_deriv}, we obtain the second derivative at the knots,
\begin{align*}
    f''(\tau_i) &= A_{i-1} =\frac{\gamma_{i-1}-\gamma_{i-2}}{h_{i-1}^{(3)}} - \frac{\gamma_{i-2}-\gamma_{i-3}}{h_{i-2}^{(3)}}\,,\\
    f''(\tau_{i+1}) &= A_i = \frac{\gamma_i-\gamma_{i-1}}{h_i^{(3)}} - \frac{\gamma_{i-1}-\gamma_{i-2}}{h_{i-1}^{(3)}}\,.
\end{align*}
To have an upper bound on the second derivatives,
\begin{align*}
    \vert f''(\tau_i) \vert &=\frac{\vert\gamma_{i-1}-\gamma_{i-2}\vert}{h_{i-1}^{(3)}} + \frac{\gamma_{i-2}-\gamma_{i-3}}{h_{i-2}^{(3)}}  \le L\,,\\
    \vert f''(\tau_{i+1})\vert &= \frac{\vert\gamma_i-\gamma_{i-1}\vert}{h_i^{(3)}} + \frac{\gamma_{i-1}-\gamma_{i-2}}{h_{i-1}^{(3)}} \le L\,,
\end{align*}
which implies that
$$
\gamma_{i}-\gamma_{i-1} \le L h_{i}^{(3)}\,,\qquad \gamma_{i-2} -\gamma_{i-3} \le Lh_{i-2}^{(3)}\,,\qquad \vert \gamma_{i-1}-\gamma_{i-2}\vert \le Lh_{i-1}^{(3)}\,.
$$
Then \eqref{eq:gamma_diff_i_2} becomes
$$
\gamma_{i-2}-\gamma_{i-1} \le C_iL\sqrt{h_{i}^{(3)}h_{i-2}^{(3)}}\,,
$$
and hence
$$
\gamma_{i-2}-\gamma_{i-1}\le L\cdot \min\left(h_{i-1}^{(3)}, C_i\sqrt{h_i^{(3)}h_{i-2}^{(3)}}\right)\,.
$$
Note that 
\begin{align*}
    h_2^{(3)} &= h_4\le \frac{3}{J-3}(1+\eta_2)\,,\\
    h_3^{(3)} &= h_4^{(2)} \le \frac{3}{J-3}(1+\eta_2)\,,\\
    h_{i-1}^{(3)} &\le \frac{3}{J-3}(1+\eta_2)\,, i\in [5, J-2]\,,\\
    h_{J-1}^{(3)} &= h_{J-1}^{(2)}\le \frac{3}{J-3}(1+\eta_2)\,,\\
    h_J^{(3)} & = h_J  \le \frac{3}{J-3}(1+\eta_2)\,,
\end{align*}
On the other hand, if $\gamma_{i-1}-\gamma_{i-2}\ge 0$, that is,
$$
\gamma_{i-2}-\gamma_{i-1} \le 0\,.
$$
Thus, for all $i\in [3, J+1]$, we have
$$
\gamma_{i-2} - \gamma_{i-1} \le \frac{3L}{J-3}(1+\eta_2)\,.
$$
By Lemma 4 of \textcite{yangContractionUniformConvergence2019}, we have
$$
\Vert \gamma-\iso(\gamma)\Vert_\infty \le \frac{3L}{J-3}(1+\eta_2)\,,
$$
and hence
\begin{equation}
\Vert\gamma -\iso(\gamma)\Vert_2^2\le J\Vert\gamma-\iso(\gamma)\Vert_\infty^2 \le \frac{9JL^2(1+\eta_2)^2}{(J-3)^2}\,.\label{eq:bound_gamma_iso}
\end{equation}
Plug \eqref{eq:bound_sum_bji} and \eqref{eq:bound_gamma_iso} into \eqref{ieq:2}, we achieve
\begin{equation}\label{eq:bound_no_error}
\frac 1n\Vert \bfB\gamma - \bfB\hat\gamma\Vert_2^2\le\frac 1n\Vert \bfB\gamma - \bfB\iso(\gamma)\Vert_2^2\le \frac{36(1+\eta_1)(1+\eta_2)^2L^2J}{(J-3)^3} = O(J^{-2})\,.    
\end{equation}
  
\end{proof}

\subsection{Nonzero error: $\sigma^2 > 0$}\label{sec:proof_bound_with_e}

\begin{proof}
When $\sigma^2 >0$, then $\bfy = \bfB\gamma + \bfe$, then \eqref{eq:error_suff_vs_iso} becomes
\begin{equation}
    \Vert \bfB\gamma+\bfe -\bfB\hat\gamma\Vert_2^2 \le \Vert \bfB\gamma+\bfe - \bfB\iso(\gamma)\Vert_2^2\,,
\end{equation}
It follows that
$$
\Vert \bfB\gamma-\bfB\hat\gamma\Vert_2^2+\Vert\bfe\Vert_2^2 + 2\bfe^T\bfB(\gamma-\hat\gamma)\le \Vert \bfB\gamma-\bfB\iso(\gamma)\Vert_2^2+\Vert\bfe\Vert_2^2 + 2\bfe^T\bfB(\gamma-\iso(\gamma))\,,
$$
that is
\begin{equation}\label{eq:error_ieq_with_e}
\frac 1n\Vert \bfB\gamma-\bfB\hat\gamma\Vert_2^2 \le \frac 1n\Vert \bfB\gamma-\bfB\iso(\gamma)\Vert_2^2 + \frac 2n\bfe^T\bfB(\hat\gamma-\iso(\gamma))\,.    
\end{equation}
Note that $\hat\gamma$ is monotonic, then $\hat\gamma = \iso(\hat\gamma)$. By the contraction of Isotonic regression \parencite{yangContractionUniformConvergence2019}, we have
$$
\Vert\hat\gamma-\iso(\gamma)\Vert_p=\Vert\iso(\hat\gamma)-\iso(\gamma)\Vert_p \le \Vert\hat\gamma-\gamma\Vert_p
$$
for $L_p, p\in [1, \infty]$ norm.
First of all, we have a lower bound for the left-hand side of \eqref{eq:error_ieq_with_e},
$$
\frac 1n\Vert \bfB\gamma-\bfB\hat\gamma\Vert_2^2 \ge \lambda_{\min}\left(\frac 1n\bfB^T\bfB\right)\Vert\gamma-\hat\gamma\Vert_2^2
$$
By \textcite{shenLocalAsymptoticsRegression1998}, there is a constant $c_1 > 0$ such that the minimum eigenvalues of $\frac 1n\bfB^T\bfB$ satisfies
\begin{equation}\label{eq:small_eigenvalue}
\lambda_{\min}\left(\frac 1n\bfB^T\bfB\right) \ge \frac{c_1}{J} \end{equation}
for sufficiently large $n$. On the other hand, consider the upper bound of the right-hand side of \eqref{eq:error_ieq_with_e},
\begin{equation}
\frac 2n\bfe^T\bfB(\hat\gamma-\iso(\gamma))\le 2\max_{1\le j\le J}\frac 1n\bfe^T\bfB_j \Vert\hat\gamma-\iso(\gamma)\Vert_1 \triangleq 2A\Vert \hat\gamma-\iso(\gamma)\Vert_1\,,    
\label{eq:rhs2}
\end{equation}
and by \eqref{eq:bound_no_error},
\begin{equation}
    \frac 1n\Vert \bfB\gamma-\bfB\iso(\gamma)\Vert_2^2 \le c_2J^{-2}\,,
    \label{eq:rhs1}
\end{equation}
where $c_2$ is a constant. 
Combine \eqref{eq:error_ieq_with_e}, \eqref{eq:rhs2}, and \eqref{eq:rhs1}, we have
\begin{equation}
    \frac{c_1}{J}\Vert\delta\Vert_2^2\le c_2J^{-2} + 2A\Vert \delta\Vert_1\,.
\end{equation}
Also note that $\Vert\delta\Vert_2\ge \frac{1}{\sqrt J}\Vert\delta\Vert_1$, it follows that
$$
c_1J^{-2}\Vert \delta\Vert_1^2 \le c_2J^{-2} + 2A\Vert \delta\Vert_1
$$
it follows that
\begin{equation}\label{eq:delta_1norm}
\Vert\delta\Vert_1 \le \frac{A+\sqrt{A^2+c_1c_2J^{-4}}}{c_1J^{-2}}\,.    
\end{equation}
\begin{lemma}
    If $\varepsilon_i$ is assumed to be Gaussian, then it is also sub-Gaussian. It follows that
$$
\bbE\exp(t\varepsilon_i) \le \exp(\sigma^2t^2/2)\quad\forall t\in\IR\,.
$$
We have the well-known tail bound,
$$
\Pr(\vert \varepsilon_i\vert > z) \le 2\exp\left(-\frac{z^2}{2\sigma^2}\right),\;z\ge 0\,.
$$
And note that $\sum_{i=1}^nv_i\varepsilon_i$ is sub-Gaussian with parameter $\Vert v\Vert_2^2\sigma^2$ for $v\in \IR^n$. 
Let $\varepsilon=[\varepsilon_1,\ldots,\varepsilon_n]$, for any collection of vectors $v_i\in \IR^n, i=1,\ldots,J$, we have \parencite{slawskiNonnegativeLeastSquares2013}
$$
\Pr\left(\max_{1\le j\le J}\vert v_j^T\varepsilon\vert > \sigma \max_{1\le j\le J}\Vert v_j\Vert_2(\sqrt{2\log J}+z) \right)\le 2\exp\left(-\frac 12z^2\right), z\ge 0\,.
$$
\end{lemma}
By the above tail bounds for sub-Gaussian random variables, take $v_j=\frac 1n\bfB_j, \varepsilon = \bfe$, and note that 
\begin{align*}
    \Vert\bfB_j\Vert_2 = \sqrt{\sum_{i=1}^nB_j^2(x_i)} \le \sqrt{\sum_{i=1}^nB_j(x_i)}\le \sqrt{\frac{4n(1+\eta_1)}{J-3}}\,,
\end{align*}
then
$$
\Vert v_j\Vert_2 \le \sqrt{\frac{4(1+\eta_1)}{n(J-3)}}\,.
$$
Pick $z=M\sqrt{2\log J}$ for $M\ge 0$, we obtain
\begin{equation}
\Pr\left(\max_{1\le j\le J}\left\vert \frac 1n\bfe^T\bfB_j\right\vert > 2\sigma(1+M)\sqrt{(1+\eta_1)}\sqrt{\frac{2\log J}{n(J-3)}}\right)\le 2J^{-M^2}\,.    \label{eq:prob_bound}
\end{equation}
If we take $J = n^{1/3}$, \eqref{eq:prob_bound} becomes
$$
\Pr(A > c_3 \sqrt{\log J}J^{-2})\le 2J^{-M^2}\,.
$$
Apply the probability bound on \eqref{eq:delta_1norm}, then it holds with at least probability $1-2J^{-M^2}$ that
\begin{equation}
\Vert \delta\Vert_1\le \frac{2c_3+o(1)}{c_1}\sqrt{\log J}\,.    
\label{eq:prob_bound_delta_1norm}
\end{equation}
Plug \eqref{eq:prob_bound_delta_1norm} and \eqref{eq:rhs1} into \eqref{eq:error_ieq_with_e}, it holds with at least probability $1-2J^{-M^2}$ that
\begin{align*}
&\frac 1n\Vert\bfB\tilde\gamma-\bfB\hat\gamma\Vert_2^2\\
&\le\frac{36(1+\eta_1)(1+\eta_2)^2L^2J}{(J-3)^3} + 2\cdot \left(2\sigma(1+M)\sqrt{(1+\eta_1)}\sqrt{\frac{2\log J}{n(J-3)}}\right)^2\cdot \frac{2\sqrt{\log J}}{c_1}\\
&\le\frac{36(1+\eta_1)(1+\eta_2)^2L^2J}{(J-3)^3} + \frac{32}{c_1}\sigma^2(1+M)^2(1+\eta_1)\frac{\log J}{(J-3)^2}\\
&\le O(J^{-2}) + O\left(\log J\cdot J^{-2}\right) = O\left(\frac{\log J}{J^2}\right)\,.    
\end{align*}

\end{proof}

\section{Proof of Theorem 2}

\subsection{Monotone spline fitting $\tilde g_n$ based on sufficient and necessary condition}

\begin{proof}

In the setting with error variable, $y=g(x)+\varepsilon, \varepsilon\sim N(0, \sigma^2)$, \textcite{shenLocalAsymptoticsRegression1998} showed that the bias of unconstrained B-spline fitting is
$$
\bbE \check g_n(x) - g(x) = b(x) + o(J^{-4})\,,
$$
where 
$$
\sup_{x\in [0, 1]}\vert b(x)\vert = O(J^{-4})
$$
Now there is no error, i.e., $\sigma^2=0$, we have
$$
\sup_{x\in [0, 1]}\vert \check g_n(x) - g(x)\vert = O(J^{-4})\,.
$$
It implies that $\check g_n(x)$ converges uniformly to $g$. 
Similarly, the first derivative \parencite{zhouDerivativeEstimationSpline2000} also converges uniformly to $g'$.
$$
\sup_{x\in[0, 1]} \vert \check g_n'(x) - g'(x)\vert = O(J^{-3})\,.
$$

Now since $g$ is strictly increasing. Let $\varepsilon_0 = \min\{g'(x), x\in [0, 1]\} > 0$. Due to the first derivative of the unconstrained B-spline fitting $\check g_n'$ converges uniformly to $g'$, then as $n$ becomes sufficiently large, say $n > n(\varepsilon_0)$, we have
  $$
  \vert \check g_n'(x) - g'(x)\vert < \varepsilon_0/2\,,
  $$
  for any $x$, then
  $$
  \varepsilon_0/2 \le g'(x) - \varepsilon_0/2 < \check g_n'(x) < g'(x) + \varepsilon_0/2\,.
  $$
  Since $\check g_n'(x)>0$, this implies that $\check g_n$ is actually monotone. In other words, the monotone spline fitting $\tilde g_n$ and the unconstrained B-spline fitting $\check g_n$ are identical when $n > n(\varepsilon_0)$. 

  Thus, when $n> n(\varepsilon_0)$,
  $$
  \frac 1n\sum_{i=1}^n (\tilde g_n(x_i)-g(x_i))^2 = \frac 1n\sum_{i=1}^n (\check g_n(x_i)-g(x_i))^2 = O(J^{-8})\,.
  $$
\end{proof}

\subsection{Monotone spline fitting $\hat g_n$ based on sufficient condition}
\begin{proof}
Let
\begin{align*}
    \hat \bfg_n &= [\hat g_n(x_1),\ldots, \hat g_n(x_n)]^T\\
    \check \bfg_n &= [\check g_n(x_1),\ldots, \check g_n(x_n)]^T\\
    \tilde \bfg_n &= [\tilde g_n(x_1),\ldots, \tilde g_n(x_n)]^T\\
    \bfg_n &= [g_n(x_1),\ldots, g_n(x_n)]^T\,,
\end{align*}
Note that
\begin{align*}
\hat\gamma &=\argmin_{\bA\gamma\le 0}\Vert \bfg_n -\bfB\gamma\Vert_2^2\\
&=\argmin_{\bA\gamma\le 0}\Vert \tilde\bfg_n+\bfg_n -\tilde \bfg_n - \bfB\gamma\Vert_2^2\\
&\triangleq \argmin_{\bA\gamma\le 0}\Vert \tilde \bfg_n+\bfe -\bfB\gamma\Vert_2^2\,.
\end{align*}
Since $\bfe$ can be viewed as an observed error for the random error studied in Section~\ref{sec:proof_bound_with_e}, we can follow the proof procedure to obtain the error bound between $\hat\bfg_n=\bfB\hat\gamma$ and $\tilde\bfg_n$.
By \eqref{eq:error_ieq_with_e} and \eqref{eq:delta_1norm}, we have
$$
\frac 1n\Vert\hat \bfg_n-\tilde \bfg_n\Vert_2^2 \le c_2J^{-2} + 2A\cdot \frac{A+\sqrt{A^2+c_1c_2J^{-4}}}{c_1J^{-2}}\,,
$$
with
$$
A = \max_{1\le j\le J}\frac 1n\bfe^T\bfB_j\le \frac 1n \max_{1\le j\le J}\left\{ \max_{1\le i\le n} \bfe_i \cdot \Vert\bfB_j\Vert_1\right\} \le \frac 1n O(J^{-4}) \cdot \frac{4n(1+\eta_1)}{J-3} = O(J^{-5})\,,
$$
where $\max_{1\le i\le n}\bfe_i = O(J^{-4})$ from the previous section.
Thus,
$$
\frac 1n \Vert\hat \bfg_n-\tilde \bfg_n\Vert_2^2 \le c_2J^{-2} + O(J^{-5})\cdot O(1) = O(J^{-2})\,.
$$
It follows that
\begin{align*}
    \frac 1n\Vert \hat\bfg_n - \bfg_n\Vert_2^2 &= \frac 1n\Vert \hat\bfg_n-\tilde\bfg_n+\tilde \bfg_n-\bfg_n\Vert_2^2\\
    &= \frac 1n\Vert \hat\bfg_n-\tilde\bfg_n\Vert_2^2 + \frac 1n\Vert \tilde \bfg_n-\bfg_n\Vert_2^2 + \frac 2n(\hat\bfg_n-\tilde\bfg_n)^T(\tilde \bfg_n - \bfg_n)\\
    &=O(J^{-2}) + O(J^{-8}) + O(J^{-1})\cdot O(J^{-4})\\
    &=O(J^{-2})\,.
\end{align*}
\end{proof}

\section{Proof of Theorem 3}
\begin{proof}
Based on KKT condition, \textcite{lawsonSolvingLeastSquares1995} shows that  there exists a $(J-1)\times 1$ vector $\hat \eta$ and a partition $\cE, \cS$ such that
$$
(-\bA)^T\hat\eta = \bfB^T(\bfB\hat\gamma - y)
$$
that is
$$
\bA^T\hat \eta = \bfB^Ty - \bfB^T\bfB\hat\gamma
$$
And let $\hat r = -\bA\hat\gamma$, we have
\begin{align}
\hat r_i = 0, i\in \cE&\qquad \hat r_i > 0, i\in \cS\,,\\
\hat \eta_i\ge 0, i\in \cE&\qquad \hat \eta_i=0, i\in \cS\,,
\end{align}
that is,
$$
\bA_\cE\hat\gamma = 0, \bA_\cS\hat\gamma < 0\,.
$$
It implies that the constrained solution is a minimizer of a least squares problem subject to the equality constraint $\bA_\cE\gamma = 0$ given the set $\cE$, that is
$$
\Vert \bfy - \bfB\hat\gamma\Vert_2^2 = \min_\gamma \Vert \bfy - \bfB\gamma\Vert_2^2 \quad \text{s.t.}\; \bA_\cE\gamma = 0\,.
$$
Plug $\gamma=\bG^T\beta$ into the above problem, then
$$
\hat\gamma = \bG^T\hat\beta = \bG^T(\bG\bB^T\bB\bG^T)^{-1}\bG\bB^T\bfy\,.
$$
Similarly, for monotone smoothing splines, there is also a set $\cE$ such that $\bA_\cE\hat\gamma=0$, then
$$
\Vert \bfy-\bfB\hat\gamma\Vert_2^2 +\lambda\hat\gamma^T\bOmega\hat\gamma = \min_{\gamma} \Vert \bfy-\bfB\gamma\Vert_2^2 + \lambda\gamma^T\bOmega\gamma\quad\text{s.t.}\; \bA_\cE\gamma=0\,,
$$
then
$$
\hat\gamma = \bG^T\hat\beta = \bG^T(\bG\bB^T\bB\bG^T + \lambda\bG\bOmega\bG^T)^{-1}\bG\bB^T\bfy\,.
$$
\end{proof}

\section{Proof of Theorem 4}
\begin{proof}
The monotone fitting can be written as
$$
\hat\bff = \bB\hat\gamma = \bB\bG^T(\bG\bB^T\bB\bG^T)^{-1}\bG\bB^T\bfy \triangleq \bH_g\bfy\,.
$$
Note that
$$
\bbE\hat\bff = \bH_g\bbE\bfy = \bH_g\bff\,.
$$
Then the squared bias is
\begin{align*}
    \Bias(\hat\bff) = \bbE\Vert \bff-\bbE\hat\bff\Vert_2^2 = \Vert \bff-\bH_g\bff\Vert_2^2 = \bff^T(\bI-\bH_g)\bff\,,
\end{align*}
and the variance is
$$
\Var(\hat\bff) = \sigma^2\bH_g\bH_g^T\,.
$$
Thus, the mean square error (MSE) is
$$
\MSE(\hat\bff) = \Bias^2(\hat\bff) + \tr[\Var(\hat\bff)] = \bff^T(\bI-\bH_g)\bff+g\sigma^2\,.
$$
On the other hand, the MSE for the unconstrained solution is
$$
\MSE(\hat\bff^\ls) = \bff^T(\bI-\bH)\bff+J\sigma^2\,,
$$
where $\bfH = \bfB(\bfB^T\bfB)^{-1}\bfB^T$.
To have a smaller MSE, we want
$$
\MSE(\hat\bff) < \MSE(\hat\bff^\ls)\,,
$$
that is,
$$
\bff^T(\bH-\bH_g)\bff +g\sigma^2 \le J\sigma^2\,,
$$
thus,
$$
\sigma^2 \ge \frac{\bff^T(\bH-\bH_g)\bff}{J-g}\,.
$$

Note that
\begin{align*}
(\bH-\bH_g)^2 &= \bfH^2 -\bH\bH_g -\bH_g\bH+\bH_g^2    \\
&=\bfH -\bfH_g -\bfH_g +\bfH_g\\
&=\bfH-\bfH_g\,,
\end{align*}
which implies that $\bfH-\bfH_g$ is an idempotent matrix and hence is positive semi-definite.



\section{Proof of Proposition 2}
\begin{proof}
Let $\xi = \gamma$ and $\theta = \bfB\gamma$, then

$$
\arg\min_\beta \Vert \bfy-\bfB\gamma\Vert^2_2 \quad \text{s.t.}\quad \bA\gamma\le 0\,,
$$
where $\bA$ is a $(J-1)\times J$ matrix,
$$
\bA_{ij} = \begin{cases}
1 & i=j=1,\ldots,J-1\\
-1 & j=i+1=2,\ldots, J\\
0 & \text{otherwise}
\end{cases}\,.
$$

The object function can be rewritten as
\begin{align*}
\arg\min_\beta \Vert \bfy-\theta\Vert^2_2\\
\bfB\xi - I_n\theta \le 0\\
-\bfB\xi + I_n\theta \le 0\\
\bA\xi + \zero_n\theta \le 0
\end{align*}
and let
$$
A = \begin{bmatrix}
\bfB_{n\times J}\\
-\bfB_{n\times J}\\
\bA_{(J-1)\times J}
\end{bmatrix}\,,
\qquad
B = \begin{bmatrix}
-I_n\\
I_n\\
\zero_{(J-1)\times n}
\end{bmatrix}\,.
$$

Note that the first $2n$ rows would always be in the index set $J_\bfy$, and $I_\bfy$ would take $n$ linearly independent rows from them. If there are $m_\bfy$ (depends on $\bfy$) equal adjacent pairs of $\beta$, and these corresponding row vectors are also linearly independent with the first $2n$ rows, then by Theorem 3.2 of \textcite{chenDegreesFreedomProjection2020},
$$
\vert I_\bfy\vert = n + m_\bfy\,.
$$
If $n > p$, then we always have $\rank(A_{I_\bfy}) = p$. Thus, the \emph{divergence} is
$$
D(\bfy) = n - (n + m_\bfy) + p = p - m_\bfy \triangleq U_\bfy\,,
$$
where $U_y$ is the number of unique coefficients, then
$$
\df = \bbE[D(\bfy)] = p - \bbE[m_\bfy] = \bbE[U_\bfy]\,,
$$
where the randomness comes from the index set $I_y$.
\end{proof}

\section{Jaccard Index and Coverage Probability when $\sigma=0.1, 0.5$}

\begin{figure}[H]
    \centering
    \begin{subfigure}{0.5\textwidth}
    \includegraphics[width=\textwidth]{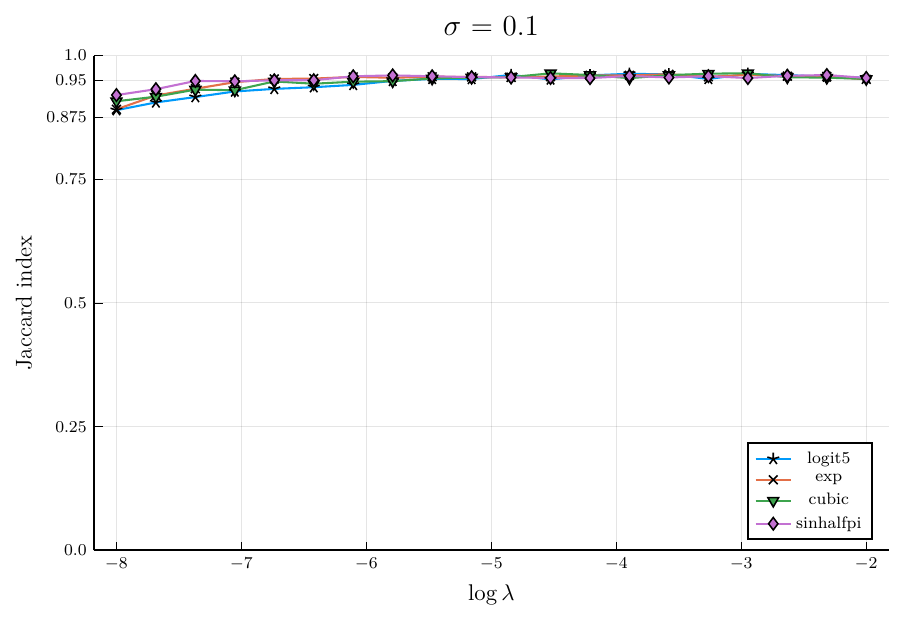}
    \end{subfigure}%
    \begin{subfigure}{0.5\textwidth}
    \includegraphics[width=\textwidth]{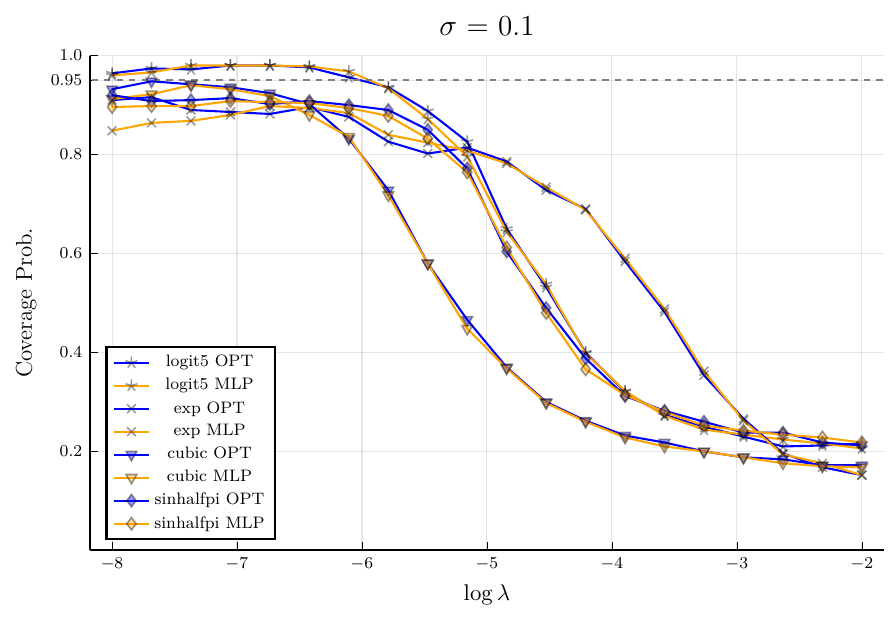}
    \end{subfigure}
    \caption{Jaccard index and Coverage probability for each curve among 5 repetitions when noise level $\sigma = 0.1$}
    \label{fig:ci_0.1}
\end{figure}

\begin{figure}[H]
    \centering
    \begin{subfigure}{0.5\textwidth}
    \includegraphics[width=\textwidth]{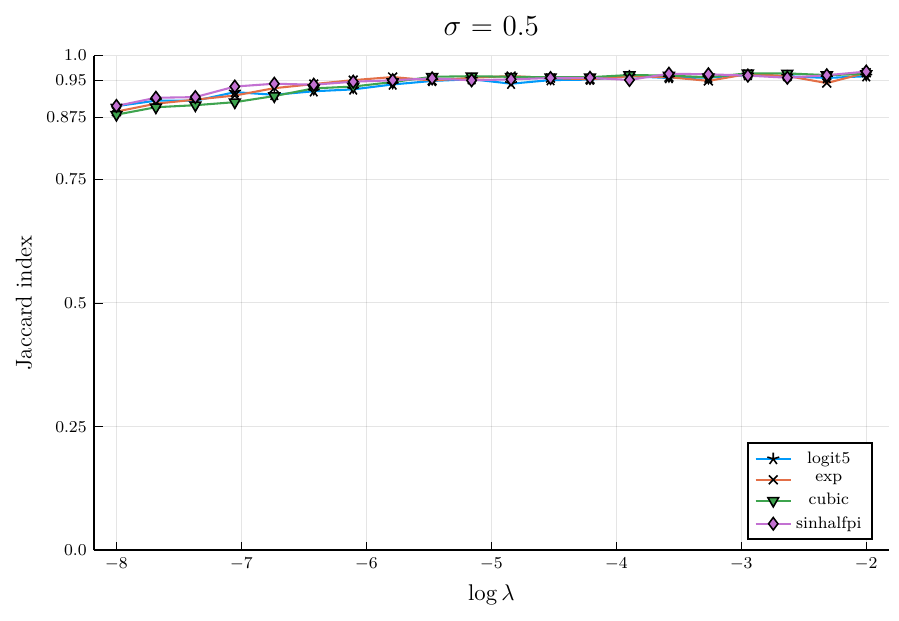}
    \end{subfigure}%
    \begin{subfigure}{0.5\textwidth}
    \includegraphics[width=\textwidth]{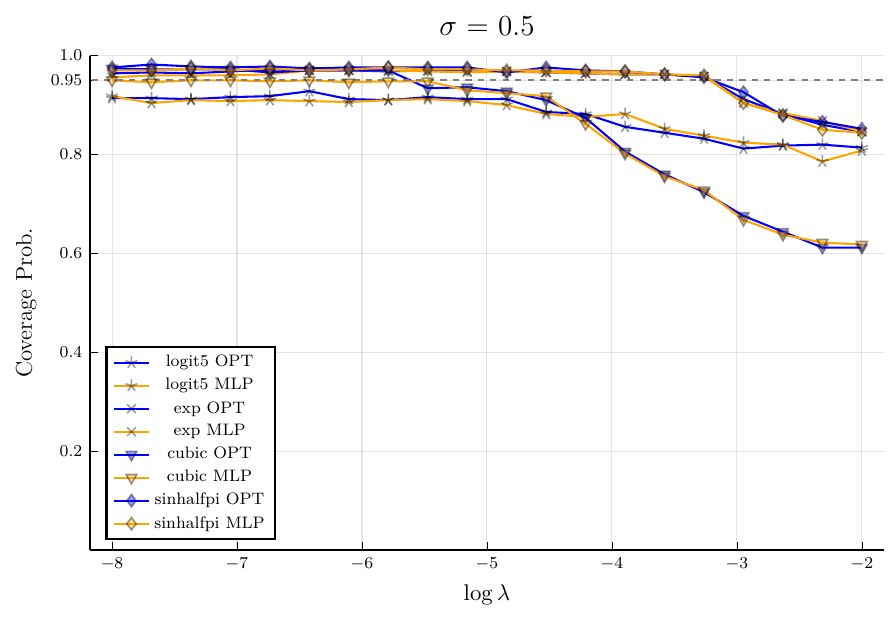}
    \end{subfigure}
    \caption{Jaccard index and Coverage probability for each curve among 5 repetitions when noise level $\sigma = 0.5$}
    \label{fig:ci_0.5}
\end{figure}

\section{$L_p$ Distances for Cubic, Step, and Growth Curves}

Tables \ref{tab:cubic}, \ref{tab:step} and \ref{tab:growth} presents the $L_p$ distances among 100 experiments on the cubic, step, and growth curves with different noise levels. Although our proposed MSS and MCS might not be the best in all cases, they are always comparable with top ranks.

\begin{table}[H]
  \caption{Average (scaled) $L_p$ distances, $p\in \{1,2,\infty\}$, over 100 experiments on the cubic curve, together with the standard error of the average in parentheses. Both the smallest one and the ones whose errors are no more than one standard error above the error of the smallest one are highlighted in bold. The superscripts indicate the rank of methods.}
  \label{tab:cubic}
\resizebox{\textwidth}{!}{
  \begin{tabular}{ccccc}
\toprule
Noise $\sigma$ & Method&$\frac 1n L_1$&$\frac{1}{\sqrt n}L_2$&$L_\infty$\tabularnewline
\midrule
\multirow{14}{*}{0.1}&Cubic Spline (CS)& 1.96e-02 (7.5e-04)\textsuperscript{4}& 2.47e-02 (9.5e-04)\textsuperscript{4}& 7.09e-02 (4.2e-03)\textsuperscript{3}\tabularnewline
&Monotone CS (MCS)& 1.94e-02 (7.0e-04)\textsuperscript{3}& 2.46e-02 (8.7e-04)\textsuperscript{3}& 7.15e-02 (3.8e-03)\textsuperscript{4}\tabularnewline
&Smoothing Spline (SS)& 2.23e-02 (5.9e-04)\textsuperscript{6}& 2.81e-02 (7.3e-04)\textsuperscript{6}& 8.25e-02 (3.6e-03)\textsuperscript{7}\tabularnewline
&Montone SS (MSS)& 2.04e-02 (6.1e-04)\textsuperscript{5}& 2.66e-02 (7.2e-04)\textsuperscript{5}& 8.13e-02 (3.6e-03)\textsuperscript{5}\tabularnewline
&Quadratic Spline (QS)& 1.91e-02 (5.1e-04)\textsuperscript{2}& 2.40e-02 (6.2e-04)\textsuperscript{2}& 6.18e-02 (2.5e-03)\textsuperscript{2}\tabularnewline
&\textcite{heMonotoneBsplineSmoothing1998}: MQS& 2.36e-02 (7.1e-04)\textsuperscript{7}& 3.04e-02 (8.6e-04)\textsuperscript{7}& 8.23e-02 (3.9e-03)\textsuperscript{6}\tabularnewline
&LOESS& 2.49e-02 (6.9e-04)\textsuperscript{10}& 3.16e-02 (7.8e-04)\textsuperscript{9}& 8.88e-02 (3.2e-03)\textsuperscript{8}\tabularnewline
&Isotonic& 3.39e-02 (5.4e-04)\textsuperscript{13}& 4.55e-02 (6.0e-04)\textsuperscript{14}& 1.54e-01 (3.5e-03)\textsuperscript{14}\tabularnewline
&\textcite{mammenEstimatingSmoothMonotone1991}: SI (LOESS+Isotonic)& 2.45e-02 (7.2e-04)\textsuperscript{8}& 3.13e-02 (8.0e-04)\textsuperscript{8}& 8.88e-02 (3.2e-03)\textsuperscript{9}\tabularnewline
&\textcite{mammenEstimatingSmoothMonotone1991}: IS (Isotonic+LOESS)& 2.64e-02 (7.1e-04)\textsuperscript{11}& 3.34e-02 (7.7e-04)\textsuperscript{11}& 9.02e-02 (3.1e-03)\textsuperscript{10}\tabularnewline
&\textcite{murrayFastFlexibleMethods2016}: MonoPoly& \textbf{1.40e-02} (5.7e-04)\textsuperscript{1}& \textbf{1.75e-02} (7.1e-04)\textsuperscript{1}& \textbf{4.34e-02} (2.4e-03)\textsuperscript{1}\tabularnewline
&\textcite{cannonMonmlpMultilayerPerceptron2017}: MONMLP& 2.97e-02 (2.4e-03)\textsuperscript{12}& 3.79e-02 (2.8e-03)\textsuperscript{12}& 1.17e-01 (7.9e-03)\textsuperscript{13}\tabularnewline
&\textcite{navarro-garciaConstrainedSmoothingOutofrange2023}: cpsplines& 2.46e-02 (6.0e-04)\textsuperscript{9}& 3.17e-02 (6.9e-04)\textsuperscript{10}& 9.35e-02 (3.3e-03)\textsuperscript{11}\tabularnewline
&\textcite{groeneboomConfidenceIntervalsMonotone2023}: SLSE& 3.40e-02 (1.0e-03)\textsuperscript{14}& 4.32e-02 (1.3e-03)\textsuperscript{13}& 1.02e-01 (4.0e-03)\textsuperscript{12}\tabularnewline
\midrule
\multirow{14}{*}{1.0}&Cubic Spline (CS)& 2.13e-01 (1.1e-02)\textsuperscript{13}& 2.68e-01 (1.4e-02)\textsuperscript{13}& 7.58e-01 (5.6e-02)\textsuperscript{12}\tabularnewline
&Monotone CS (MCS)& 1.60e-01 (6.0e-03)\textsuperscript{5}& 2.07e-01 (7.3e-03)\textsuperscript{7}& 6.01e-01 (3.7e-02)\textsuperscript{11}\tabularnewline
&Smoothing Spline (SS)& 1.74e-01 (6.3e-03)\textsuperscript{9}& 2.17e-01 (7.6e-03)\textsuperscript{9}& 5.31e-01 (2.3e-02)\textsuperscript{7}\tabularnewline
&Montone SS (MSS)& 1.58e-01 (5.6e-03)\textsuperscript{4}& 2.01e-01 (6.5e-03)\textsuperscript{5}& 5.03e-01 (2.1e-02)\textsuperscript{5}\tabularnewline
&Quadratic Spline (QS)& 1.80e-01 (5.4e-03)\textsuperscript{10}& 2.21e-01 (6.1e-03)\textsuperscript{10}& 5.55e-01 (2.1e-02)\textsuperscript{8}\tabularnewline
&\textcite{heMonotoneBsplineSmoothing1998}: MQS& 1.85e-01 (6.1e-03)\textsuperscript{11}& 2.31e-01 (7.4e-03)\textsuperscript{11}& 5.93e-01 (2.7e-02)\textsuperscript{10}\tabularnewline
&LOESS& 1.63e-01 (5.9e-03)\textsuperscript{7}& 2.00e-01 (7.1e-03)\textsuperscript{4}& 4.62e-01 (2.1e-02)\textsuperscript{4}\tabularnewline
&Isotonic& 1.88e-01 (5.8e-03)\textsuperscript{12}& 2.58e-01 (7.2e-03)\textsuperscript{12}& 9.81e-01 (5.1e-02)\textsuperscript{14}\tabularnewline
&\textcite{mammenEstimatingSmoothMonotone1991}: SI (LOESS+Isotonic)& 1.51e-01 (6.0e-03)\textsuperscript{2}& 1.87e-01 (6.9e-03)\textsuperscript{2}& 4.39e-01 (1.9e-02)\textsuperscript{3}\tabularnewline
&\textcite{mammenEstimatingSmoothMonotone1991}: IS (Isotonic+LOESS)& 1.55e-01 (6.2e-03)\textsuperscript{3}& 1.92e-01 (7.2e-03)\textsuperscript{3}& \textbf{4.27e-01} (1.9e-02)\textsuperscript{2}\tabularnewline
&\textcite{murrayFastFlexibleMethods2016}: MonoPoly& \textbf{1.39e-01} (6.4e-03)\textsuperscript{1}& \textbf{1.72e-01} (7.1e-03)\textsuperscript{1}& \textbf{4.10e-01} (1.9e-02)\textsuperscript{1}\tabularnewline
&\textcite{cannonMonmlpMultilayerPerceptron2017}: MONMLP& 2.43e-01 (5.6e-03)\textsuperscript{14}& 3.01e-01 (6.3e-03)\textsuperscript{14}& 7.90e-01 (3.5e-02)\textsuperscript{13}\tabularnewline
&\textcite{navarro-garciaConstrainedSmoothingOutofrange2023}: cpsplines& 1.62e-01 (9.0e-03)\textsuperscript{6}& 2.07e-01 (1.1e-02)\textsuperscript{6}& 5.30e-01 (3.9e-02)\textsuperscript{6}\tabularnewline
&\textcite{groeneboomConfidenceIntervalsMonotone2023}: SLSE& 1.67e-01 (6.7e-03)\textsuperscript{8}& 2.14e-01 (8.7e-03)\textsuperscript{8}& 5.67e-01 (3.0e-02)\textsuperscript{9}\tabularnewline
\midrule
\multirow{14}{*}{1.5}&Cubic Spline (CS)& 2.91e-01 (1.3e-02)\textsuperscript{14}& 3.69e-01 (1.7e-02)\textsuperscript{14}& 1.06e+00 (7.1e-02)\textsuperscript{13}\tabularnewline
&Monotone CS (MCS)& 2.15e-01 (6.9e-03)\textsuperscript{3}& 2.67e-01 (8.2e-03)\textsuperscript{4}& 7.07e-01 (3.7e-02)\textsuperscript{8}\tabularnewline
&Smoothing Spline (SS)& 2.25e-01 (8.8e-03)\textsuperscript{7}& 2.78e-01 (1.1e-02)\textsuperscript{6}& 6.42e-01 (3.5e-02)\textsuperscript{5}\tabularnewline
&Montone SS (MSS)& \textbf{2.03e-01} (6.8e-03)\textsuperscript{2}& 2.51e-01 (8.4e-03)\textsuperscript{2}& \textbf{5.83e-01} (2.7e-02)\textsuperscript{2}\tabularnewline
&Quadratic Spline (QS)& 2.25e-01 (6.0e-03)\textsuperscript{6}& 2.80e-01 (7.6e-03)\textsuperscript{7}& 7.04e-01 (2.9e-02)\textsuperscript{7}\tabularnewline
&\textcite{heMonotoneBsplineSmoothing1998}: MQS& 2.47e-01 (9.4e-03)\textsuperscript{10}& 3.03e-01 (1.2e-02)\textsuperscript{10}& 7.01e-01 (3.4e-02)\textsuperscript{6}\tabularnewline
&LOESS& 2.44e-01 (7.1e-03)\textsuperscript{9}& 3.01e-01 (8.8e-03)\textsuperscript{9}& 7.32e-01 (3.3e-02)\textsuperscript{9}\tabularnewline
&Isotonic& 2.53e-01 (7.3e-03)\textsuperscript{11}& 3.43e-01 (1.1e-02)\textsuperscript{11}& 1.29e+00 (7.3e-02)\textsuperscript{14}\tabularnewline
&\textcite{mammenEstimatingSmoothMonotone1991}: SI (LOESS+Isotonic)& 2.16e-01 (6.7e-03)\textsuperscript{4}& 2.67e-01 (8.0e-03)\textsuperscript{3}& 6.41e-01 (2.8e-02)\textsuperscript{4}\tabularnewline
&\textcite{mammenEstimatingSmoothMonotone1991}: IS (Isotonic+LOESS)& 2.22e-01 (7.9e-03)\textsuperscript{5}& 2.71e-01 (1.0e-02)\textsuperscript{5}& \textbf{5.89e-01} (3.2e-02)\textsuperscript{3}\tabularnewline
&\textcite{murrayFastFlexibleMethods2016}: MonoPoly& \textbf{1.97e-01} (6.9e-03)\textsuperscript{1}& \textbf{2.42e-01} (8.3e-03)\textsuperscript{1}& \textbf{5.82e-01} (2.7e-02)\textsuperscript{1}\tabularnewline
&\textcite{cannonMonmlpMultilayerPerceptron2017}: MONMLP& 2.91e-01 (6.9e-03)\textsuperscript{13}& 3.51e-01 (6.7e-03)\textsuperscript{13}& 7.71e-01 (1.6e-02)\textsuperscript{11}\tabularnewline
&\textcite{navarro-garciaConstrainedSmoothingOutofrange2023}: cpsplines& 2.78e-01 (2.7e-02)\textsuperscript{12}& 3.44e-01 (3.4e-02)\textsuperscript{12}& 8.47e-01 (9.4e-02)\textsuperscript{12}\tabularnewline
&\textcite{groeneboomConfidenceIntervalsMonotone2023}: SLSE& 2.34e-01 (9.0e-03)\textsuperscript{8}& 2.96e-01 (1.1e-02)\textsuperscript{8}& 7.41e-01 (3.7e-02)\textsuperscript{10}\tabularnewline
\bottomrule
\end{tabular}
}
\end{table}

\begin{table}[H]
  \caption{Average (scaled) $L_p$ distances, $p\in \{1,2,\infty\}$, over 100 experiments on the step curve, together with the standard error of the average in parentheses. Both the smallest one and the ones whose errors are no more than one standard error above the error of the smallest one are highlighted in bold. The superscripts indicate the rank of methods.}
  \label{tab:step}
\resizebox{\textwidth}{!}{
  \begin{tabular}{ccccc}
\toprule
Noise $\sigma$ & Method&$\frac 1n L_1$&$\frac{1}{\sqrt n}L_2$&$L_\infty$\tabularnewline
\midrule
\multirow{14}{*}{0.1}&Cubic Spline (CS)& 1.24e-01 (5.3e-03)\textsuperscript{5}& 1.77e-01 (6.7e-03)\textsuperscript{5}& 6.43e-01 (2.1e-02)\textsuperscript{5}\tabularnewline
&Monotone CS (MCS)& 1.38e-01 (6.2e-03)\textsuperscript{6}& 2.10e-01 (6.8e-03)\textsuperscript{6}& 7.34e-01 (2.0e-02)\textsuperscript{6}\tabularnewline
&Smoothing Spline (SS)& 7.68e-02 (1.3e-03)\textsuperscript{3}& 1.14e-01 (2.5e-03)\textsuperscript{2}& 4.60e-01 (1.4e-02)\textsuperscript{2}\tabularnewline
&Montone SS (MSS)& 7.39e-02 (1.5e-03)\textsuperscript{2}& 1.29e-01 (2.5e-03)\textsuperscript{3}& 5.40e-01 (1.2e-02)\textsuperscript{3}\tabularnewline
&Quadratic Spline (QS)& 1.92e-01 (3.1e-03)\textsuperscript{8}& 2.66e-01 (4.5e-03)\textsuperscript{8}& 8.75e-01 (2.1e-02)\textsuperscript{8}\tabularnewline
&\textcite{heMonotoneBsplineSmoothing1998}: MQS& 1.54e-01 (2.5e-03)\textsuperscript{7}& 2.52e-01 (3.7e-03)\textsuperscript{7}& 9.76e-01 (2.6e-02)\textsuperscript{9}\tabularnewline
&LOESS& 3.72e-01 (6.0e-03)\textsuperscript{12}& 4.71e-01 (7.8e-03)\textsuperscript{11}& 1.26e+00 (2.8e-02)\textsuperscript{11}\tabularnewline
&Isotonic& \textbf{3.54e-02} (6.3e-04)\textsuperscript{1}& \textbf{5.00e-02} (7.7e-04)\textsuperscript{1}& \textbf{1.90e-01} (5.7e-03)\textsuperscript{1}\tabularnewline
&\textcite{mammenEstimatingSmoothMonotone1991}: SI (LOESS+Isotonic)& 3.71e-01 (6.0e-03)\textsuperscript{11}& 4.71e-01 (7.8e-03)\textsuperscript{10}& 1.26e+00 (2.8e-02)\textsuperscript{12}\tabularnewline
&\textcite{mammenEstimatingSmoothMonotone1991}: IS (Isotonic+LOESS)& 3.72e-01 (6.0e-03)\textsuperscript{13}& 4.71e-01 (7.8e-03)\textsuperscript{12}& 1.26e+00 (2.8e-02)\textsuperscript{10}\tabularnewline
&\textcite{murrayFastFlexibleMethods2016}: MonoPoly& 4.77e-01 (1.3e-02)\textsuperscript{14}& 5.88e-01 (1.5e-02)\textsuperscript{14}& 1.45e+00 (3.7e-02)\textsuperscript{14}\tabularnewline
&\textcite{cannonMonmlpMultilayerPerceptron2017}: MONMLP& 2.04e-01 (1.2e-02)\textsuperscript{9}& 2.78e-01 (1.3e-02)\textsuperscript{9}& 8.28e-01 (2.7e-02)\textsuperscript{7}\tabularnewline
&\textcite{navarro-garciaConstrainedSmoothingOutofrange2023}: cpsplines& 9.63e-02 (2.2e-03)\textsuperscript{4}& 1.60e-01 (3.4e-03)\textsuperscript{4}& 6.10e-01 (1.5e-02)\textsuperscript{4}\tabularnewline
&\textcite{groeneboomConfidenceIntervalsMonotone2023}: SLSE& 3.56e-01 (6.1e-03)\textsuperscript{10}& 4.75e-01 (9.4e-03)\textsuperscript{13}& 1.43e+00 (4.6e-02)\textsuperscript{13}\tabularnewline
\midrule
\multirow{14}{*}{1.0}&Cubic Spline (CS)& 3.94e-01 (8.3e-03)\textsuperscript{7}& 4.97e-01 (9.9e-03)\textsuperscript{7}& 1.40e+00 (3.5e-02)\textsuperscript{12}\tabularnewline
&Monotone CS (MCS)& 3.40e-01 (6.3e-03)\textsuperscript{5}& 4.29e-01 (6.9e-03)\textsuperscript{5}& 1.21e+00 (2.4e-02)\textsuperscript{4}\tabularnewline
&Smoothing Spline (SS)& 3.30e-01 (4.8e-03)\textsuperscript{4}& 4.17e-01 (5.5e-03)\textsuperscript{4}& 1.15e+00 (2.4e-02)\textsuperscript{3}\tabularnewline
&Montone SS (MSS)& \textbf{3.12e-01} (5.2e-03)\textsuperscript{1}& \textbf{3.96e-01} (5.8e-03)\textsuperscript{1}& \textbf{1.11e+00} (2.5e-02)\textsuperscript{2}\tabularnewline
&Quadratic Spline (QS)& 4.04e-01 (6.6e-03)\textsuperscript{9}& 5.08e-01 (8.2e-03)\textsuperscript{9}& 1.37e+00 (3.1e-02)\textsuperscript{10}\tabularnewline
&\textcite{heMonotoneBsplineSmoothing1998}: MQS& 4.02e-01 (6.9e-03)\textsuperscript{8}& 5.04e-01 (8.5e-03)\textsuperscript{8}& 1.38e+00 (3.2e-02)\textsuperscript{11}\tabularnewline
&LOESS& 4.11e-01 (5.7e-03)\textsuperscript{12}& 5.15e-01 (7.1e-03)\textsuperscript{11}& 1.37e+00 (2.9e-02)\textsuperscript{8}\tabularnewline
&Isotonic& 3.19e-01 (6.0e-03)\textsuperscript{3}& 4.14e-01 (6.5e-03)\textsuperscript{3}& 1.30e+00 (3.1e-02)\textsuperscript{5}\tabularnewline
&\textcite{mammenEstimatingSmoothMonotone1991}: SI (LOESS+Isotonic)& 4.10e-01 (5.8e-03)\textsuperscript{11}& 5.14e-01 (7.1e-03)\textsuperscript{10}& 1.37e+00 (2.9e-02)\textsuperscript{9}\tabularnewline
&\textcite{mammenEstimatingSmoothMonotone1991}: IS (Isotonic+LOESS)& 4.17e-01 (6.1e-03)\textsuperscript{13}& 5.21e-01 (7.4e-03)\textsuperscript{12}& 1.37e+00 (2.8e-02)\textsuperscript{7}\tabularnewline
&\textcite{murrayFastFlexibleMethods2016}: MonoPoly& 5.17e-01 (1.3e-02)\textsuperscript{14}& 6.37e-01 (1.5e-02)\textsuperscript{14}& 1.59e+00 (3.5e-02)\textsuperscript{14}\tabularnewline
&\textcite{cannonMonmlpMultilayerPerceptron2017}: MONMLP& 3.59e-01 (7.5e-03)\textsuperscript{6}& 4.56e-01 (8.3e-03)\textsuperscript{6}& 1.30e+00 (3.3e-02)\textsuperscript{6}\tabularnewline
&\textcite{navarro-garciaConstrainedSmoothingOutofrange2023}: cpsplines& \textbf{3.12e-01} (5.0e-03)\textsuperscript{2}& \textbf{3.96e-01} (5.5e-03)\textsuperscript{2}& \textbf{1.11e+00} (2.4e-02)\textsuperscript{1}\tabularnewline
&\textcite{groeneboomConfidenceIntervalsMonotone2023}: SLSE& 4.04e-01 (7.9e-03)\textsuperscript{10}& 5.21e-01 (1.0e-02)\textsuperscript{13}& 1.50e+00 (4.5e-02)\textsuperscript{13}\tabularnewline
\midrule
\multirow{14}{*}{1.5}&Cubic Spline (CS)& 4.85e-01 (1.2e-02)\textsuperscript{11}& 6.04e-01 (1.5e-02)\textsuperscript{12}& 1.60e+00 (5.6e-02)\textsuperscript{13}\tabularnewline
&Monotone CS (MCS)& 4.32e-01 (7.8e-03)\textsuperscript{5}& 5.42e-01 (9.1e-03)\textsuperscript{4}& 1.44e+00 (3.6e-02)\textsuperscript{4}\tabularnewline
&Smoothing Spline (SS)& 4.29e-01 (7.9e-03)\textsuperscript{3}& 5.35e-01 (9.3e-03)\textsuperscript{3}& 1.42e+00 (3.1e-02)\textsuperscript{3}\tabularnewline
&Montone SS (MSS)& \textbf{4.07e-01} (7.9e-03)\textsuperscript{2}& \textbf{5.09e-01} (8.8e-03)\textsuperscript{2}& \textbf{1.37e+00} (2.7e-02)\textsuperscript{2}\tabularnewline
&Quadratic Spline (QS)& 4.85e-01 (9.0e-03)\textsuperscript{12}& 6.03e-01 (1.0e-02)\textsuperscript{11}& 1.55e+00 (3.5e-02)\textsuperscript{9}\tabularnewline
&\textcite{heMonotoneBsplineSmoothing1998}: MQS& 5.07e-01 (9.9e-03)\textsuperscript{13}& 6.30e-01 (1.2e-02)\textsuperscript{13}& 1.57e+00 (3.9e-02)\textsuperscript{10}\tabularnewline
&LOESS& 4.60e-01 (7.7e-03)\textsuperscript{9}& 5.73e-01 (9.5e-03)\textsuperscript{9}& 1.47e+00 (3.2e-02)\textsuperscript{6}\tabularnewline
&Isotonic& 4.32e-01 (7.7e-03)\textsuperscript{4}& 5.62e-01 (8.8e-03)\textsuperscript{5}& 1.72e+00 (4.4e-02)\textsuperscript{14}\tabularnewline
&\textcite{mammenEstimatingSmoothMonotone1991}: SI (LOESS+Isotonic)& 4.58e-01 (7.7e-03)\textsuperscript{8}& 5.71e-01 (9.4e-03)\textsuperscript{8}& 1.47e+00 (3.2e-02)\textsuperscript{5}\tabularnewline
&\textcite{mammenEstimatingSmoothMonotone1991}: IS (Isotonic+LOESS)& 4.65e-01 (7.5e-03)\textsuperscript{10}& 5.79e-01 (9.1e-03)\textsuperscript{10}& 1.48e+00 (3.0e-02)\textsuperscript{7}\tabularnewline
&\textcite{murrayFastFlexibleMethods2016}: MonoPoly& 5.22e-01 (1.1e-02)\textsuperscript{14}& 6.43e-01 (1.3e-02)\textsuperscript{14}& 1.58e+00 (3.8e-02)\textsuperscript{11}\tabularnewline
&\textcite{cannonMonmlpMultilayerPerceptron2017}: MONMLP& 4.53e-01 (9.7e-03)\textsuperscript{7}& 5.71e-01 (1.1e-02)\textsuperscript{7}& 1.54e+00 (3.6e-02)\textsuperscript{8}\tabularnewline
&\textcite{navarro-garciaConstrainedSmoothingOutofrange2023}: cpsplines& \textbf{4.02e-01} (7.3e-03)\textsuperscript{1}& \textbf{5.05e-01} (8.2e-03)\textsuperscript{1}& \textbf{1.37e+00} (2.8e-02)\textsuperscript{1}\tabularnewline
&\textcite{groeneboomConfidenceIntervalsMonotone2023}: SLSE& 4.49e-01 (8.4e-03)\textsuperscript{6}& 5.67e-01 (1.1e-02)\textsuperscript{6}& 1.60e+00 (5.2e-02)\textsuperscript{12}\tabularnewline
\bottomrule
\end{tabular}
}
\end{table}

\begin{table}[H]
  \caption{Average (scaled) $L_p$ distances, $p\in \{1,2,\infty\}$, over 100 experiments on the growth curve, together with the standard error of the average in parentheses. Both the smallest one and the ones whose errors are no more than one standard error above the error of the smallest one are highlighted in bold. The superscripts indicate the rank of methods.}
  \label{tab:growth}
\resizebox{\textwidth}{!}{
  \begin{tabular}{ccccc}
\toprule
Noise $\sigma$ & Method&$\frac 1n L_1$&$\frac{1}{\sqrt n}L_2$&$L_\infty$\tabularnewline
\midrule
\multirow{14}{*}{0.1}&Cubic Spline (CS)& 3.43e-02 (7.1e-04)\textsuperscript{3}& 4.59e-02 (1.0e-03)\textsuperscript{3}& 1.71e-01 (7.7e-03)\textsuperscript{3}\tabularnewline
&Monotone CS (MCS)& 3.37e-02 (6.8e-04)\textsuperscript{2}& 4.51e-02 (9.8e-04)\textsuperscript{2}& 1.67e-01 (7.6e-03)\textsuperscript{2}\tabularnewline
&Smoothing Spline (SS)& 4.40e-02 (6.4e-04)\textsuperscript{7}& 5.71e-02 (7.9e-04)\textsuperscript{6}& 2.08e-01 (6.4e-03)\textsuperscript{6}\tabularnewline
&Montone SS (MSS)& 4.25e-02 (5.9e-04)\textsuperscript{5}& 5.51e-02 (7.3e-04)\textsuperscript{5}& 2.01e-01 (6.7e-03)\textsuperscript{5}\tabularnewline
&Quadratic Spline (QS)& 7.68e-02 (4.0e-03)\textsuperscript{9}& 1.39e-01 (8.2e-03)\textsuperscript{9}& 6.99e-01 (4.6e-02)\textsuperscript{9}\tabularnewline
&\textcite{heMonotoneBsplineSmoothing1998}: MQS& 4.38e-02 (8.0e-04)\textsuperscript{6}& 6.39e-02 (1.6e-03)\textsuperscript{7}& 3.06e-01 (1.8e-02)\textsuperscript{8}\tabularnewline
&LOESS& 4.20e-01 (1.0e-02)\textsuperscript{12}& 8.70e-01 (1.9e-02)\textsuperscript{10}& 4.62e+00 (1.2e-01)\textsuperscript{10}\tabularnewline
&Isotonic& 5.95e-02 (6.2e-04)\textsuperscript{8}& 7.58e-02 (7.8e-04)\textsuperscript{8}& 2.27e-01 (4.1e-03)\textsuperscript{7}\tabularnewline
&\textcite{mammenEstimatingSmoothMonotone1991}: SI (LOESS+Isotonic)& 4.20e-01 (1.0e-02)\textsuperscript{13}& 8.70e-01 (1.9e-02)\textsuperscript{11}& 4.62e+00 (1.2e-01)\textsuperscript{11}\tabularnewline
&\textcite{mammenEstimatingSmoothMonotone1991}: IS (Isotonic+LOESS)& 4.20e-01 (1.0e-02)\textsuperscript{11}& 8.70e-01 (1.9e-02)\textsuperscript{12}& 4.62e+00 (1.2e-01)\textsuperscript{12}\tabularnewline
&\textcite{murrayFastFlexibleMethods2016}: MonoPoly& 7.07e-01 (1.2e-02)\textsuperscript{14}& 1.02e+00 (2.1e-02)\textsuperscript{14}& 4.97e+00 (1.4e-01)\textsuperscript{13}\tabularnewline
&\textcite{cannonMonmlpMultilayerPerceptron2017}: MONMLP& \textbf{2.54e-02} (5.5e-04)\textsuperscript{1}& \textbf{3.37e-02} (6.8e-04)\textsuperscript{1}& \textbf{1.36e-01} (5.1e-03)\textsuperscript{1}\tabularnewline
&\textcite{navarro-garciaConstrainedSmoothingOutofrange2023}: cpsplines& 4.16e-02 (5.9e-04)\textsuperscript{4}& 5.36e-02 (7.3e-04)\textsuperscript{4}& 1.89e-01 (6.3e-03)\textsuperscript{4}\tabularnewline
&\textcite{groeneboomConfidenceIntervalsMonotone2023}: SLSE& 2.98e-01 (8.6e-03)\textsuperscript{10}& 9.72e-01 (3.1e-02)\textsuperscript{13}& 6.21e+00 (1.9e-01)\textsuperscript{14}\tabularnewline
\midrule
\multirow{14}{*}{1.0}&Cubic Spline (CS)& 2.69e-01 (7.4e-03)\textsuperscript{5}& 3.47e-01 (9.4e-03)\textsuperscript{3}& 1.09e+00 (5.3e-02)\textsuperscript{3}\tabularnewline
&Monotone CS (MCS)& 2.42e-01 (5.4e-03)\textsuperscript{2}& 3.12e-01 (7.1e-03)\textsuperscript{2}& \textbf{1.01e+00} (5.1e-02)\textsuperscript{2}\tabularnewline
&Smoothing Spline (SS)& 2.88e-01 (5.6e-03)\textsuperscript{6}& 3.79e-01 (7.2e-03)\textsuperscript{6}& 1.33e+00 (5.6e-02)\textsuperscript{6}\tabularnewline
&Montone SS (MSS)& 2.64e-01 (4.9e-03)\textsuperscript{4}& 3.49e-01 (6.4e-03)\textsuperscript{5}& 1.28e+00 (5.7e-02)\textsuperscript{5}\tabularnewline
&Quadratic Spline (QS)& 5.10e-01 (1.6e-02)\textsuperscript{10}& 7.51e-01 (2.3e-02)\textsuperscript{9}& 3.49e+00 (1.5e-01)\textsuperscript{9}\tabularnewline
&\textcite{heMonotoneBsplineSmoothing1998}: MQS& 2.95e-01 (6.3e-03)\textsuperscript{7}& 4.08e-01 (9.5e-03)\textsuperscript{7}& 1.62e+00 (9.9e-02)\textsuperscript{8}\tabularnewline
&LOESS& 5.35e-01 (1.1e-02)\textsuperscript{13}& 9.45e-01 (1.7e-02)\textsuperscript{11}& 4.88e+00 (1.3e-01)\textsuperscript{11}\tabularnewline
&Isotonic& 3.49e-01 (4.6e-03)\textsuperscript{8}& 4.63e-01 (6.1e-03)\textsuperscript{8}& 1.59e+00 (4.4e-02)\textsuperscript{7}\tabularnewline
&\textcite{mammenEstimatingSmoothMonotone1991}: SI (LOESS+Isotonic)& 5.35e-01 (1.1e-02)\textsuperscript{11}& 9.45e-01 (1.7e-02)\textsuperscript{10}& 4.88e+00 (1.3e-01)\textsuperscript{12}\tabularnewline
&\textcite{mammenEstimatingSmoothMonotone1991}: IS (Isotonic+LOESS)& 5.35e-01 (1.0e-02)\textsuperscript{12}& 9.47e-01 (1.7e-02)\textsuperscript{12}& 4.87e+00 (1.3e-01)\textsuperscript{10}\tabularnewline
&\textcite{murrayFastFlexibleMethods2016}: MonoPoly& 7.39e-01 (1.2e-02)\textsuperscript{14}& 1.07e+00 (2.0e-02)\textsuperscript{14}& 5.24e+00 (1.4e-01)\textsuperscript{13}\tabularnewline
&\textcite{cannonMonmlpMultilayerPerceptron2017}: MONMLP& \textbf{2.35e-01} (6.9e-03)\textsuperscript{1}& \textbf{3.03e-01} (8.4e-03)\textsuperscript{1}& \textbf{9.77e-01} (4.5e-02)\textsuperscript{1}\tabularnewline
&\textcite{navarro-garciaConstrainedSmoothingOutofrange2023}: cpsplines& 2.63e-01 (4.9e-03)\textsuperscript{3}& 3.48e-01 (6.6e-03)\textsuperscript{4}& 1.27e+00 (5.7e-02)\textsuperscript{4}\tabularnewline
&\textcite{groeneboomConfidenceIntervalsMonotone2023}: SLSE& 4.56e-01 (1.0e-02)\textsuperscript{9}& 1.05e+00 (3.0e-02)\textsuperscript{13}& 6.40e+00 (1.8e-01)\textsuperscript{14}\tabularnewline
\midrule
\multirow{14}{*}{1.5}&Cubic Spline (CS)& 3.78e-01 (9.6e-03)\textsuperscript{5}& 4.94e-01 (1.2e-02)\textsuperscript{5}& 1.52e+00 (6.0e-02)\textsuperscript{3}\tabularnewline
&Monotone CS (MCS)& \textbf{3.26e-01} (7.9e-03)\textsuperscript{1}& \textbf{4.33e-01} (1.0e-02)\textsuperscript{1}& \textbf{1.43e+00} (6.0e-02)\textsuperscript{1}\tabularnewline
&Smoothing Spline (SS)& 4.19e-01 (8.6e-03)\textsuperscript{7}& 5.46e-01 (1.0e-02)\textsuperscript{6}& 1.84e+00 (6.8e-02)\textsuperscript{6}\tabularnewline
&Montone SS (MSS)& 3.63e-01 (7.0e-03)\textsuperscript{4}& 4.83e-01 (9.3e-03)\textsuperscript{4}& 1.78e+00 (7.0e-02)\textsuperscript{5}\tabularnewline
&Quadratic Spline (QS)& 6.45e-01 (1.8e-02)\textsuperscript{13}& 9.27e-01 (2.6e-02)\textsuperscript{9}& 4.33e+00 (1.9e-01)\textsuperscript{9}\tabularnewline
&\textcite{heMonotoneBsplineSmoothing1998}: MQS& 4.07e-01 (9.8e-03)\textsuperscript{6}& 5.73e-01 (1.4e-02)\textsuperscript{7}& 2.33e+00 (1.4e-01)\textsuperscript{8}\tabularnewline
&LOESS& 5.57e-01 (1.1e-02)\textsuperscript{11}& 9.47e-01 (1.6e-02)\textsuperscript{11}& 4.88e+00 (1.4e-01)\textsuperscript{11}\tabularnewline
&Isotonic& 4.74e-01 (6.9e-03)\textsuperscript{8}& 6.31e-01 (9.3e-03)\textsuperscript{8}& 2.24e+00 (5.8e-02)\textsuperscript{7}\tabularnewline
&\textcite{mammenEstimatingSmoothMonotone1991}: SI (LOESS+Isotonic)& 5.55e-01 (1.1e-02)\textsuperscript{10}& 9.46e-01 (1.6e-02)\textsuperscript{10}& 4.88e+00 (1.4e-01)\textsuperscript{12}\tabularnewline
&\textcite{mammenEstimatingSmoothMonotone1991}: IS (Isotonic+LOESS)& 5.60e-01 (1.1e-02)\textsuperscript{12}& 9.53e-01 (1.5e-02)\textsuperscript{12}& 4.85e+00 (1.3e-01)\textsuperscript{10}\tabularnewline
&\textcite{murrayFastFlexibleMethods2016}: MonoPoly& 7.18e-01 (1.2e-02)\textsuperscript{14}& 1.06e+00 (1.8e-02)\textsuperscript{14}& 5.19e+00 (1.4e-01)\textsuperscript{13}\tabularnewline
&\textcite{cannonMonmlpMultilayerPerceptron2017}: MONMLP& 3.55e-01 (1.1e-02)\textsuperscript{2}& 4.56e-01 (1.4e-02)\textsuperscript{2}& \textbf{1.45e+00} (6.5e-02)\textsuperscript{2}\tabularnewline
&\textcite{navarro-garciaConstrainedSmoothingOutofrange2023}: cpsplines& 3.61e-01 (7.2e-03)\textsuperscript{3}& 4.77e-01 (9.4e-03)\textsuperscript{3}& 1.72e+00 (7.0e-02)\textsuperscript{4}\tabularnewline
&\textcite{groeneboomConfidenceIntervalsMonotone2023}: SLSE& 5.19e-01 (1.3e-02)\textsuperscript{9}& 1.03e+00 (3.0e-02)\textsuperscript{13}& 6.03e+00 (2.0e-01)\textsuperscript{14}\tabularnewline
\bottomrule
\end{tabular}
}
\end{table}

\end{proof}

\end{document}